\documentclass[12pt]{article}
\usepackage{fancyhdr}
\usepackage{amsmath}
\usepackage{epsfig}
\usepackage{bbm}
\usepackage{booktabs}
\usepackage{color,colortbl}
\usepackage{xcolor}

\usepackage[utf8]{inputenc}
\usepackage[english]{babel}
\usepackage{enumerate}
\usepackage{multirow}
\definecolor{Gray}{gray}{0.9}
\usepackage{epstopdf}
\usepackage{amsthm}

\usepackage{comment}
\usepackage{longtable}
\usepackage[round]{natbib}
\usepackage{lscape}
\usepackage{setspace}
\usepackage{comment}
\usepackage{epstopdf}
\usepackage{subcaption}
\usepackage[textsize=tiny,textwidth=2in]{todonotes}
\usepackage{adjustbox}
\usepackage{soul}

\makeatletter
\newcommand*\bigcdot{\mathpalette\bigcdot@{.5}}
\newcommand*\bigcdot@[2]{\mathbin{\vcenter{\hbox{\scalebox{#2}{$\m@th#1\bullet$}}}}}

\newcommand\EightPtClose{\@setfontsize\EightPtClose\@viiipt{9}}
\newcommand\TenPtType{\@setfontsize\TenPtType\@xpt\@xiipt}
\def\notesize{\TenPtType}
\def\notesize{\EightPtClose}
\newenvironment{figurenotes}[1][Note]{\begin{minipage}[t]{\linewidth}\notesize{\itshape#1: }}{\end{minipage}}
\makeatother

\usepackage{dlfltxbcodetips}
\usepackage{amsfonts}
\usepackage{amssymb}
\usepackage{amsthm}
\usepackage{graphicx}
\usepackage{graphics}
\usepackage{color}
\usepackage{pstricks}
\usepackage{sgame}
\usepackage{egameps}
\usepackage{amsthm}
\usepackage{tikz}
\usepackage{tikz-3dplot}
\usetikzlibrary{arrows}
\usetikzlibrary{calc}
\usepackage{mathrsfs}
\usetikzlibrary{calc,trees,positioning,arrows,chains,shapes.geometric,%
    decorations.pathreplacing,decorations.pathmorphing,shapes,%
    matrix,shapes.symbols}

\tikzset{
>=stealth',
  punktchain/.style={
    rectangle,
    rounded corners,
    draw=black, very thick,
    text width=10em,
    minimum height=3em,
    text centered,
    on chain},
  line/.style={draw, thick, <-},
  element/.style={
    tape,
    top color=white,
    bottom color=blue!50!black!60!,
    minimum width=8em,
    draw=blue!40!black!90, very thick,
    text width=10em,
    minimum height=3.5em,
    text centered,
    on chain},
  every join/.style={->, thick,shorten >=1pt},
  decoration={brace},
  tuborg/.style={decorate},
  tubnode/.style={midway, right=2pt},
}
\makeatletter
\pgfdeclareshape{datastore}{
  \inheritsavedanchors[from=rectangle]
  \inheritanchorborder[from=rectangle]
  \inheritanchor[from=rectangle]{center}
  \inheritanchor[from=rectangle]{base}
  \inheritanchor[from=rectangle]{north}
  \inheritanchor[from=rectangle]{north east}
  \inheritanchor[from=rectangle]{east}
  \inheritanchor[from=rectangle]{south east}
  \inheritanchor[from=rectangle]{south}
  \inheritanchor[from=rectangle]{south west}
  \inheritanchor[from=rectangle]{west}
  \inheritanchor[from=rectangle]{north west}
  \backgroundpath{
    \southwest \pgf@xa=\pgf@x \pgf@ya=\pgf@y
    \northeast \pgf@xb=\pgf@x \pgf@yb=\pgf@y
    \pgfpathmoveto{\pgfpoint{\pgf@xa}{\pgf@ya}}
    \pgfpathlineto{\pgfpoint{\pgf@xb}{\pgf@ya}}
    \pgfpathmoveto{\pgfpoint{\pgf@xa}{\pgf@yb}}
    \pgfpathlineto{\pgfpoint{\pgf@xb}{\pgf@yb}}
 }
}
\makeatother

\newtheorem{lemma}{Lemma}

\newtheorem{assumption}{Assumption}

\theoremstyle{definition}

\theoremstyle{plain}
\usetikzlibrary{arrows}
\usetikzlibrary{calc}

\newcommand{\mbb}{\mathbb}
\newcommand{\E}{\mbb{E}}
\theoremstyle{plain}

\theoremstyle{definition}

\theoremstyle{remark}

\makeatletter
\newcommand{\distas}[1]{\mathbin{\overset{#1}{\kern\z@\sim}}}%
\newsavebox{\mybox}\newsavebox{\mysim}
\newcommand{\distras}[1]{%
  \savebox{\mybox}{\hbox{\kern3pt$\scriptstyle#1$\kern3pt}}%
  \savebox{\mysim}{\hbox{$\sim$}}%
  \mathbin{\overset{#1}{\kern\z@\resizebox{\wd\mybox}{\ht\mysim}{$\sim$}}}%
}
\makeatother

\usepackage{graphicx}

\usepackage{lineno}
\usepackage{geometry}
\usepackage{threeparttable}
\usepackage{amsfonts}
\usepackage{graphicx}
\usepackage{mathtools}
\usepackage{bbm}

\usepackage[colorlinks=true,
            linkcolor=red,
            urlcolor=blue,
            citecolor=blue]{hyperref}

\defcitealias{DB1B}{U.S., DOT 2012} 
\defcitealias{SIAT}{U.S., DOC 2012} 	
\defcitealias{T100}{U.S., DOT 2012} 
\defcitealias{OAG}{OAG, 2012} 
\begin{document}

\title{Price Discrimination in International Airline Markets\thanks{
We thank the editor Francesca Molinari and three referees for numerous helpful suggestions that sub- stantially improved the paper. We are very grateful to the U.S. Department of Commerce for providing the data
for our analysis. We are thankful to research computing facilities at Boston College, UNC and UVA, with special thanks to Sandeep Sarangi at UNC research computing. Seminar 
(Arizona, BC, Cornell, Iowa, MIT, Penn State, Rice, Richmond Fed, Rochester, TAMU, Toulouse, UC-Irvine, UNC-Chapel Hill, UWO, WUSTL, Vanderbilt, Virginia) 
and conference (SEA 2014, EARIE 2015, IIOC 2016, ASSA 2016, DSE 2019, SITE 2019, QME 2021)
participants provided helpful comments. We thank Jan Brueckner, Ben Eden, Gautam Gowrisankaran, Jacob Gramlich, Qihong Liu, Yao Luo, Brian McManus, Nancy Rose, Nick Rupp, Andrew Sweeting and Peter R. Hansen for insightful suggestions. All remaining errors are our own.}}

\author{Gaurab Aryal\thanks{
Department of Economics, Washington University in St. Louis, aryalg@wustl.edu.}
\\
\and
 Charles Murry\thanks{
Department of Economics, Boston College, charles.murry@bc.edu.}
\\
\and
Jonathan W. Williams\thanks{
Department of Economics, University of North Carolina - Chapel Hill, jonwms@unc.edu.} \\
}

\date{\today}

\maketitle

\begin{abstract}

We develop a model of inter-temporal and intra-temporal price discrimination by monopoly airlines to study the ability of different discriminatory pricing mechanisms to increase efficiency and the associated distributional implications.
To estimate the model, we use unique data from international airline markets with flight-level variation in prices across time, cabins, and markets and information on passengers' reasons for travel and time of purchase.
The current pricing practice yields approximately 77\% of the first-best welfare. The source of this inefficiency arises primarily
from private information about passenger
valuations, not dynamic uncertainty about demand. 
We also find that if airlines could discriminate between business and leisure passengers, total welfare would improve at the expense of business passenger surplus.
Also, replacing the current pricing that involves screening passengers across cabin classes with offering a single cabin class has minimal effect on total welfare.
\end{abstract}

\newpage

\section{Introduction}
Firms with market power often use discriminatory prices to increase their profits. 
However, such price discrimination has ambiguous implications for total welfare. 
Enhanced price discrimination may increase welfare by reducing allocative inefficiencies, but it may also reduce consumer welfare. 
So an essential aspect of economic- and public-policy towards price discrimination is to understand how well various discriminatory prices perform in terms of the total welfare and its distribution, relative to each other and the first-best  \citep[e.g.,][]{Pigou1920, Varian1985, WH2015}.

We evaluate the welfare consequences of price discrimination and quantify sources of inefficiencies in a large and economically important setting, international air travel markets. 
{The pricing problem faced by airlines is complex. They lack perfect information about passengers' willingness-to-pay and the time passengers arrive at the marketplace. Measuring the welfare consequences of this missing information requires carefully studying a model of airline pricing.} 
To that end, we develop and estimate a model of inter-temporal and intra-temporal price discrimination by a monopoly airline and study {how} different discriminatory mechanisms {affect} welfare and the associated distributional implications.
The model incorporates a rich specification of passenger valuations for two vertically differentiated seat classes on international flights and a capacity-constrained airline that faces stochastic and time-varying demand. 
The airline screens passengers between the two cabins while updating prices and seat offerings over time.
Using the model estimates, we implement various counterfactuals in the spirit of \cite{BergemannBrooksMorris2015}, where we change the airline's information about current and future preferences and measure the welfare under various discriminatory pricing strategies. Our counterfactual pricing strategies are motivated by current airline practices intended to raise profits by reducing allocative inefficiencies, including attempts to solicit passengers' reason to travel and use of auctions \citep{Nicas2013, Vora2014, Tully2015, McCartney2016}.

We find that the current pricing practice yields approximately 77\% of the first-best welfare; most (87\%) of this inefficiency is due to private information about travelers' valuations, and the rest (13\%) is due to uncertainty about demand. These results suggest that airlines' attempts to collect passenger information for discriminatory pricing could improve efficiency. We also find that relative to the current pricing, the ability to screen passengers based on their reason to travel improves the total welfare at the expense of business passengers and in favor of airline and leisure passengers. However, reducing the scope of second-degree price discrimination by requiring airlines to choose the same price for the economy and first-class seats has a negligible effect on total welfare, suggesting a strong cross-cabin substitution.

Our empirical strategy uses a novel dataset of international air travel from the U.S. Department of Commerce's Survey of International Air Travelers \citepalias{SIAT}. Compared to data used in the extant literature, in our data, we observe the date of transactions, ticket prices, and passenger characteristics for dozens of airlines in hundreds of markets. One such characteristic is the passenger's reason for travel--business or leisure--which enables us to study third-degree price discrimination based on reasons for travel. Furthermore, several nonstop markets in international travel are highly concentrated, allowing us to focus on monopoly pricing. 

Airlines can segment customers in various ways, enabling them to price discriminate. 
We examine ways in which, despite all their advantages, airlines are limited in their ability to price discriminate perfectly. We also document variability and non-monotonicity in prices before departure, suggesting prices respond to demand uncertainty. In particular, we document the late arrival of passengers traveling for business, who tend to have inelastic demand, and the associated price changes. Although business travelers' late arrival puts upward pressure on fares, fares do not increase monotonically for every flight. This pattern suggests that the underlying demand for air travel is stochastic and non-stationary. 

We propose a flexible but tractable demand and supply model to capture these salient data features.
Each period before a flight departs, a random number of nonstop and connecting passengers arrive at the marketplace and purchase either a first-class ticket or an economy-class ticket or decide not to fly. 
Passengers are short-lived, so those who do not buy a ticket do not remain in or ever return to the marketplace.
Our focus is on estimating demand from only the nonstop travelers, whose willingness-to-pay depends on the seat class and their reason to travel. We allow the travelers to have different willingness-to-pay for first-class, such that, for some, the two cabins are close substitutes but not for others. Furthermore, we allow the mix of business and leisure travelers to vary over time.

On the supply side, we model a monopolist airline whose problem is to sell a fixed number of economy and first-class seats in a finite number of periods to maximize total expected profit. The airline knows the distribution of passengers' valuations and the expected number of nonstop and connecting passenger arrivals, and each period it chooses prices and seats to release before observing the demand.\footnote{Airline committing to a seat release policy is meant to capture the ``fare bucket" strategy used by airlines in practice. See, for example, \cite{AlderighiNicoliniPiga2015} for more on these buckets. }
We assume that the airline's choices do not affect the future arrival process, and thus the number of unsold seats is the only endogenous state variable. Thus, every period before the flight, the airline balances the expected profit from selling a seat today against the forgone future expected profit.
This inter-temporal trade-off results in a time-specific endogenous opportunity cost for each seat that varies with the expected future demand and number of unsold seats. 

Besides this temporal consideration, the airline screens passengers between two cabins every period. Thus, our model captures the \textit{inter}-temporal and \textit{intra}-temporal aspects of price discrimination. Moreover, these temporal and intra-temporal considerations are interrelated. For instance, if the airline increases the price for economy seats today, it will trade off fewer economy sales today, higher first-class sales today, with higher economy and lower first-class sales in the future. As demand varies over time, these trade-offs vary too. 

The estimation of our model presents numerous challenges.
The demand and supply specifications result in a non-stationary dynamic programming problem that involves solving a mixed-integer nonlinear program for each state. We solve this problem to determine optimal prices and seat-release policies for every possible combination of unsold seats and days until departure.
Moreover, our data include numerous flights across hundreds of routes, so not only do we allow for heterogeneity in preferences across passengers \emph{within a flight}, but we also allow different flights to have different distributions of passenger preferences.

To estimate the distribution of preferences across flights, we use a simulated method of moments approach based on the importance sampling procedure of \cite{Ackerberg2009}. 
Similar approaches to estimate a random coefficient specification has been used by \cite{FoxKimYang2016}, \cite{NevoTurnerWilliams2016}, and \cite{BlundellGowrisankaranLanger2020}. 
Like them, we match the empirical moments with the model-implied moments that describe within- and across-flight variation in fares and purchases.

Our estimates suggest substantial heterogeneity across passengers within and across flights.
The estimated marginal distributions of the willingness-to-pay for one-way travel for business and leisure travelers are consistent with the observed fares distribution, with the average willingness-to-pay for an economy seat by leisure and business travelers being \$392 and \$537, respectively. Furthermore, on average, passengers value a first-class seat 58\% more than an economy seat, implying meaningful cross-cabin substitution.
The arrival rate of nonstop and connecting passengers declines as we approach the flight date. However, among the nonstop passengers, the share of business travelers increases from zero to almost 30\% for the average flight.
We calculate the time-specific opportunity cost of selling a seat using the model estimates, which provide novel insight into airlines' dynamic incentives.

Using the estimates and the model, we characterize the efficiency level and the associated distribution of surplus under alternative pricing mechanisms that provide new insights into the welfare consequences of price discrimination.
{Our main finding is that the airlines' inability to segment passengers by their willingness-to-pay is the main reason for efficiency loss.} As previously mentioned, the current pricing achieves 77\% of the first-best welfare, and most of this inefficiency is due to the asymmetry of private information about passengers' values. The efficiency loss takes the form of seats allocated to passengers who do not necessarily have the highest valuation. We also find that most inefficiencies occur in the earlier periods because airlines sell too many seats from a total efficiency standpoint early in the sales process due to uncertainties about passenger willingness-to-pay and the future number of arrival of passengers, thus excluding potentially high-surplus late-arriving passengers. 
{Thus, making information about demand and passenger arrival time public would increase total efficiency. 
For instance, if airlines observe both this information, there are mechanisms, e.g., Vickery-Clarke-Grove (VCG) auctions, that could get to 97\% efficiency.}

Regarding the surplus distribution between airlines and passengers, we find that price discrimination skews the distribution of surplus in favor of the airlines.
In particular, the gap between producer and consumer surplus increases by approximately 25\% when airlines price differently across cabins, compared to setting only one price per period for both cabins.
If an airline were to price discriminate even more based on reasons to travel, it would increase the airline's share of surpluses. 
We illustrate the effect of airlines knowing the reason to travel on welfare by determining surplus division when the airline does not know the reason to travel and uses a VCG auction, period-by-period. We find that using VCG auctions doubles the consumer surplus under current pricing practices and is as efficient as eliminating all static informational frictions.

\paragraph{Contribution and Related Literature.}
Our paper relates to extensive research on the economics of price discrimination and research in the empirical industrial organization on estimating the efficiency and division of welfare under asymmetric information. Most of these papers, however, focus on either cross-sectional price discrimination \citep[e.g.,][]{IvaldiMartimort1994, Leslie2004,BusseRysman2005, CrawfordShum2006, mortimer2007price, McManus2007, AryalGabrielli2019} or inter-temporal price discrimination dynamics \citep[e.g.,][]{NevoWolfram2002, Nair2007,Escobari2012, Jian2012, NevoHendel2013, Lazarev2013, ChoLeeRustYu2018}, but not both. 

We also contribute to the literature that focuses on dynamic pricing \citep[e.g.,][]{GRADDY20116,Sweeting2010, ChoLeeRustYu2018, Waisman2021}. 
However, none studies intra-temporal price discrimination, inter-temporal price discrimination, and dynamic pricing together, even though many industries involve all three. Relatedly, \cite{CoeyLarsenPlatt2020} document inter-temporal and intra-temporal price dispersion in an environment with consumer search and deadlines but without price discrimination. We contribute to this research by developing an empirical framework where both static discriminative pricing and dynamic pricing incentives are present and obtain results that characterize the welfare implications.\footnote{There is a long theoretical literature on static and inter-temporal price discrimination; see \cite{Stokey1979, GaleHolmes1993,dana1999using, CourtyLi2000, Armstrong2006}. Airline pricing has also been studied extensively from the perspective of \emph{revenue management}; see \cite{RyzinTalluri2005}.} Our welfare analysis has some similarities with \cite{DubeMisra2021} and \cite{Williams2020}, who find that price discrimination improves welfare.

Additionally, we complement recent research related specifically to airline pricing, particularly \cite{Lazarev2013}, \cite{LiGranadosNetessine2014}, and \cite{Williams2020}.\footnote{Also see \cite{BorensteinRose1994}, \cite{puller2012does, ChandraLederman2018}, who use a regression approach to study effects of oligopoly on dispersion in airfares.} \cite{Lazarev2013} considers a model of inter-temporal price discrimination with one service cabin and finds large potential gains from allowing reallocation among passengers arriving at different times before the flight. Like us \cite{Williams2020} further allows for dynamic adjustment of prices in response to stochastic demand and finds that dynamic pricing (relative to a single price) increases total welfare at the expense of relatively inelastic customers who arrive late.
Our paper, however, differs from the existing literature on price discrimination in several ways. 
First, we model inter-temporal and intra-temporal price discrimination, closely capturing a multiproduct airline's decisions to screen passengers between cabin classes. Our results suggest that these two aspects jointly affect welfare. Second, like \cite{Lazarev2013}, we model connecting passengers, who play an important role in dynamic adjustments to remaining seats. Third, we model the airlines' ability to choose the number of seats available to passengers at any given time before the flight, which mimics the complicated selling algorithms that airlines use in practice. Lastly, we observe passengers' travel reasons, so we can measure business and leisure passenger welfare separately. 

Most papers use a random utility discrete choice framework to model the demand for air travel  \citep[e.g.,][]{BerryCarnallSpiller2006}. Instead, we use a pure characteristics approach, following the theoretical literature on price discrimination \citep[e.g.,][]{MussaRosen1978}. See also \cite{BerryPakes2007} and \cite{Molinari2021AER} for empirical examples. A pure characteristic approach is natural when modeling vertically differentiated products--economy and premium-class cabins. We estimate a flexible random coefficients demand model that captures within-flight and across-flight heterogeneity in consumer preferences. This approach is important as we have data for many types of flights. The richness in our specification allows us to capture variation in demand parameters due to permanent and unobserved differences across markets. However, the flexibility of our random-effects approach comes at the cost of not modeling flight-specific factors that affect demand.


\section{Data \label{Data}}

The Department of Commerce's Survey of International Air Travelers (SIAT) \citepalias{SIAT} gathers information on international air passengers traveling to and from the U.S.
Passengers are asked detailed questions about their flight itinerary. These surveys take place either during the flight or at the gate area before the flight.
The SIAT targets U.S. residents traveling abroad and non-residents visiting the U.S. Passengers in our sample are randomly chosen flights from more than 70 participating U.S. and international airlines, including some charter carriers. The survey contains ticket information, which includes the cabin class (first, business, or economy), date of purchase, total fare, and the trip's purpose (business or leisure). We combine fares reported as business class and first-class into a single cabin class that we label ``first-class.'' This richness distinguishes the SIAT data from other data like the  Origin and Destination Survey (DB1B) conducted by the Department of Transportation \citepalias{DB1B}. In particular, the additional details about passengers (e.g., time of purchase, individual ticket fares, and reason for travel) make the SIAT dataset ideal for studying price discrimination.

We create a dataset from the survey where a unit of observation is a single ticket purchased by a passenger flying a nonstop (or direct) route. We then use fares and purchase (calendar) dates associated with these tickets to estimate price paths for each flight in our data, where a flight is a single instance of a plane serving a particular route. For example, in our sample, we observed some nonstop passengers flying United Airlines from SEA to TPE on August 12, 2010, departing at 5:10 pm. Then we say that this is one flight. 
In this section, we detail how we selected the sample and display descriptive statistics that motivate our model and analysis.

\subsection{Sample Selection}
Our sample from the DOC includes 413,309 passenger responses for 2009-2011. We clean the data to remove contaminated and missing observations and construct a sample of flights that will inform our airline pricing model, which we specify in the following section. We detail our sample selection procedure in Appendix \ref{siat}, but, for example, we exclude responses that do not report a fare, are part of a group travel package, or are non-revenue tickets. We supplement our data with schedule data from the Official Aviation Guide of the Airways (OAG) company \citepalias{OAG}, which reports cabin-specific capacities by flight number.
Using the flight date and flight number in SIAT, we can merge the two data sets.
We include flights for which we observe at least ten nonstop tickets after applying the sample selection criteria.

\paragraph{Monopoly Markets.}
As mentioned earlier,  we focus on monopoly markets.
In international air travel, nonstop markets tend to be concentrated, except for a few busiest airport pairs. We classify a market as a monopoly market if it satisfies one of the following two criteria: (i) one airline flies at least 95\% of the total capacity on the route (where the capacity is measured using the OAG data); or (ii) a U.S. carrier and foreign carrier operate on the market with antitrust immunity from the U.S. Department of Justice. These immunities do not have any additional regulatory oversight on airfares.

These antitrust exemptions come from market access treaties signed between the U.S. and the foreign country that specify a local foreign carrier (usually an alliance partner of the U.S. airline) that will share the route. For example, on July 20, 2010, antitrust exemption was granted to the OneWorld alliance, which includes American Airlines, British Airways, Iberia, Finnair, and Royal Jordanian, for ten years subject to a slot remedy.\footnote{To determine such markets, we use information from DOT and carriers' 10K reports filed with the SEC.}

We define markets at the city-pair level in a few cases because we are concerned that within-city airports are substitutable. The airports that we aggregate up to a city-pairs market definition include airports in the New York, London, and Tokyo metropolitan areas. Thus, we treat a flight from New York JFK to London Heathrow to be in the same market as a flight from Newark EWR to London Gatwick.

\begin{table}[ht!]
\caption{Top 20 Markets, Ordered by Sample Representation}\label{tab:markets}
\hspace{-0.6in}
\scalebox{0.7}{
\begin{tabular}{lccc p{.5cm} |lccc}
\toprule
      &  {\bf Unique}  &      &           {\bf Distance}&&        &  {\bf Unique}  & & {\bf Distance}    \\
 {\bf Market} &  {\bf Flights} &  {\bf Obs.} &  {\bf (miles)} &&  {\bf Market} &  {\bf Flights} &  {\bf Obs.}  &  {\bf (miles)}\\
\midrule
Los Angeles (LAX)-Shanghai (PVG)   &   53  & 1,646 &  6,485  && New York (JFK)-Vargas (CCS)   &   20  &  476  &  2,109 \\
San Francisco (SFO)-Auckland (AKL)   &   28  &  1,623  &  6,516  && Boston (BOS)-Keflavik (KEF)   &   13   & 427  &  2,413 \\
New York (JFK)-Helsinki (HEL)   &   35  &  1,191  &  4,117  && New York (JFK)-Santiago (SCL)   &   16   & 423  &  5,096 \\
New York (JFK)-Johannesburg (JNB)   &   26  &  1,035  &  7,967  && New York (JFK)-Nice (NCE)   &    9   & 329  &  3,991 \\
New York (JFK)-Warsaw (WAW)   &   34  &  787  &  4,267  && New York (JFK)- Dubai (DXB)   &    9   & 325  &  6,849 \\
New York (JFK)-Vienna (VIE)   &   29  &  729 &  4,239  && Pheonix (PHX)-Puerto Vallarta (PVR)   &    6   & 289  &  971 \\
Phoenix (PHX)-Cabo (SJD)   &   18   & 703  &    721 && Orlando (MCO)-Frankfurt (FRA)   &   8  & 284  &  4,734 \\
New York (JFK)-Casablanca (CMN)   &   21  &  625  &  3,609  && San Francisco (SFO)-Paris (CDG)   &   10   & 251  &  5,583\\
Seattle (SEA)-Taipei (TPE)   &   21  &  328  &  6,074 &&San Francisco (SFO)- Tokyo (HND)   &   4   & 251  &  5,161\\
New York (JFK)-Buenos Aires (EZE)   &   24  &  547  &  5,281  && Orlando (SFB)-Birmingham (BHX)   &   10   & 245  &  4,249 \\
\bottomrule
\end{tabular}}
\begin{figurenotes}
The table displays the Top 20 markets, their representation in our sample, their distance (air miles) from U.S. Data from the Survey of International Air Travelers, and the sample described in the text.
\end{figurenotes}
\end{table}

\paragraph{Description of Markets and Carriers.}
After our initial selection and restriction on nonstop monopoly markets, we have 14,930 observations representing 224 markets and 64 carriers. We list the top 20 markets based on sample representation in Table \ref{tab:markets}, along with the number of unique flights and the total observations. The most common U.S. airports in our final sample are New York JFK (JFK), San Francisco (SFO), and Phoenix (PHX), and the three most common routes are Los Angeles to Shanghai, China (PVG), San Francisco to Auckland, New Zealand (AKL),  and New York to Helsinki, Finland (HEL). The median flight has a distance of 4,229 miles and an inter-quartile range of $2,287$ to $5,583$ miles. In Table \ref{tab:carriers}, we display the top ten carriers from our final sample, representing approximately 60\% of our final observations.

\begin{table}[ht!]
\centering
\caption{Top Ten Carriers, Ordered by Sample Representation}\label{tab:carriers}

\begin{tabular}{lcc p{2cm} lcc}
\toprule
{\bf Carrier} & {\bf Unique Flights} & {\bf Obs.} \\
\midrule
US Airways         &  77  & 3,179  \\
Continental        &  98  & 2,036  \\
American Airlines  &  88  & 1,926  \\
Air New Zealand    &  30  &  1,666  \\
China Eastern      &  53  & 1,646  \\
Finnair            &  35  &  1,191  \\
Lufthansa          &  40  &  1,079  \\
Delta              &  35  &  1,062  \\
South African  & 26& 1,035\\
Austrian           &  35  &  881  \\
\bottomrule
\end{tabular}
\begin{figurenotes}
The table displays the top ten carriers ordered by frequency in the final sample. Data are from the Survey of International Air Travelers (SIAT) described in the text.
\end{figurenotes}

\end{table}

\subsection{Passenger Arrivals and Ticket Sale Process\label{descriptives}}
Passengers differ in their time of purchase and reasons for travel, and prices vary over time and across cabins. In this subsection, we present key features in our data about passengers and prices.
Throughout the paper, we use the reason for passengers' travel to define their type, which can be either business or leisure. Regardless of their type, passengers can choose between an economy-class or a first-class ticket, or not fly at all.

\paragraph{Timing of Purchase.}  Airlines typically start selling tickets a year before the flight date.
Although passengers can buy their tickets throughout the year, in our sample, most passengers buy in the last 120 days.
To keep the model and estimation tractable, we classify the purchase day into a fixed number of bins.
At least two factors motivate our choice of bin sizes. First, it appears that airlines typically adjust fares frequently in the last few weeks before the flight date, but less often farther away from the flight date \citep[see][]{HortacsuNatanParsleySchweigWillaims2021}. Second, there is usually a spike in passengers buying tickets at focal points, like 30 days, 60 days, and so on. 

Although we focus on pricing and seat release decisions for nonstop travel, flights also have connecting passengers, which we account for in our model. For example, in the LAX-PVG market, there can be nonstop passengers whose trips originate in LAX and terminate in PVG, but there may also be those flying from SFO (or other originating airports) to PVG via LAX. 
In Table \ref{tab:advance}, we present eight fixed bins for nonstop and connecting passengers and the number of observations in each bin. We use narrower bins closer to the flight date to reflect airlines' dynamic pricing strategies. Each of these eight bins corresponds to one period in our model, giving us eight periods. 

\begin{table}[ht!]
\centering
\caption{Distribution of Advance Purchase}\label{tab:advance}

\begin{adjustbox}{scale=1}
\begin{tabular}{lcc}
\toprule
{\bf Days Until Flight} & {\bf Nonstop Obs.}& { \bf Connecting Obs.}\\
\midrule
   $0-3$ days &        597 &        {437}  \\
   $4-7$ days &        753 &        {528}  \\
  $8-14$ days &      1,038 &       {839}  \\
$ 15-29$ days &      1,555&        {1,126}   \\
$ 30-44$ days &      2,485&        {1,855}   \\
$ 45-60$ days &      2,561&        {1,936}   \\
$61-100$ days &      2,438&        {1,800}   \\
$  101+$ days &      3,400&        {2,332}   \\
\bottomrule
\end{tabular}
\end{adjustbox}
\begin{figurenotes}
The table displays the distribution of advance purchases. The first column is the number of days before the flight, and the second and third columns show how many nonstop and connecting passengers bought their tickets on those days, respectively.
\end{figurenotes}

\end{table}
{Although we observe prices, ticket class, and reason-to-travel for connecting passengers, from now onwards we only present information about nonstop passengers unless otherwise stated. All the figures and tables shown here use the sample of nonstop passengers.}

\paragraph{Passenger Characteristics.} We classify each passenger as either a business traveler or a leisure traveler based on the reason to travel. \textit{Business} includes business, conference, and government or military, while \textit{leisure} includes visiting family, vacation, religious purposes, study or teaching, health, and other; see Table \ref{table:summarystat} for more.\footnote{We also observe the channel through which the tickets were purchased: travel agents (36.47\%), personal computer (39.76\%), airlines (13.36\%), company travel department (3.99\%), and others  (6.42\%).} We classify service cabins into economy class and first-class, where we {combine business and first class} as the latter.\footnote{ For 90\% of the flights, passengers are seated in only one of these two cabins, but not both.}

In Table \ref{sumstatTKT} we display some key statistics for relevant ticket characteristics in our {nonstop} sample.
 As is common in the literature, to make one-way and round-trip fares comparable, we divide round-trip fares by two. Approximately 4.5\% passengers report to have bought a one-way ticket; see Table \ref{table:summarystat}. The standard deviations in fares are large, with a coefficient of variation of 0.85 for the economy and 1.06 for first-class.

\begin{table}[htbp]
\caption{Summary Statistics from SIAT, Ticket Characteristics of Nonstop Passengers} \label{sumstatTKT}
\centering
\begin{adjustbox}{scale=1}
\begin{threeparttable}
\begin{tabular}{p{5cm}ccc}
\toprule
 &\multicolumn{1}{c}{{Proportion}}&\multicolumn{2}{c}{{Fare}}   \\
\multicolumn{1}{l}{\bf Ticket Class }&\multicolumn{1}{c}{of Sample}&\multicolumn{1}{c}{Mean}&\multicolumn{1}{c}{S.D.} \\
\hline
\hspace{.25cm} Economy              &  92.50  & 447   &    382\\
\hspace{.25cm} First                &  7.50   & 897   &    956\\
\hline
\multicolumn{1}{l}{{ \bf Advance Purchase }}&\multicolumn{1}{c}{{}}&\multicolumn{1}{c}{{}}&\multicolumn{1}{c}{{}}\\
\hline
\hspace{.25cm} 0-3 Days          &   4.03  &    617  &     636\\
\hspace{.25cm} 4-7 Days          &   5.08  &    632  &     679\\
\hspace{.25cm} 8-14 Days         &   7.00  &    571  &     599\\
\hspace{.25cm} 15-29 Days        &  10.49  &    553  &     567\\
\hspace{.25cm} 30-44 Days        &  16.76  &    478  &     429\\
\hspace{.25cm} 45-60 Days        &  17.27  &    467  &     432\\
\hspace{.25cm} 61-100 Days       &  16.44  &    414  &     315\\
\hspace{.25cm} 101+ Days         &  22.93  &    419  &     387\\[4mm]
\hline
\multicolumn{1}{l}{{ \bf Travel Purpose }}&\multicolumn{1}{c}{{}}&\multicolumn{1}{c}{{}}&\multicolumn{1}{c}{{}}\\
\hline
\hspace{.25cm} Leisure     &  85.57 & 446    &   400  \\
\hspace{.25cm} Business    &  14.43 & 684    &   716\\[2mm]
\bottomrule
\end{tabular}
\begin{figurenotes}
Data from the Survey of International Air Travelers. Sample of nonstop passengers described in the text.
\end{figurenotes}
\end{threeparttable}
\end{adjustbox}

\end{table}

In the second panel of Table \ref{sumstatTKT}, we display the same statistics by the number of days in advance of a flight's departure that the ticket was purchased (aggregated to eight ``periods"). The average fare increases for tickets purchased closer to the departure date, as does the standard deviation.

Similarly, at the bottom panel of Table \ref{sumstatTKT}, we report price statistics by the passenger's trip purpose. About 14\% of the passengers flew for business purposes, and these passengers paid an average price of \$684 for one direction of their itinerary. Leisure passengers paid an average of \$446. This price difference arises for at least three reasons: business travelers tend to buy their tickets much closer to the flight date, prefer first-class seats, and fly to different markets.

\begin{figure}[htbp]
    \centering
        \caption{\em Business versus Leisure Passengers before the Flight Date \label{fig:timetrend}}
    \begin{subfigure}[b]{0.44\textwidth}
        \includegraphics[width=\textwidth]{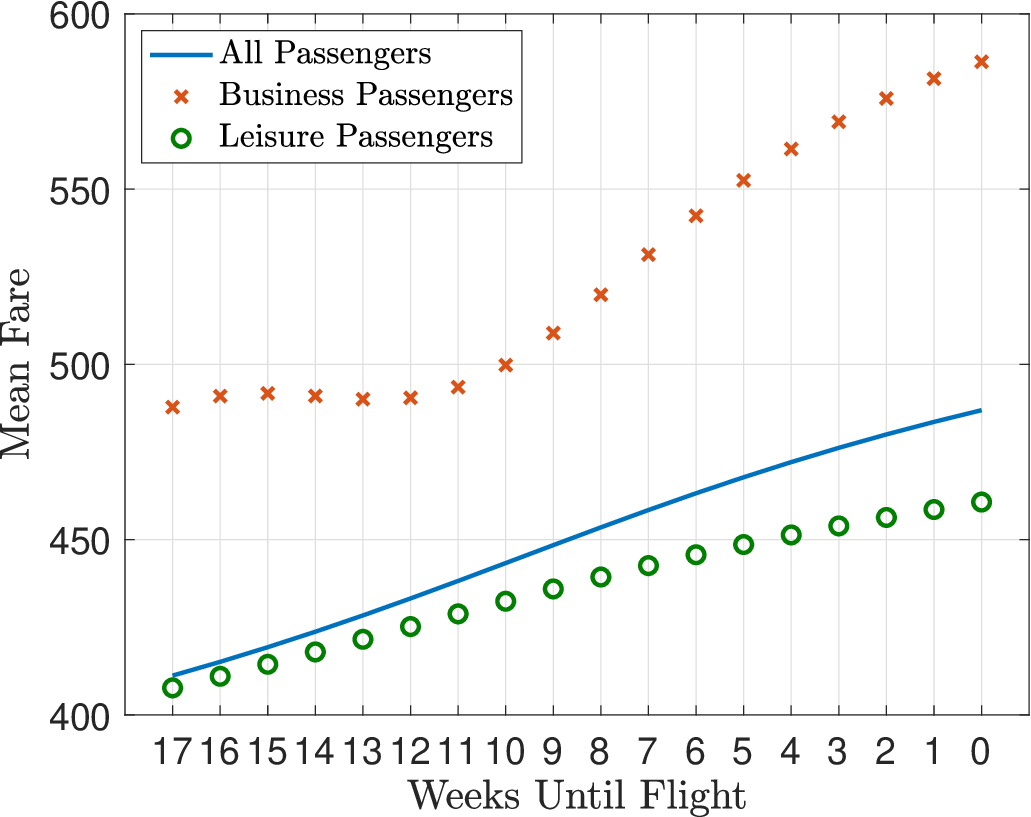}
        \caption{Average Economy Fares Prior to Flight}
    \end{subfigure}
    \begin{subfigure}[b]{0.49\textwidth}
        \includegraphics[width=\textwidth]{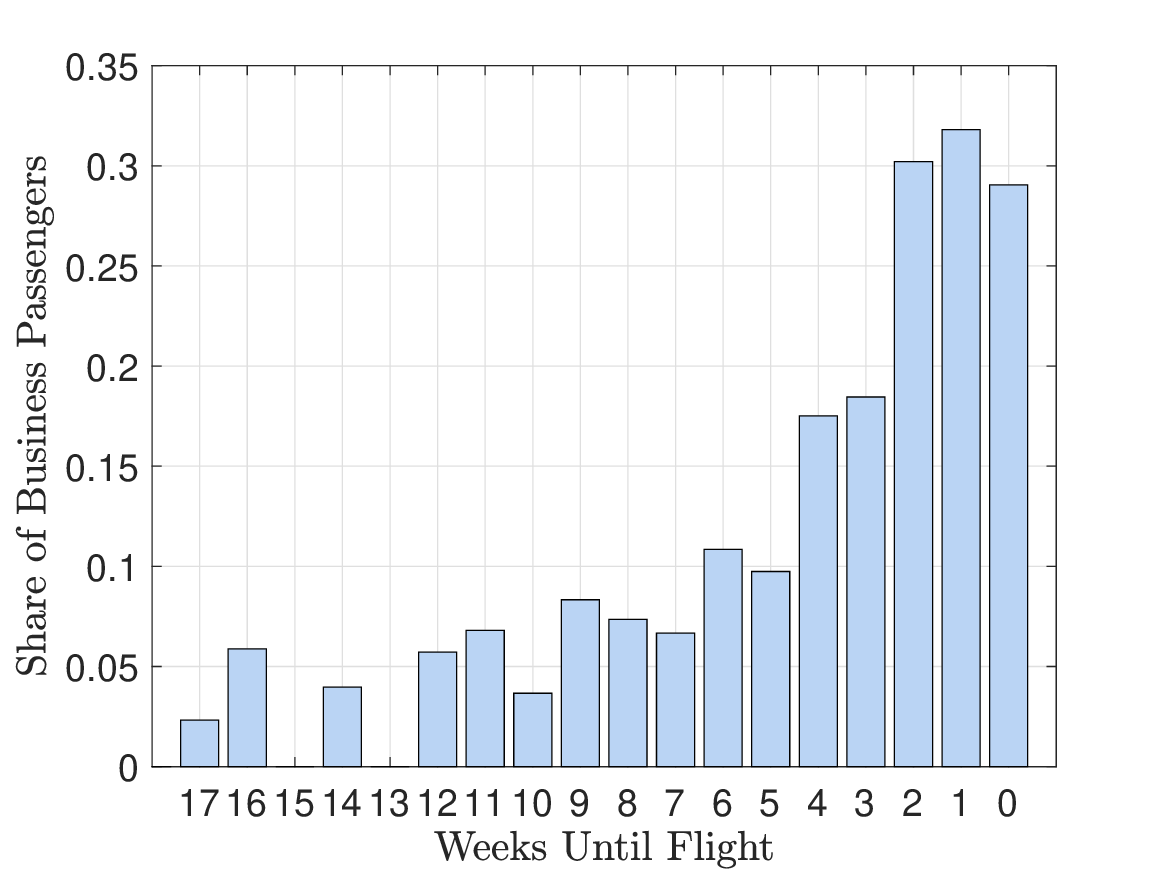}
        \caption{Business and Leisure Travelers}
    \end{subfigure}
\begin{figurenotes}
(a) Average price paths across all flights for tickets in economy class by the week of purchase prior to the flight date, by self-reported business and leisure travelers. Individual transaction prices are smoothed using the nearest neighbor with a Gaussian kernel with optimal bandwidth of 0.5198. (b) The proportion of business passengers across all flights, by advance purchase weeks.
\end{figurenotes}

\end{figure}

In Figure \ref{fig:timetrend}(a), we plot the average price for economy fares as a function of when the ticket was purchased.
Both business and leisure travelers pay more if they buy the ticket closer to the flight date, but the increase is more substantial for business travelers. The solid line in Figure \ref{fig:timetrend}(a) reflects the average price across both reasons for travel. At earlier dates, the total average price is closer to the average price paid by leisure travelers, while it gets closer to the average price paid by business travelers as the flight date nears. 
These are averages across all flights in our data, and we see that, on average, business passengers take more expensive flights than leisure passengers, motivating our modeling of flight heterogeneity in our analysis below.
In Figure \ref{fig:timetrend}(b), we display the proportion of business to leisure travelers across all flights by the advance purchase categories. In the last two months before the flight, the share of passengers traveling for leisure is approximately 90\%. This share decreases to 65\% a week before the flight. Taken together, business travelers purchase closer to the flight date than leisure travelers, and markets with a greater proportion of business travelers have a steeper price gradient.

\begin{figure}[htbp]
\begin{center}
  \caption{\em Histogram of Percent of Nonstop Business Passengers by Flight} \label{fig:mktbus}
  \includegraphics[scale=0.4]{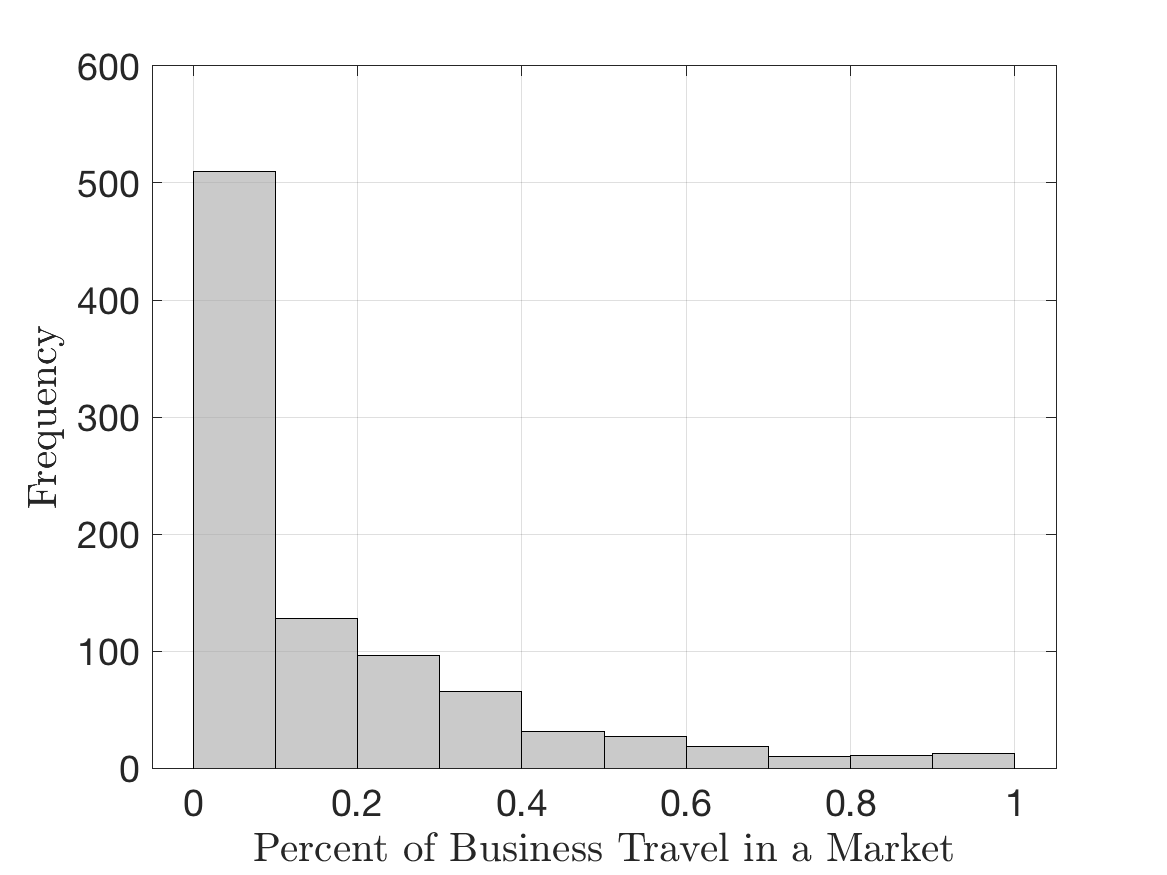}
\end{center}
\begin{figurenotes}
 Histogram of the business-travel index (BTI). The business-traveler index is the flight-specific ratio of self-reported business travelers to leisure travelers. The mean is 0.154, and the standard deviation is 0.210.
 \end{figurenotes}
\end{figure}

Observing the purpose of travel plays an important role in our empirical analysis, reflecting substantial differences in the behavior and preferences of business and leisure passengers. This passenger heterogeneity across markets drives variation in pricing, and this covariation permits us to estimate a model with richer consumer heterogeneity than the existing literature like \cite{BerryCarnallSpiller2006}, and \cite{CilibertoWilliams2014}. Further, a clean taxonomy of passenger types allows a straightforward exploration of the role of asymmetric information in determining inefficiencies and the distribution of surplus that arises from discriminatory pricing of different forms.\footnote{Passengers can also use reward points to buy their tickets, but we do not consider them here. However, reward points are not common in our data. Only 3.5\% of tickets are bought using reward points, and 5.46\% of passengers report that the main reason they choose their airline is frequent-flyer miles. In our final (estimation) sample, 1.7\% passengers get upgraded, i.e., buy economy but fly first-class. Our working hypothesis is that those upgrades materialized only on the day of travel.}

To further explore the influence this source of observable passenger heterogeneity has on fares, we present statistics on across-market variation in the dynamics of fares. Specifically, we first calculate the proportion of business travelers in each market, i.e., across all flights with the same origin and destination. Like \cite{Borenstein2010}, we call this market-specific ratio the business-traveler index (BTI).  In Figure \ref{fig:mktbus}, we present the histogram of the BTI across markets in our data. If airlines know of this across-market heterogeneity and use it to discriminate intra-temporally (across cabins) and inter-temporally (across time before a flight departs), different within-flight temporal patterns in fares should arise.

\begin{figure}[htbp]
    \centering
    \caption{\em Proportion of Business Travelers by Ticket Class  \label{fig:jon1}}
    \begin{subfigure}[b]{0.48\textwidth}
        \includegraphics[width=\textwidth]{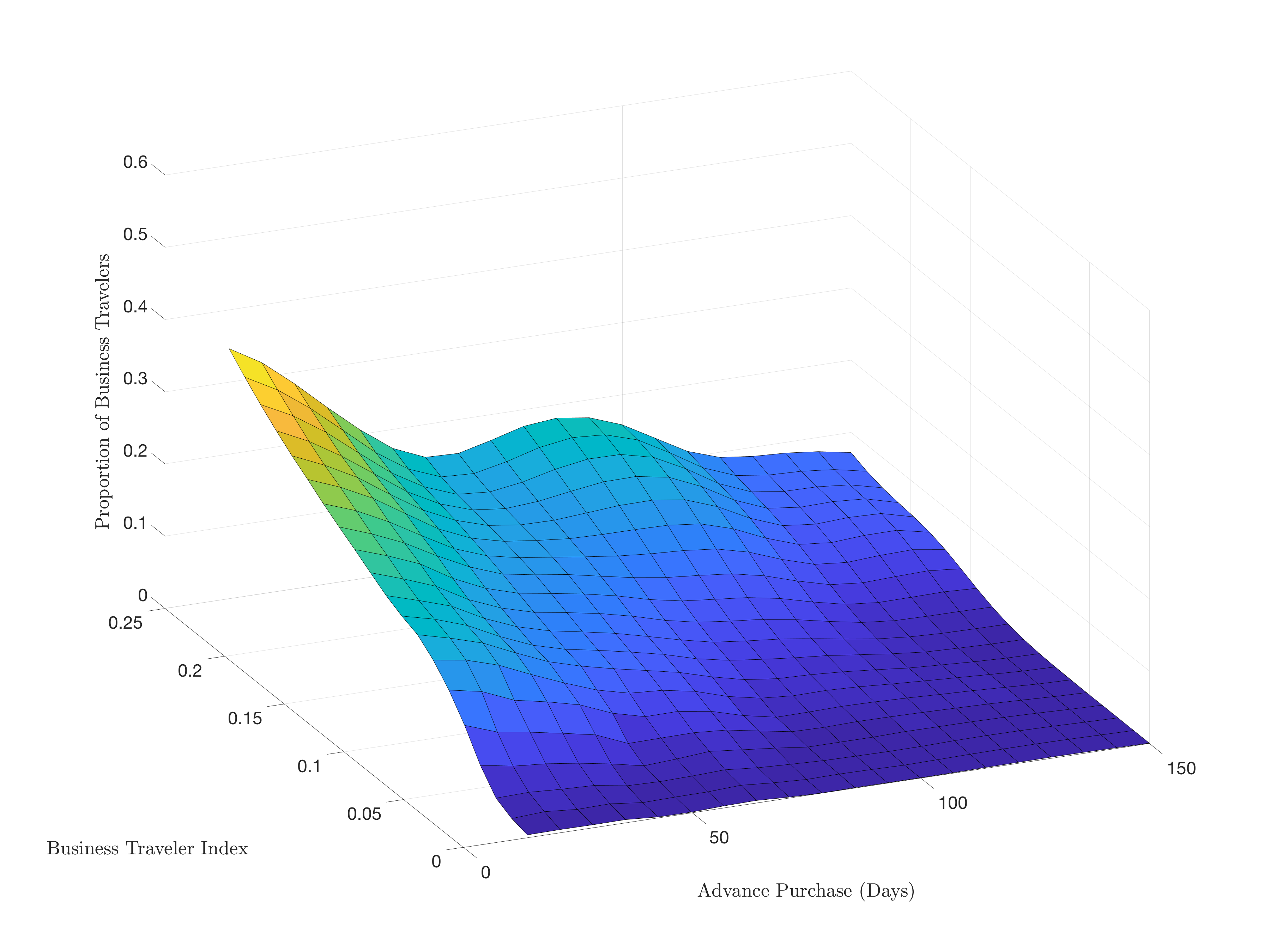}
        \caption{Economy Class}
    \end{subfigure}
    \begin{subfigure}[b]{0.48\textwidth}
        \includegraphics[width=\textwidth]{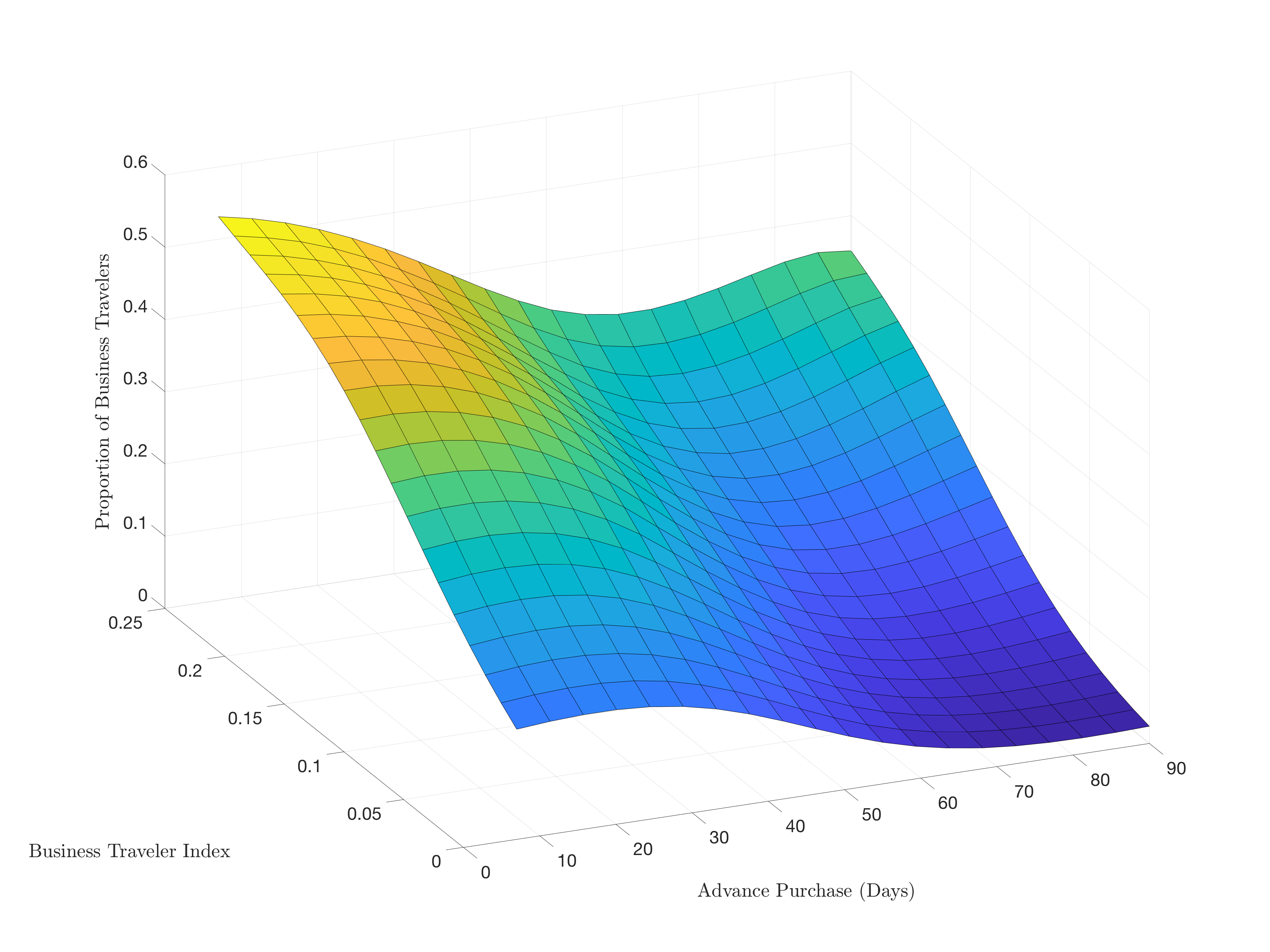}
        \caption{First-Class}
    \end{subfigure}
  \begin{figurenotes}
The figure presents kernel regression of the reason to travel on BTI and the purchase date. Panels (a) and (b) show the regression for the economy and the first-class seats, respectively. The regression uses a Gaussian kernel with optimal ``rule-of-thumb'' bandwidth. \end{figurenotes}
\end{figure}

In Figure \ref{fig:jon1} we present the results of a bivariate kernel regression where we regress an indicator for whether a passenger is traveling for business on the BTI in that market and the number of days the ticket was purchased in advance of the flight's departure. Figures \ref{fig:jon1}(a) and \ref{fig:jon1}(b) present the results for economy and first-class passengers, respectively.  There are two important observations. First, across all values of the BTI, business passengers arrive later than leisure passengers. Second, business passengers disproportionately choose first-class seats. To capture this feature, in Section \ref{model}, we model the difference between business and leisure passengers in terms of the timing of purchases and the preference for quality by allowing the passenger mix to change as the flight date approaches. This modeling approach results in a non-stationary demand process.

\begin{figure}[ht!]
    \centering
        \caption{\em Across-Market Variation in Fares\label{fig:jon2}}
    \begin{subfigure}[b]{0.48\textwidth}
        \includegraphics[width=\textwidth]{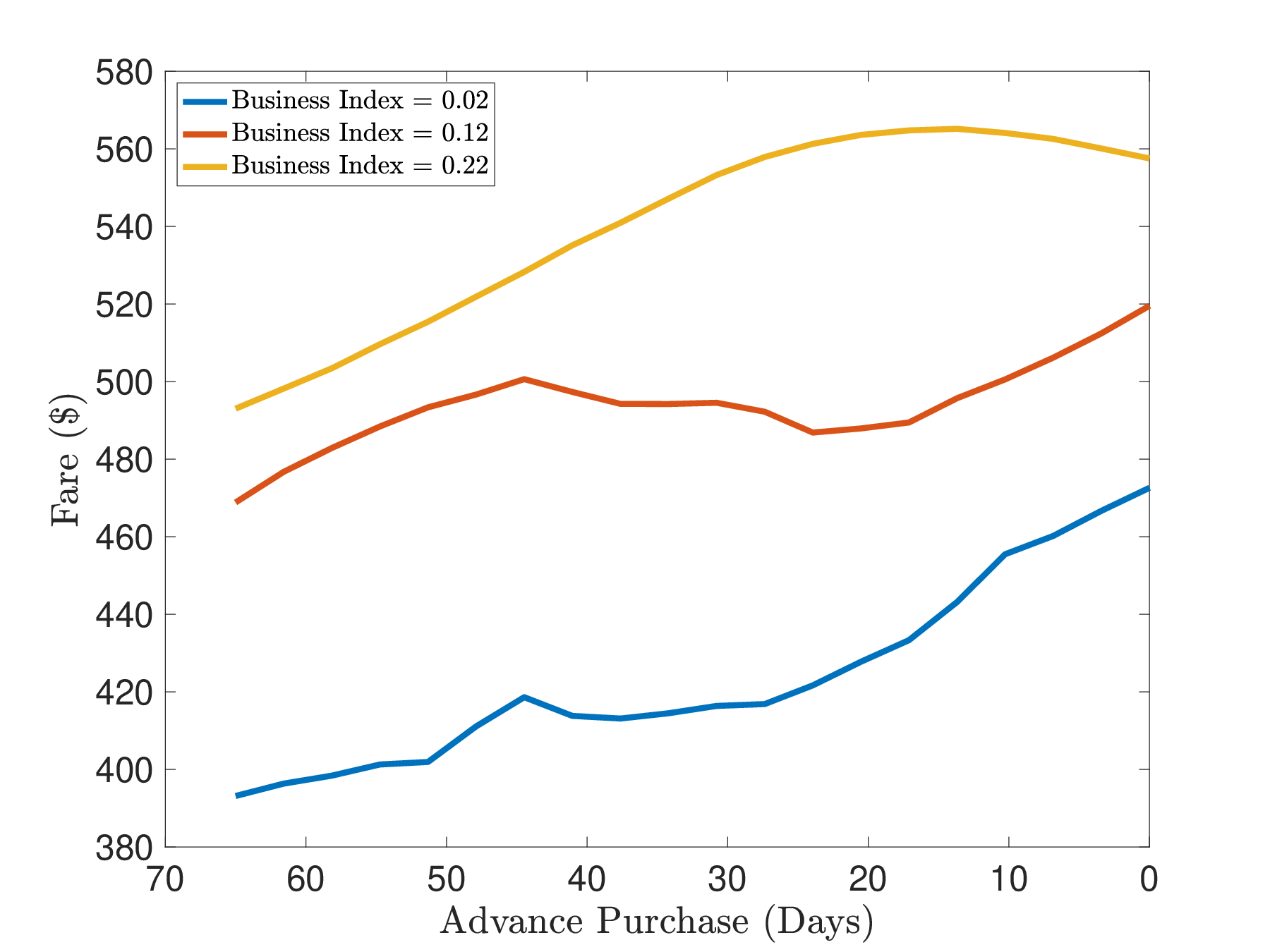}        
        \caption{Economy}
    \end{subfigure}
    \begin{subfigure}[b]{0.48\textwidth}
        \includegraphics[width=\textwidth]{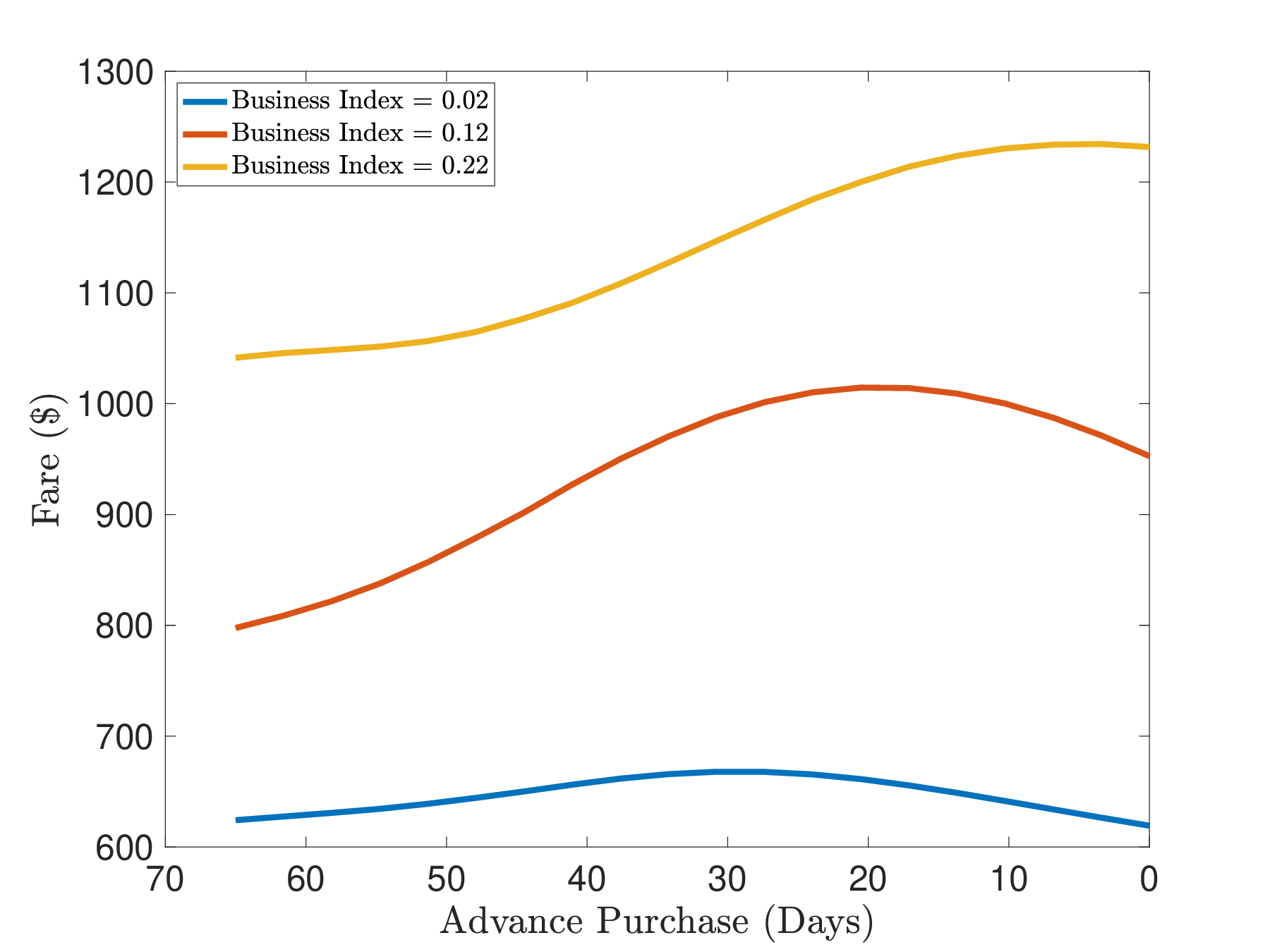}
        \caption{First-Class}
    \end{subfigure}
    \flushleft
   \begin{figurenotes}
 The figure presents results from a kernel regression of fares paid on BTI and the purchase date. Panels (a) and (b) show the regression results evaluated at different BTI values for the economy and the first-class seats, respectively. The regression uses a Gaussian kernel with optimal ``rule-of-thumb'' bandwidth.
\end{figurenotes}
\end{figure}
The influence of business passengers is evident on prices. Like Figure \ref{fig:jon1}, Figure \ref{fig:jon2}(a) and Figure \ref{fig:jon2}(b) present the results of a kernel regression with fare paid as the dependent variable for economy and first-class cabins, respectively. In both, we present cross-sections of these estimated surfaces for the $25^{th}$, $50^{th}$, and $75^{th}$ percentile values of the BTI.  For both cabins, greater values of the BTI are associated with substantially higher fares.

While there are clear patterns in how the dynamics of \textit{average} fares vary with the BTI, there is substantial heterogeneity across flights in how fares change as the flight date approaches. {To see the heterogeneity that Figures \ref{fig:jon1} and \ref{fig:jon2} mask, Figure \ref{fig:price_dists} presents a visualization of the time paths of economy fares for all flights in our data. Specifically, for each flight, we estimate a smooth relationship between economy fares and time before departure using a kernel regression and then normalize the path relative to the initial fare we observe for that flight. Each line is a single flight in our data and is normalized to begin at $1$. This heterogeneity, including some flights with non-monotonic prices, motivates the heterogeneity we include in our analysis.}

\begin{figure}[ht!]
   \centering
       \caption{\em Flight-Level Dispersion in Fares\label{fig:price_dists}}
       \includegraphics[width=.65\textwidth]{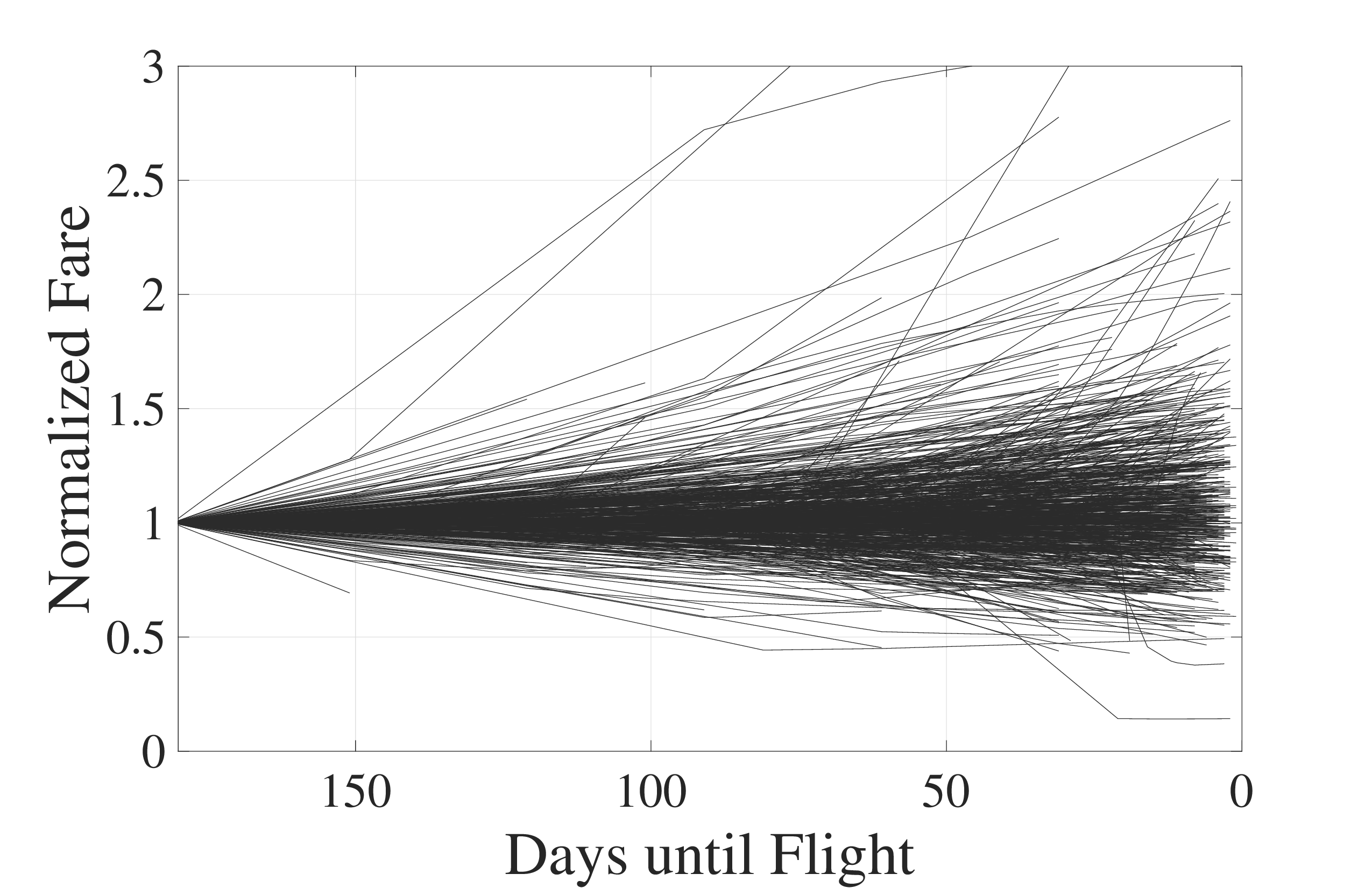}
\begin{figurenotes}
The figure presents the estimated fare paths from every flight we use in our estimation. Each line represents a flight. We normalize the fares to unity by dividing fares by the first fare we observe for the flight. The horizontal axis represents the days until the flight, where 1 is the flight date.
\end{figurenotes}
\end{figure}

For most flights, we observe little movement in fares until approximately 100 days before departure. Nevertheless, for a small proportion of flights, there are substantial decreases and increases in fares as much as five months before departure. Further, by the date of departure, the interquartile range of the ratio of current fare to initial fare is  0.75 to 1.85.
Thus, 25\% of flights experience a decrease of more than 25\%, while 25\% of flights experience an increase of greater than 85\%. The variation in the temporal patterns in fares across flights is attributable to the across-market heterogeneity in the mix of passengers and how airlines respond to demand uncertainty.

\subsection{Aircraft Characteristics}
Airlines' fares and {how they respond} to realized demand depend on the number of unsold seats. In Figure \ref{fig:initialcapacity}(a), we display the joint density of initial capacity of first and economy class in our sample. 
{The smaller of the two peaks corresponds to approximately 160 economy seats and 18 first-class seats, and the higher peak corresponds to approximately 265 economy seats and 47 first-class seats.}
\begin{figure}[ht!]
  \centering
    \caption{\em Initial Capacity and Load Factor\label{fig:initialcapacity}
}
  \begin{subfigure}[b]{0.48\textwidth}
    \includegraphics[width=\textwidth]{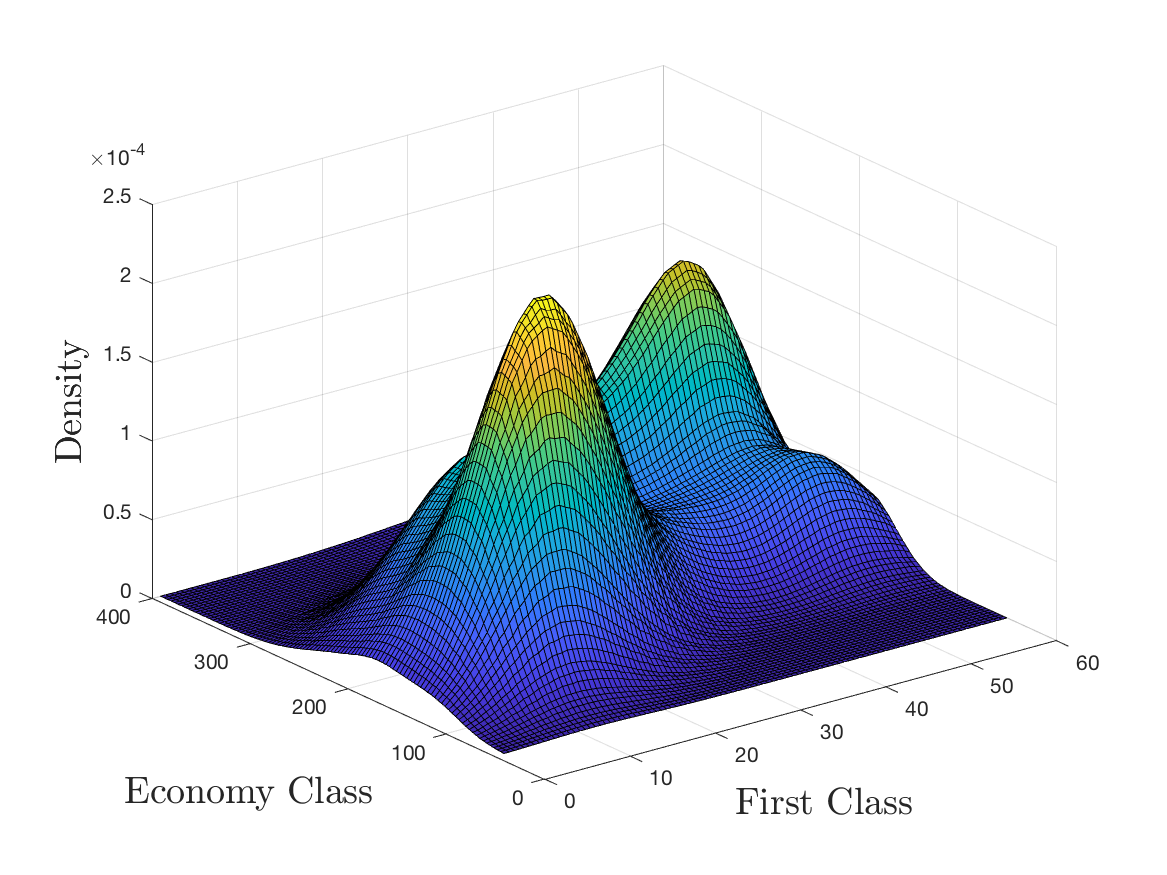}
    \caption{Density of Initial Capacities}
  \end{subfigure}
  \begin{subfigure}[b]{0.48\textwidth}
    \includegraphics[width=\textwidth]{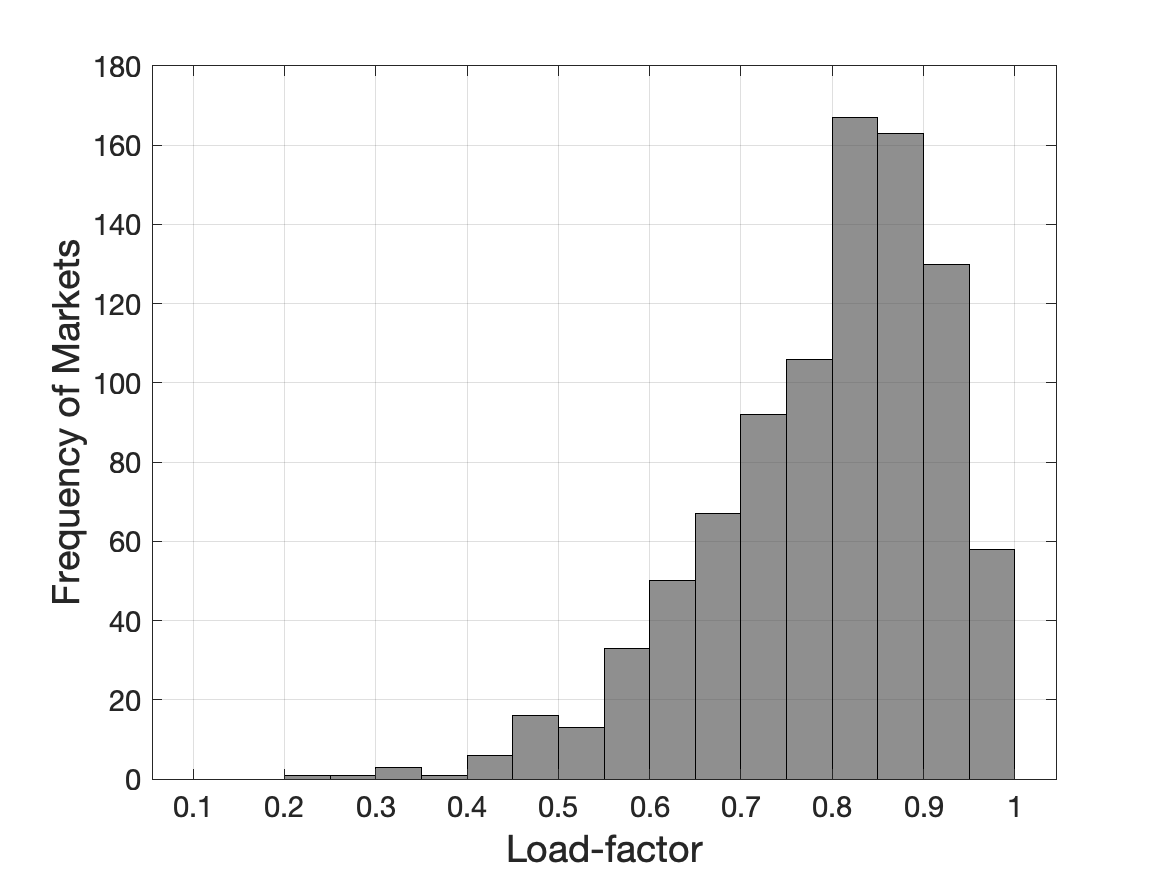}
    \caption{ Histogram of Load Factor}
  \end{subfigure}
  \begin{figurenotes}
 In part (a), this figure presents the Parzen-Rosenblatt Kernel density estimate of the joint density of initial capacities available for nonstop travel. In part (b), this figure presents the histogram of the passenger load factor across our sample.
 \end{figurenotes}

\end{figure}

The five most common aircraft types in our sample are the Boeing 777 (23\%), 767 (13\%), 757 (13\%), 737 (12\%), and the Airbus A340 (8\%). The 777, 747, and A340 are wide-body jets used typically on long-haul flights{, and they roughly correspond to the higher of the two peaks in Figure \ref{fig:initialcapacity}(a)}. The 737 has a typical configuration of around 160 seats, which roughly corresponds to the smaller of the two peaks in Figure \ref{fig:initialcapacity}(a). We refer to these two peaks as small and large modal capacities. Furthermore, the exact capacity configurations change from flight to flight. 

{The SIAT provides a \emph{sample} of passengers on each flight, but we cannot construct the full dynamic paths of seat purchases for individual flights. However, we merge the SIAT data with the Department of Transportation's T-100 segment data to measure the departure load factor for each flight. We discuss this issue later when we discuss the moments we use to estimate our model parameters.}
From the T100, we know the average load factor across a \emph{month} for a particular route flown by a particular type of equipment. In Figure \ref{fig:initialcapacity}(b), we display the histogram of load factor across flights in our sample. The median load factor is 82\%, but substantial heterogeneity exists across flights. {The low load factor itself suggests inefficiencies in this market, but we need a model of airline pricing to quantify this inefficiency.}

Overall, our descriptive analysis reveals several salient features we capture in our model. We find that a business-leisure taxonomy of passenger types helps capture differences in the timing of purchase, willingness-to-pay for an economy seat, and a vertically differentiated first-class seat. Further, we find evidence consistent with airlines responding to substantial heterogeneity in the business or leisure mix of passengers across markets, creating variation in both the level and temporal patterns of fares across markets. Finally, across flights, we observe considerable heterogeneity in fare paths as the flight date approaches. Together, these features motivate our model of non-stationary and stochastic demand and dynamic pricing by airlines that we present in Section \ref{model}, and the estimation in Section \ref{econometrics}.

\section{Model \label{model}}

In this section, we present a model of dynamic pricing by a profit-maximizing multi-product monopoly airline that sells a fixed number of economy seats ($0 \leq K^{e}<\infty$) and first-class seats ($0 \leq K^{f}<\infty$) on a  flight. We assume that nonstop passengers with heterogeneous and privately known preferences (i.e., their willingness-to-pay for an economy and first-class seat) arrive before departure ($T$). Additionally, the plane's capacity is occupied by passengers that use the flight segment as part of a multi-segment trip. Every period the airline has to choose ticket prices and the maximum number of seats to {release} at those prices, but before the demand realizes (for that period).

Our data indicate essential sources of heterogeneity in preferences that differ by reason-for-travel and purchase timing. Further, variability and non-monotonicity in {observed} fares suggest a role for uncertain demand. Our demand-side {model} seeks to flexibly capture this multi-dimensional heterogeneity and uncertainty that serves as an input into the airline's dynamic-pricing problem.
Furthermore, our supply-side {model} seeks to capture an airline's inter-temporal and intra-temporal trade-offs in choosing its optimal policy.

\subsection{Demand}

Let $N_t$ denote the number of nonstop individuals that \textit{arrive} in period $t\in\{1,\ldots,T\}$ to consider buying a nonstop ticket, where $T<\infty$ is the flight date. We model $N_{t}$ as a Poisson random variable with parameter $\lambda_t^n\in\mathbb{R}_+$, i.e., $\E(N_t)=\lambda_t^n$. {Additionally, let $C_t$ denote the number of connecting passengers that utilize the airplane in period $t\in\{1,\ldots,T\}$. We model $C_{t}$ also as a Poisson random variable with parameter $\lambda_t^c\in\mathbb{R}_+$, i.e., $\E(C_t)=\lambda_t^c$.} The airline knows $(\lambda_t^n,\lambda_t^c)$ for $t\in\{1,\ldots,T\}$, but must make pricing and seat-release decisions before the uncertainty over the \emph{realized} number of passenger arrivals is resolved each period. The {nonstop} passengers are one of two types, business or leisure. The probability that a given individual is business varies across time before departure and is denoted by $\theta_t \in[0,1]$.

To model individual preferences of nonstop passengers, we use a pure characteristics approach \citep[e.g.,][]{MussaRosen1978, BerryPakes2007, Molinari2021AER}. Passengers have different willingness-to-pay for flying, but everyone prefers a first-class seat to an economy seat. 
 Let $v\subset\mathbb{R}_{+}$ denote the value a passenger assigns to flying in the economy cabin, and let the indirect utility of this individual from flying economy and first-class at a price $p$, respectively, be
\begin{eqnarray*}
u^{e}(v,p, \xi) = v- p; \quad
u^{f}(v,p, \xi) = v\times \xi- p, \qquad \xi\in[1,\infty).
\end{eqnarray*}

\noindent The (utility) premium associated with flying in a first-class seat, $\xi$ captures the vertical quality differences between the two cabins. Passengers are heterogeneous in terms of their $v$ and $\xi$, and $v$ and $\xi$ are independent and privately known to the individual. Thus, we assume that utility from the outside option, which includes either not flying or flying through other routes, is normalized to zero.

We assume that the preferences across passengers are realizations from type-specific distributions. Specifically, the $v$ of business and leisure passengers are drawn from $F^{b}_v(\cdot)$ and $F^{l}_v(\cdot)$, respectively, and $\xi$ is drawn from $F_{\xi}(\cdot)$. {As $\xi$ enters multiplicatively in the utility function, passengers with higher base willingness-to-pay ($v$) will have a higher added value for first-class for the same draw of $\xi$. Regardless of their type, passengers can choose to fly on an economy-class or a first-class or not at all.}  Together with the arrival process, the type-specific distribution of valuations creates a stochastic and non-stationary demand process that we assume is known to the airline.

{We also model the arrival and purchase process for connecting passengers but using a simpler setup. There are potentially hundreds of connecting itineraries that use any given flight for one of the legs \citep{DixOrzach2022}. So, modeling the pricing process for these passengers would be computationally prohibitive. In the data, we observe how many connecting passengers purchase in each period, which gives us the ``rate" at which connecting passengers fill the plane. To capture this feature, we model connecting passengers as Poisson arrivals that passively take seats in the plane, alternating with the nonstop passengers.}

At given prices and the number of seats available at those prices, Figure \ref{fig:timing} summarizes a realization of the demand process for nonstop passengers. Specifically, demand realization and timing of information known by an airline leading up to a departure is as follows:
\begin{enumerate}[\hspace{1cm}(i)]
 \item Airline chooses a price and seat-release policy for economy cabin, $(p_{t}^{e},\overline{q}_{t}^{e})$, and the first-class cabin, $(p_{t}^{f},\overline{q}_{t}^{f})$, that determine the prices at which a maximum number of seats in the two cabins may be sold.
 \item $N_{t}$ nonstop individuals \textit{arrive}, the number being drawn from a Poisson distribution with parameter $\lambda_t^n$. Each passenger realizes their reason for flying from a Bernoulli distribution with parameter $\theta_{t}$ (i.e., business equals one). Each passenger observes their own $(v,\xi)$, drawn from the respective distributions, $F^{b}_v(\cdot)$, $F^{l}_v(\cdot)$, and $F_{\xi}(\cdot)$.
 \item $C_t$ connecting individuals arrive, the number is drawn from a Poisson distribution with parameter $\lambda_t^c$. A connecting passenger arrival is assigned as a potential economy-cabin passenger with a fixed probability of $r$ and a potential first-class passenger with the complementary probability.
 \item The nonstop and connecting passenger arrivals are allocated a seat alternatingly, but conditional on that, the seats are randomly allocated among either nonstop or connecting passengers. Thus, connecting passengers affects only the seats available for nonstop passengers.
 \item If neither seat-release policy is binding (realized demand does not exceed the number of seats released in either cabin), passengers select their most preferred cabin: first-class if $v\times \xi-p_t^f\geq \max\{0, v-p_t^e\}$, economy if $v-p_t^e\geq \max\{0, v\times \xi - p_t^f\}$, and no purchase if $0\geq \max\{v\times \xi -p_t^f,v-p_t^e\}$. Those passengers choosing the no-purchase option leave the market and never return. If the seat-release policy is binding in either or both cabins, we assume that passengers make sequential decisions in a randomized order until either none remaining wishes to travel in the cabin with capacity or all available seats are allocated.
 \item If there are more connecting passenger arrivals than nonstop, then the connecting passengers are allocated until none are left, or a capacity binds. 
 \item Steps (i)-(iv) repeat until the departure date, $t=T$, or all seats are allocated.

\end{enumerate}

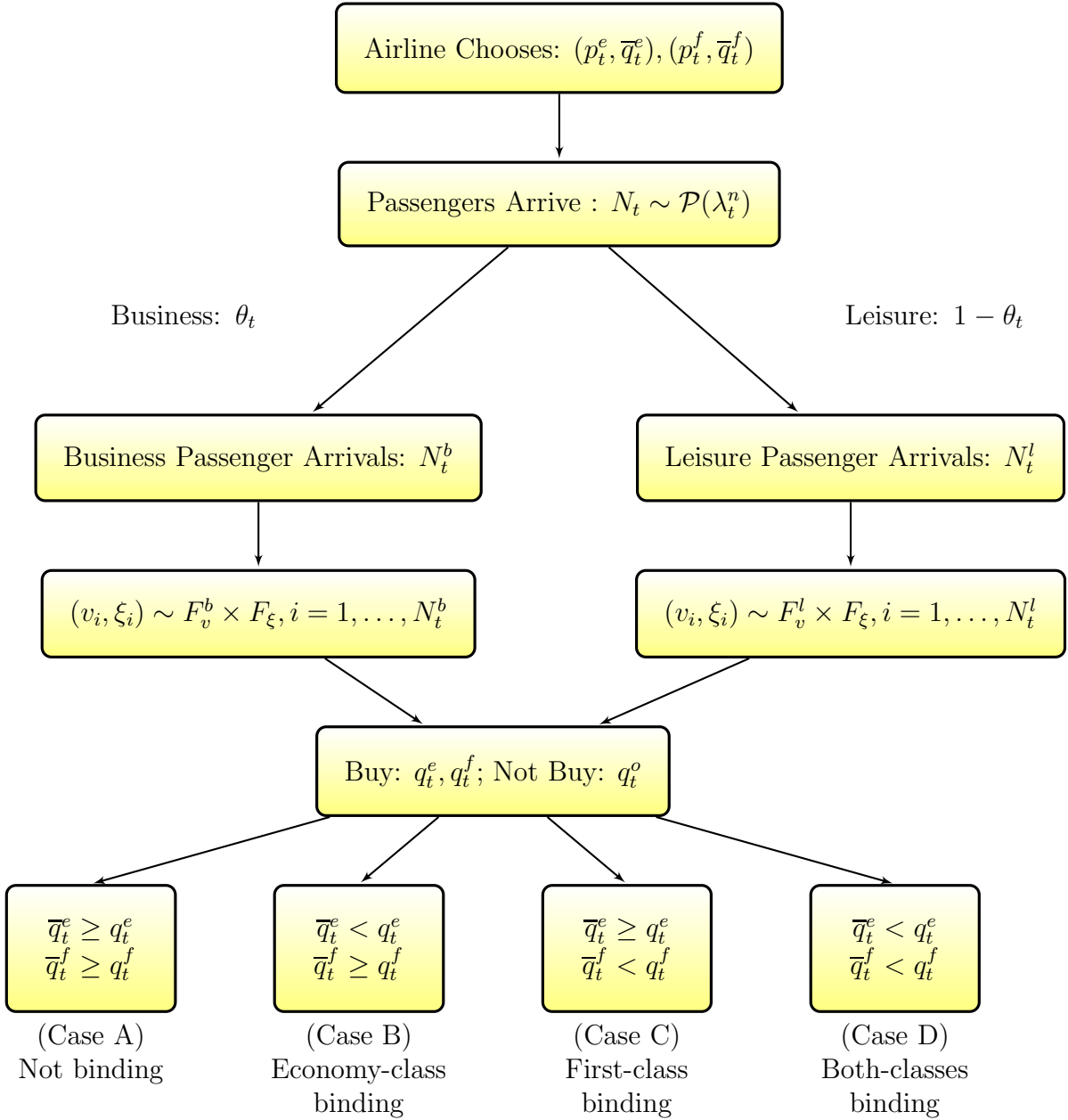
\begin{figure}
\caption{Realization of Demand of Nonstop Passenger Arrivals}\label{fig:timing}
\centering
\vspace{0.25cm}
\begin{tikzpicture}[node distance=1cm, scale=0.22]
\tikzset{
 mynode/.style={rectangle,rounded corners,draw=black, top color=white, bottom color=yellow!50,very thick, inner sep=1em, minimum size=3em, text centered},
 myarrow/.style={->, >=latex', shorten >=1pt, thick},
 mylabel/.style={text width=7em, text centered}
}
\node[mynode] (choices) {Airline Chooses: $(p_{t}^{e}, \overline{q}_{t}^{e}), (p_{t}^{f},\overline{q}_{t}^{f})$};
\node[mynode, below = of choices] (arrival) {Passengers Arrive : $N_{t}\sim {\mathcal P}(\lambda_{t}^n)$};
\node[below=3cm of arrival] (dummy) {};
\node[mynode, left=of dummy] (business) {Business Passenger Arrivals: $N_{t}^{b}$};
\node[mynode, right=of dummy] (leisure) {Leisure Passenger Arrivals: $N_{t}^{l}$};
\node[mylabel, below left=of arrival] (label1) {Business: $\theta_t$};
\node[mylabel, below right=of arrival] (label2) {Leisure: $1-\theta_{t}$};
\node[mynode, below = of business ] (wtp1) {$(v_{i},\xi_{i})\sim F^{b}_{v}\times F_{\xi}, i=1,\ldots, N_{t}^{b}$};
\node[mynode, below = of leisure ] (wtp2) {$(v_{i},\xi_{i})\sim F^{l}_{v}\times F_{\xi}, i=1,\ldots, N_{t}^{l}$};
\node[mynode, below = of wtp1, xshift=3.5cm] (demand) {Buy: $q_{t}^{e}, q_{t}^{f}$; Not Buy: $q_{t}^{o}$};
\node[mynode, below = of demand, xshift=-6cm] (case1) {$\begin{array}{c}\overline{q}_{t}^{e}\geq q_{t}^{e}\\ \overline{q}_{t}^{f}\geq q_{t}^{f}\end{array}$};
\node[mynode, below = of demand, xshift=-2cm] (case2) {$\begin{array}{c}\overline{q}_{t}^{e}< q_{t}^{e}\\ \overline{q}_{t}^{f}\geq q_{t}^{f}\end{array}$};
\node[mynode, below = of demand, xshift=2cm] (case3) {$\begin{array}{c}\overline{q}_{t}^{e}\geq q_{t}^{e}\\ \overline{q}_{t}^{f}< q_{t}^{f}\end{array}$};
\node[mynode, below = of demand, xshift=6cm] (case4) {$\begin{array}{c}\overline{q}_{t}^{e}< q_{t}^{e}\\ \overline{q}_{t}^{f}< q_{t}^{f}\end{array}$};
\draw[myarrow] (choices) -- (arrival);
\draw[myarrow] (arrival) -- (business);
\draw[myarrow] (arrival) -- (leisure);
\draw[myarrow] (business) -- (wtp1);
\draw[myarrow] (leisure) -- (wtp2);
\draw[myarrow] (wtp1) -- (demand);
\draw[myarrow] (wtp2) -- (demand);
\draw[myarrow] (demand) --(case1.north);
\draw[myarrow] (demand) --(case2.north);
\draw[myarrow] (demand) --(case3.north);
\draw[myarrow] (demand) --(case4.north);
\node[mylabel, below =of case1, yshift =1cm] {(Case A)\\ {Not binding}};
\node[mylabel, below =of case2, yshift =1cm] {(Case B)\\ {Economy-class binding}};
\node[mylabel, below =of case3, yshift =1cm] {(Case C)\\{First-class binding}};
\node[mylabel, below =of case4, yshift =1cm] {(Case D)\\{Both-classes binding}};
\end{tikzpicture}
\begin{figurenotes}
A schematic representation of the timing of demand for nonstop passenger arrivals.
\end{figurenotes}

\end{figure}

In any given period ($t$), there are four possible outcomes given a demand realization: neither seat-release policy is binding, first-class or economy class seat-release policies are binding, or both are binding. If the seat-release policy is not binding for either of the two cabins, then the expected demand for the respective cabins in period $t$ when the airline chooses policy $\chi_t:=(p_{t}^{e},\overline{q}_{t}^{e},p_{t}^{f},\overline{q}_{t}^{f})$ is
\begin{eqnarray*}
\E_t(q^e;\chi_t):=\sum_{n=0}^\infty \left\{n\times\Pr(N_t=n) \underbrace{\Pr(v-p_{t}^{e}\geq \max\{0, v\times\xi-p_{t}^{f}\})}_{:=P_t^e(\chi_t)}\right\}=\lambda_t \times P_t^e(\chi_t);\label{eq:demand1}\\
\E_t(q^f;\chi_t):=\sum_{n=0}^\infty \left\{n\times\Pr(N_t=n) \underbrace{\Pr(v\times \xi-p_{t}^{f}\geq \max\{0, v-p_{t}^{e}\})}_{:=P_t^f(\chi_t)}\right\}=\lambda_t\times P_t^f(\chi_t).\label{eq:demand2}
\end{eqnarray*}
\noindent If one or both of the seat-release policies are binding, the rationing process creates the possibility for inefficiencies by excluding passengers with a greater willingness-to-pay than those allocated a seat and misallocating passengers across cabins.

In Figure \ref{fig:seat-allocation-random}, we present a simple example to illustrate inefficiency arising from asymmetric information in this environment under random allocation. Assume the airline has one first-class and two economy seats remaining and chooses to release one seat in each cabin at $p^f=2000$ and $p^e=500$. Suppose three nonstop passengers arrive with values $v_1 = 2500$, $v_2 = 1600$, and $v_3 = 5000$, with $\xi_1=\xi_2=2$ and $\xi_3=1$. Passengers $n_1$ and $n_2$ are willing to pay twice as much for a first-class seat as an economy seat, whereas passenger $n_3$ values the two cabins equally. Furthermore, suppose two connecting passengers arrive, and out of these two, suppose the first connecting passenger $c_1$ is to be allocated an economy seat, and $c_2$ is to be allocated a first-class seat, if available. 
Suppose that under the random allocation rule, nonstop passengers get to choose before connecting, and among nonstop passengers, $n_2$ gets to choose first, and $n_3$ is the last, and among connecting passengers, suppose $c_1$ is allocated first. 

As shown in Figure \ref{fig:seat-allocation-random}, the final allocation is inefficient because a) passenger $n_2$ gets first-class even though $n_1$ values it more; and b) both nonstop passengers $n_1$ and $n_3$ do not get any seat. This difference in passengers' arrival time allows multiple welfare-enhancing trades. Given the limited opportunity for coordination among passengers to make such trades and the legal or administrative barriers, we believe random rationing is reasonable to allocate seats within a period.\footnote{Passengers considering flying are not always in the marketplace, looking for the best deal on tickets. They arrive at different times for exogenous reasons. Random rationing captures this realistic feature of the data. The seat release policies, presence of connecting passengers, and random assignments imply that there may be instances when a potential passenger with a high willingness-to-pay shows up to the market and cannot be served because a low willingness-to-pay passenger arrived earlier and was allocated the seat. If, instead, we use optimal rationing, where seats are assigned in the order of willingness-to-pay, it would lead to higher baseline efficiency.}

\begin{figure}
\caption{Illustration of Random Rationing Rule.}\label{fig:seat-allocation-random}
\begin{center}
\begin{tikzpicture}
[font=\sffamily,
 every matrix/.style={ampersand replacement=\&,column sep=1cm,row sep=1cm},
 source/.style={draw,thick,rounded corners,fill=yellow!20,inner sep=.23cm},
 process/.style={draw,thick,circle,fill=green!20},
 sink/.style={source,fill=yellow!20},
 sinks/.style={source,fill=green!20},
 datastore/.style={draw,very thick,shape=datastore,inner sep=.23cm},
 dots/.style={gray,scale=2},
 to/.style={->,>=stealth',shorten >=1pt,semithick,font=\sffamily\footnotesize},
 every node/.style={align=center, scale=0.7}]

 \matrix{
 \node[source] (hisparcbox) {Capacity: $\quad K^f=1, K^e=2$};
 \& \node[process] (daq) {$p^f=2000; \overline{q}^f=1$\\$p^e=500, \overline{q}^e=1$}; \& \\

 \& \node[datastore] (buffer) {$\begin{array}{ccc}v_1=1800& v_2=1600, & v_3=1900\\\xi_1=2, & \xi_2=2, & \xi_3=1\end{array}$}; \& \\

  \& \node[sinks] (monitor) {$\begin{array}{ccc}\texttt{nonstop}&\texttt{passenger-id}&\texttt{preference}\\ yes& n_1 & f\succ e\succ o\\yes&n_2&f\succ e\succ o\\yes&n_3&e\succ o\succ f\\no& c_1 & -\\no & c_2 &-\end{array}$};
 \& \node[sink] (datastore) {$\begin{array}{cc}\texttt{passenger-order}&\texttt{allocation}\\n_2& f\\c_1&e\\n_1 &o\\n_3 &o\\c_2&o\end{array}$}; \\
 };

 \draw[to] (hisparcbox) -- node[midway,above] {prices}
 node[midway,below] {seats} (daq);
 \draw[to] (daq) -- node[midway,left] {demand}node[midway,right]{realization} (buffer);
 \draw[to] (buffer) --
 node[midway,left] {preference}node[midway, right]{ordering} (monitor);

 \draw[to] (monitor) -- node[midway,above] {random}
 node[midway,below] {allocation}(datastore);
\end{tikzpicture}
\end{center}
\begin{figurenotes}
Example to demonstrate how random-rationing rule can generate inefficiency in the model.
\end{figurenotes}
\end{figure}
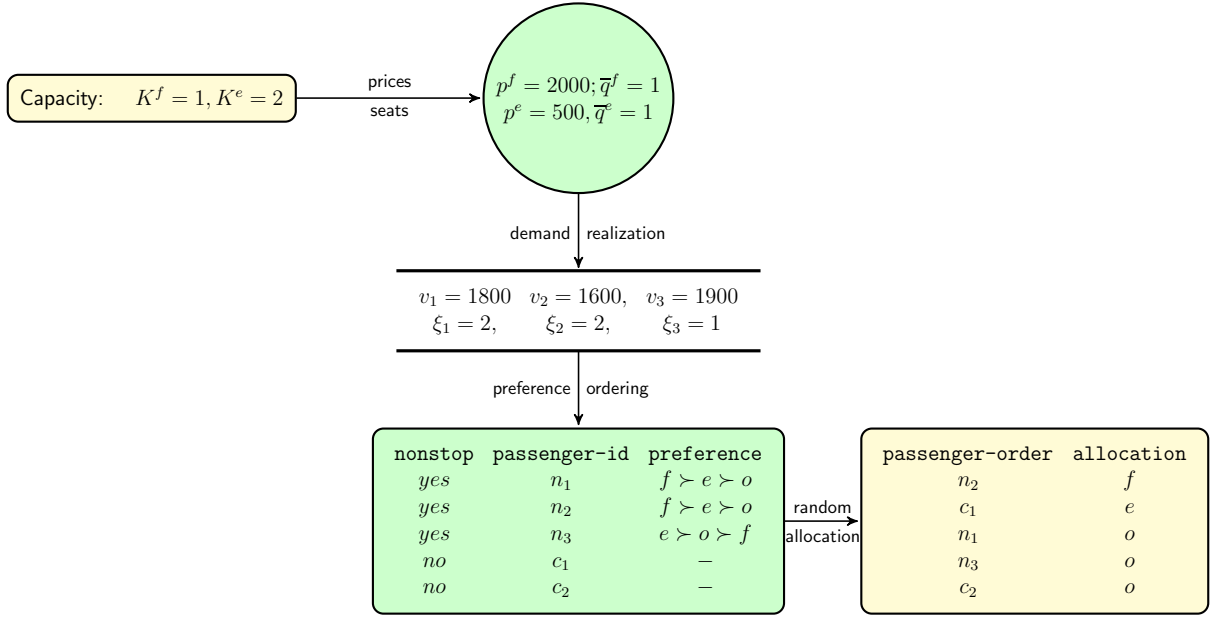

\subsection{Supply}

The airline has $T$ periods before the departure to sell $K^{e}$ and $K^{f}$ economy and first-class seats, respectively. Each period, the airline chooses prices $\{p_{t}^{e}, p_{t}^{f}\}$ and commits to selling no more than $\{\overline{q}_{t}^{e}, \overline{q}_{t}^{f}\}\leq \omega_t$ seats at those prices, where $\omega_t := (K_{t}^{e},K_{t}^{f})$ is the number of unsold seats in each cabin. {The airline cannot observe passengers' reasons for travel, so it cannot price differently for business and leisure passengers.} We model that airlines must commit to a seat-release policy to mimic the ``fare bucket'' strategy that airlines use in practice \citep[e.g.,][]{AlderighiNicoliniPiga2015}, which helps the airline insure against a favorable demand shock where too many seats are sold today at the expense of future passengers with higher willingness-to-pay. One of this market's defining characteristics is that the airline must commit to policies every period before realizing the demand. The airline does not observe a passenger's reason to fly or valuations $(v,\xi)$; however, the airline knows the underlying stochastic process that governs demand and uses the information to price discriminate, both within and across periods.\footnote{See \cite{BarnhartBelobabaOdoni2003} for an overview of forecasting airline demand.}

Let $c^e$ and $c^f$ denote the constant cost of servicing a passenger in the respective cabins.
These marginal costs, or the so-called ``peanut costs,'' capture variable costs like food and beverage service that do not vary with the timing of the purchase but may vary with the differential service given in the two cabins. Let $\Psi := (\{F^{b}_{v}, F^{l}_{v}, F_{\xi}, c^f, c^e \}, \{\lambda_{t}^{n}, \theta_{t}, \lambda_{t}^{c}\}_{t=1}^{T})$ denote the vector of demand and cost primitives.

The airline maximizes the sum of discounted expected profits by choosing price and seat-release policies for each cabin, $\chi_{t}= \left( p_{t}^{e}, p_{t}^{f} ,\overline{q}_{t}^{e}, \overline{q}_{t}^{f} \right)$, in each period $t=1,\ldots, T$ given $ \omega_{t}$. The optimal policy is a vector $\{\chi_t: t=1, \ldots, T\}$ that maximizes expected profit
 \begin{eqnarray*}
 \sum_{t=1}^{T} \mathbb{E}_{t} \left\{ \pi (\chi_{t},\omega_{t};\Psi_{t})\right\},
 \end{eqnarray*}
where $\pi
 (\chi_{t},\omega_{t};\Psi_{t}) = (p_{t}^{f}-c^{f})q_{t}^{f}+ (p_{t}^{e}-c^{e})q_{t}^{e}$ is the per-period profit after the demand for each cabin is realized ($q_{t}^{e}$ and $q_{t}^{f}$) and $\Psi_{t} = (\{F^{b}_{v}, F^{l}_{v}, F_{\xi}, c^f, c^e \}, \{\lambda_{t}^n, \theta_{t}, \lambda_{t}^c\})$.
 When choosing its policy, the airline observes the unsold capacity ($\omega_{t}$) but not the particular realization of passenger valuations that determine the realized demand. The optimal seat-release policy must satisfy $\overline{q}_{t}^{e} \leq K_{t}^{e}$ and $\overline{q}_{t}^{f} \leq K_{t}^{f}$ and take on integer values.

The stochastic process for demand, capacity-rationing algorithm, connecting passengers, and optimally chosen seat-release and pricing policies induce a non-stationary transition process between states, $Q_{t}(\omega_{t+1} | \chi_{t},\omega_{t},\Psi_{t})$. The optimal policy in periods $t\in\{1,\ldots, T-1\}$ is characterized by the solution to the Bellman equation,
\begin{eqnarray}
V_{t}(\omega_{t},\Psi)=\max_{\chi_{t}}\mathbb{E}_{t}\left\{
\pi(\chi_{t},\omega_{t};\Psi_{t})+
 \sum_{\omega\in\Omega_{t+1}}V_{t+1}(\omega_{t+1},\Psi)\times Q_{t}(\omega_{t+1}|\chi_{t},\omega_{t},\Psi_{t})\right\},\label{eq:value}
\end{eqnarray}
\noindent where $\Omega_{t+1}$ represents the set of reachable states in period $t+1$ given $\omega_{t}$ and $\chi_{t}$. The expectation, $\mathbb{E}_{t}$, is over realizations from the demand process ($\Psi_{t}$) from period $t$ to the date of departure $T$. In period $T$, optimal prices maximize
\begin{eqnarray*}
V_{T}(\omega_{T},\Psi_{T})=\max_{\chi_{T}}\mathbb{E}_{T}
\pi(\chi_{T},\omega_{T};\Psi_{T}),\label{eq:vT}
\end{eqnarray*}
because the firm no longer faces any inter-temporal trade-offs. We assume that passengers are short-lived for model tractability and cannot strategically time their purchases. So, the time they arrive at the marketplace and the time they purchase their tickets are the same and do not depend on the price path.\footnote{\cite{LiGranadosNetessine2014} use an instrumental variables strategy to study the strategic behavior of passengers and infer that between 5\% and 20\% of passengers wait to purchase in a sample of domestic markets, with the share decreasing in the market distance. We expect the share to be small in our context because we study international long-haul markets.} \cite{BoardSkrzypacz2016} allow consumers to be strategic under an additional assumption that the seller has a full commitment and chooses its supply function only once, in the first period. However, the assumption that airlines choose their fares only once at the beginning is too strong in the airline industry.
 The dynamic programming that characterizes an airline's problem helps us understand the shadow cost of an unsold seat.\footnote{Throughout the paper, we focus on only one flight and do not consider substitution across future flights.}

The optimal pricing strategy includes both inter-temporal and intra-temporal price discrimination. First, given the limited capacity, the airline must weigh allocating a seat to a passenger today versus a passenger tomorrow, who may have a higher mean willingness-to-pay because the fraction of business passengers increases as it gets closer to the flight date. This decision is difficult because both the volume ($\lambda_{t}$) and composition ($\theta_{t}$) of demand change as the date of departure nears. Thus, the good's perishable nature does not necessarily generate declining price paths like \cite{Sweeting2010}. Simultaneously, the airline must allocate passengers across the two cabins every period by choosing $\chi_t$ so that the price and supply restriction-induced selection into cabins is optimal.

To understand the problem further, consider the trade-off faced by an airline from increasing the price for economy seats today: (i) it decreases the expected number of economy seat purchases but increases the revenue associated with each purchase; (ii) it increases the expected number of first-class seat purchases but causes no change to revenue associated with each purchase; (iii) it increases the expected number of economy seats and decreases the expected number of first-class seats available to sell in future periods. Effects (i) and (ii) capture the multi-product trade-off faced by the firm, while (iii) capture the inter-temporal trade-off. More generally, differentiating Equation \ref{eq:value} with respect to the two prices gives two first-order conditions that characterize optimal prices given a particular seat-release policy:
\begin{eqnarray}
\left(\begin{array}{c}\mathbb{E}_{t}(q^{e};\chi_{t})\\ \mathbb{E}_{t}(q^{f};\chi_{t})\end{array}\right)
+\left[\begin{array}{cc}
\frac{\partial \mathbb{E}_{t}(q^{e};\chi_{t})}{\partial p_{t}^{e}}& -\frac{\partial \mathbb{E}_{t}(q^{f};\chi_{t})}{\partial p_{t}^{e}}\\-\frac{\partial \mathbb{E}_{t}(q^{e};\chi_{t})}{\partial p_{t}^{f}}&\frac{\partial \mathbb{E}_{t}(q^{f};\chi_{t})}{\partial p_{t}^{f}}
\end{array}\right] \left(\begin{array}{c}p_{t}^{e}-c^{e}\\p_{t}^{f}-c^{f}\end{array}\right)=\left(\begin{array}{c}\frac{\partial {\mathbb E}_{t}V_{t+1}}{\partial p_{t}^{e}}\\
\frac{\partial {\mathbb E}_{t}V_{t+1}}{\partial p_{t}^{f}}\end{array}\right)\label{eq:foc-f}.
\end{eqnarray}
The left side is the contemporaneous marginal benefit net of the ``peanut costs," while the right is the discounted future benefit.

Equation \ref{eq:foc-f} makes clear the two components of marginal cost: (i) the constant variable cost, or peanut cost, associated with servicing seats occupied by passengers; (ii) the opportunity cost of selling additional seats in the current period rather than in future periods. We refer to (iii), the vector on the right side of Equation \ref{eq:foc-f}, as the shadow cost of a seat in the respective cabins. These shadow costs depend on the firm's expectations regarding future demand (i.e., variation in the volume of passengers and business-leisure mix as the flight date nears) and the number of seats remaining in each cabin (i.e., $K_{t}^{f}$ and $K_{t}^{e}$). The stochastic nature of demand drives variation in the shadow costs, leading to equilibrium price paths non-monotonic in time. This flexibility is crucial given the variation observed in our data (see Figure \ref{fig:price_dists}).\footnote{The model implies a mapping between prices and the unobserved state (i.e., the number of remaining seats in each cabin). This rules out the inclusion of serially correlated unobservables (to the researcher) that shift demand or costs that could otherwise explain price variation.}

The decision to release an additional economy-class seat has multiple consequences. From an \emph{intra-}temporal point of view, an additional economy-class seat will lead to higher expected economy-class sales but might cannibalize revenue from the first-class cabin. From an \emph{inter-}temporal point of view, releasing an additional economy-class seat in this period may cannibalize future economy-class revenue, the same way lowering the fare this period may cannibalize future revenue. 

Formally, and holding prices fixed, the expected discounted profit from choosing a seat release policy $\bar{q}^e$ must be greater than an alternative of releasing one extra seat. The following inequality captures this trade-off:
\begin{eqnarray}
 \underbrace{[\mathbb{E}_{t}(q^{e}; {(\overline{q}^e+1)}) - \mathbb{E}_{t}(q^{e}; {\overline{q}^e})]\times (p^e_t-c^e)}_{\text{economy sales (+)}} &\leq& -\underbrace{[\mathbb{E}_{t}(q^{f}; {(\overline{q}^e+1)}) - \mathbb{E}_{t}(q^{f}; {\overline{q}^e})]\times (p^f_t-c^f)}_{\text{first-class sales (-)}} \notag\\
 &&- \underbrace{(\mathbb{E}_t V_{t+1}({(\overline{q}^e+1)}) - \mathbb{E}_t V_{t+1}({\overline{q}^e}))}_{\text{future expected profits (-)}}.\label{eq:inequality1}
\end{eqnarray}
The term on the left-hand side of \ref{eq:inequality1} is weakly positive because there is some chance of filling the extra seat. 
Likewise, if we ignore the minus sign in front, both terms on the right-hand side are weakly negative because releasing extra economy seats may cannibalize first-class today, and there will be fewer economy seats tomorrow, which lowers future expected profit. Similarly, if the airline offers one less economy seat, then the following inequality must hold:
\begin{eqnarray}
 \underbrace{[\mathbb{E}_{t}(q^{e}; {(\overline{q}^e-1)}) - \mathbb{E}_{t}(q^{e}; {\overline{q}^e})]\times (p^e_t-c^e)}_{\text{economy sales (-)}} &\leq& \underbrace{[\mathbb{E}_{t}(q^{f}; {(\overline{q}^e-1)}) - \mathbb{E}_{t}(q^{f}; {\overline{q}^e}]\times (p^f_t-c^f)}_{\text{first-class sales (+)}} \notag\\
 &&+ \underbrace{(\mathbb{E}_t V_{t+1}({(\overline{q}^e-1)}) - \mathbb{E}_t V_{t+1}({\overline{q}^e}))}_{\text{future expected profits (+)}}. \label{eq:inequality2}
\end{eqnarray}
Thus, inequalities (\ref{eq:inequality1}) and (\ref{eq:inequality2}) determine the optimal seat release policies. Although we discussed only one economy seat for brevity, similar logic applies to any seat in economy class or first-class. {In Appendix \ref{section:dp} we present simulation of optimal seat release policy and demand for economy seats for two types of markets; one with high arrival rate and the other with low arrival rate. }

The profit function is bounded, so the value function is well defined, and under some regularity conditions, we can show the uniqueness of the optimal policy. 
We present these regularity conditions and the proof of uniqueness in Appendix \ref{section:prooflemma}.

\section{Estimation and Identification\label{econometrics}}
In this section, we discuss the parametrization of the model, estimation methodology (generalized method of moments), and the sources of identifying variation. 
The model's parametrization balances the dimensionality of the parameters and the desired richness of the demand structure, and the estimation algorithm seeks to limit the number of times we have to solve our model due to its computational burden. At the same time, we hope to avoid strong assumptions on the relationship between model primitives and both observable (e.g., business travelers) and unobservable market-specific factors. Our identification discussion provides details of the moments we use in the estimation and how they identify each parameter.

\subsection{Model Parametrization and Solution}

Recall our model primitives, ${\Psi} = (\{F_{b}, F_{l}, F_{\xi}, c^f,c^e \}, \{\lambda_{t}^n, \theta_{t}, \lambda_{t}^c\}_{t=1}^{T})$, include distributions of valuations for business and leisure passengers, $(F_{b}, F_{l})$, distribution of valuations for 1st-class premium, $F_{\xi}$, marginal costs for economy and 1st-class, $(c^f,c^e)$, and the time-varying Poisson arrival rate of passengers, $\lambda_{t}^n$, the fraction of business passengers, $\theta_{t}$, and the Poisson rate of connecting passenger demand, $\lambda_{t}^c$. 
{Recall that the airlines do not observe passengers' reason-to-travel, so for them, the valuation is a linear mixture of business and leisure valuation with weights $\theta_t$ and $(1-\theta_t)$, respectively. }

Motivated by our data, we choose $T=8$ to capture temporal trends in fares and passengers' reasons for travel, where each period is defined as in Table \ref{sumstatTKT}. There are {three} demand primitives, $\lambda_{t}^n$ and $\theta_{t}$, and the connecting passengers, $\lambda_{t}^c$, that vary as the flight date approaches. To {capture} the relationship between the time before departure and these parameters, we use a linear parameterization, 
\begin{equation} 
\theta_{t}:=\min\left\{\Delta^\theta \times(t-1),1 \right\}; \qquad \lambda_{t}^n:= \lambda^n+\Delta^{\lambda}\times (t-1); \qquad \lambda_{t}^c:=\lambda^c + \Delta^{\lambda}\times (t-1),
\end{equation}
where $\Delta^{\theta},\lambda^n,\lambda^c$ and $\Delta^{\lambda}$ are scalar constants. This parametrization of the arrival process permits the volume ($\lambda$ and $\Delta^{\lambda}$) and composition ($\Delta^{\theta}$) of demand to change with time while also limiting the number of parameters to estimate. Here we are assuming that there is no dependence on the rate of arrivals between $t$ and $(t-1)$, which would drastically complicate the model.\footnote{We assume that the rate of change of connecting passenger arrivals is the same as the rate of change for the nonstop passengers, $\Delta^{\lambda}$ based on the aggregate patterns in our data.}

{Three distributions $( F_{b}, F_{l}, F_{\xi})$ determine passenger preferences. We assume that business and leisure passenger valuations distributions, $F_{b}$ and $F_{l}$, are Normal distributions with means $\mu^b$ and $\mu^{\ell}$, respectively, that are left-truncated at zero.} Given the disparity in average fares paid by business and leisure passengers, we assume $\mu^b \geq \mu^{\ell}$, which we model by letting $\mu^b =\mu^{\ell}\times (1+\delta^{b})$ with $\delta^{b} \geq 0$. The two cabins are vertically differentiated, and passengers weakly prefer first-class to the economy. To capture this product differentiation, we assume that the quality premium, $\xi$, equals one plus an Exponential random variable with mean $\mu^{\xi}$. 
{In other words, $\xi$ is the premium that all passengers place on first-class seats. The interaction between $\delta^b$ and $\mu^\xi$ determines the sorting of business travelers disproportionately into first class.}

In our final sample, 92\% of connecting passengers choose an economy-class seat. So we assign a connecting passenger arriving at the market as a potential economy-cabin passenger with a probability of $r=0.92$.
Finally, we fix the marginal cost of supplying a first-class and economy seat, $c^f$ and $c^e$, respectively, to equal industry estimates of marginal costs for servicing passengers. Specifically, we set $c^f=40$ and $c^e=14$ based on information from the International Civil Aviation Organization, Association of Asia Pacific Airlines, and \cite{Doganis2002}.\footnote{\cite{Doganis2002} finds the peanut costs in first-class are 2.9 times that of the economy and the average overall cost (across economy and first-class seats) is \$17.8 (inflation-adjusted). We use the average relative number of economy and first-class cabins to determine the cost of serving each seat. This exercise implies the costs for servicing one first-class and economy are \$40 and \$14, respectively. We chose not to estimate these peanut costs because we cannot identify them separately from market-specific demand parameters, and even if they were identified, estimating market-specific peanut costs would be computationally prohibitive.} The inter-temporal and intra-temporal changes in the shadow costs of seats are the primary reason for price variation. In international travel, where the average fare is substantially greater than domestic travel, these shadow costs are more important than passenger-related services' direct costs. Thus, we expect our estimates and counterfactuals to be less sensitive to our choice of peanut costs.\footnote{Our model does not incorporate taxes. Unit taxes would result in higher marginal costs. However, small changes to marginal costs do not qualitatively affect our results. If taxes are ad valorem, we are likely misinterpreting part of our willingness-to-pay estimate.} 

Given this parametrization of the model, the demand process can be described by a vector of parameters, ${\Psi}=\left( \mu^{\ell},cv^{\ell},\delta^{b},cv^{b},\mu^{\xi},\lambda^n,\lambda^c,\Delta^{\lambda},\Delta^\theta \right)\in[\underline{\Psi}, \overline{\Psi}]\subset\mathbb{R}^9$, where $cv$ denotes coefficient of variation. The model is a finite period non-stationary dynamic program. We solve the model for state-dependent pricing and seat-release policies by working backward; computing expected values for every state in the state space, where the state is the number of seats remaining in each cabin. The optimal policy is the solution to a mixed-integer nonlinear program (MINLP) at each state because seats are discrete and prices are continuous controls.\footnote{We perform computations in MATLAB R2020a, using MIDACO solver \citep{SchlueterGerdts2012}. We use warm start points for speed and accuracy by (a) solving the no-price discrimination problem first and (b) starting the MINLP problem for each state at the solution to an adjacent state.}

\subsection{Estimation}\label{estimation}

Although the parameterization above in Section 4.1 is for a single flight or a specific flight at a specific time between two airports, our data represent many diverse fights, and there may be many observed and unobserved factors that impact model primitives in an unknown way. For example, the distance or commerce between cities may affect willingness-to-pay for a first-class seat. Instead of further parameterizing the model as a function of observables, we propose a flexible approach to estimate the distribution of flight-level heterogeneity. The approach has the benefit of limiting the number of times the model is solved.

To understand our approach, consider the following example. We have many instances of the SEA-TPE route in our data, and we treat the prices and quantities sold on each instance of this route as a separate flight. Demand for such routes may vary across seasons; for example, there may be a higher willingness-to-pay for flights in the summer during the tourist season than in the winter. One approach would be to incorporate the observable characteristics of different flights (e.g., season, sporting events, college attendance) and allow them to affect the willingness-to-pay through some functional form.

Instead, {we use} a random coefficients model to estimate a distribution of demand primitives across flights. So two different instances of the SEA-TPE route (two different flights) are allowed to differ in their demand primitives, and the differences due to seasonality in demand will be captured by (parameters of) the distribution. We take this approach because (1) including enough observables to capture differences across flights would result in too many parameters to feasibly estimate the model, and (2) for our counterfactuals, our primary goal is to learn the distribution of demand and not so much about the relationship between prices and flight-market observable characteristics.

We estimate the model using the generalized method of moments. In particular, our approach combines the methodologies of \cite{Ackerberg2009}, \cite{FoxKimYang2016}, \cite{NevoTurnerWilliams2016}, and \cite{BlundellGowrisankaranLanger2020}. We posit that empirical moments are a mixture of theoretical moments, with a mixing distribution that we know up to a finite-dimensional vector of parameters. To limit the computational burden of estimating these parameters that describe the mixing distribution, we rely on the importance sampling procedure of \cite{Ackerberg2009}. Our estimation proceeds in three steps. First, we calculate moments from the data to summarize the heterogeneity in equilibrium outcomes within and across flights. Second, we solve the model once, at $S$ different parameter values that cover the parameter space $[\underline{\Psi}, \overline{\Psi}]$. Third, we optimize an objective function that matches the empirical moments to the analogous moments for a mixture of candidate data-generating processes. The mixing density that describes across-market heterogeneity in our data is the object of inference.

Specifically, for a given level of observed initial capacity, $\omega_{1}:=(K^{f}_{1}, K^{e}_{1})$, our model produces a data-generating process characterized by parameters $\Psi$ that describe demand and costs. This data-generating process can be described by a set of $N_{\rho}$--many moment conditions that we denote by $\boldsymbol{\rho}(\omega_{1};\Psi)$. We assume that the analogous empirical moment conditions, $\rho(\omega_{1})$, can be written as a mixture of candidate moment conditions, i.e.,
\begin{eqnarray}
 \rho(\omega_{1}) = \int_{\underline{\Psi}}^{\overline{\Psi}}\boldsymbol{\rho}(\omega_{1};{\Psi}) h(\Psi|\omega_{1})d\Psi,\label{eq:mixture}
\end{eqnarray}
where $h(\Psi|\omega_{1})$ is the conditional (on initial capacity $\omega_{1}$) density of $\Psi$.\footnote{Thus, we estimate a different density for each initial capacity $\omega_1$, which captures the possibility that the observed initial capacities may be correlated with the unobserved demand.}

The goal is to estimate the mixing density, $h(\Psi|\omega_{1})$,
that best matches the empirical moments (left side of Equation \ref{eq:mixture}) to the expectation of the theoretical moments (right side of Equation \ref{eq:mixture}). To identify the mixing density, we assume a particular parametric form for $h(\Psi|\omega_{1})$ that reduces the matching of empirical and theoretical moments to a finite-dimensional nonlinear search. Specifically, we let the distribution of $\Psi$ conditional on $\omega_{1}$ be a truncated multivariate normal distribution, i.e., 
\begin{equation*}
\Psi|\omega_{1}\sim h(\Psi|\omega_{1};{\mu_{\Psi}},{\Sigma_{\Psi}}),
\end{equation*}
where $\mu_{\Psi}$ and $\Sigma_{\Psi}$ are the vector of means and covariance matrix, respectively, of the non-truncated distribution. We choose our estimates based on a least-squares criterion
\begin{eqnarray}
\label{eq:objfun}
\left( \hat{\mu}_{\Psi}(\omega_{1}), \hat{\Sigma}_{\Psi}(\omega_{1}) \right)=\arg\min_{ \left( \mu_{\Psi}, \Sigma_{\Psi} \right) } \left( \hat{\rho}(\omega_{1})- \widetilde{\mathbb{E}}({\boldsymbol{\rho}}(\omega_{1};\mu_{\Psi},\Sigma_{\Psi}) )\right)^{\top} \left( \hat{\rho}(\omega_{1})- \widetilde{\mathbb{E}}({\boldsymbol{\rho}}(\omega_{1};\mu_{\Psi},\Sigma_{\Psi})) \right)
\end{eqnarray}
where $\hat{\rho}(\omega_{1})$ is an estimate of the ($M\times1$) vector of empirical moments and $\widetilde{\mathbb{E}}({\boldsymbol{\rho}}(\omega_{1};\mu_{\Psi},\Sigma_{\Psi}))$ is a Monte Carlo simulation estimate of $ \int_{\underline{\Psi}}^{\overline{\Psi}}\boldsymbol{\rho}(\omega_{1};{\Psi}) h(\Psi|\omega_{1};\mu_{\Psi},\Sigma_{\Psi})d\Psi$ equal to
$
 \frac{1}{S}\sum_{j=1}^{S}\boldsymbol{\rho}(\omega_{1};{\Psi}_{j})
$
with the $S$ draws of $\Psi$ taken from $h(\Psi|\omega_{1};\mu_{\Psi},\Sigma_{\Psi})$.\footnote{To model the variance-covariance matrices, we use the representation proposed by \cite{ArchakovHansen2021}. For the consistency of our estimator we assume that, for each initial capacity $\omega_1$, the number of flights and the number of passengers in those flights are sufficiently large so that $\hat{\rho}(\omega_{1})$ is a consistent estimator of the true moment ${\rho}(\omega_{1})$, and our importance sampling procedure to determine $\widetilde{\mathbb{E}}({\boldsymbol{\rho}}(\omega_{1};\mu_{\Psi},\Sigma_{\Psi}) )$ is also consistent. For a formal analysis of the subject, see \cite{GourierouxMonfortRenault1993} Proposition 1.}

The dimensionality of the integral we approximate through simulation requires many draws, and we choose $S=10,000$. Thus, the most straightforward approach to optimization of Equation \ref{eq:objfun} would require solving the model $S=10,000$ times for each value of $(\mu_{\Psi}, \Sigma_{\Psi})$ until a minimum is found. Our model is complex, and the dimensionality of the parameter space to search over makes such an option prohibitive. For this reason, we appeal to the importance sampling methodology of \cite{KloekDijk1978} as adapted in this context by \cite{Ackerberg2009}.

The integral in Equation \ref{eq:mixture} can be rewritten as
\begin{eqnarray*}
\int_{\underline{\Psi}}^{\overline{\Psi}}\boldsymbol{\rho}(\omega_{1};{\Psi}) \frac{h(\Psi|\omega_{1};\mu_{\Psi},\Sigma_{\Psi})}{g(\Psi)}g(\Psi) d\Psi
,\label{eq:impsamp}
\end{eqnarray*}
where $g(\Psi)$ is a known well-defined probability density with strictly positive support for $\Psi \in \left[ \underline{\Psi} , \overline{\Psi} \right] $ and zero elsewhere like $h(\Psi | \omega_{1};\mu_{\Psi},\Sigma_{\Psi})$. Recognizing this, one can use importance sampling to approximate this integral with
\begin{eqnarray*}
\label{simexp}
 \frac{1}{S}\sum_{j=1}^{S}\boldsymbol{\rho}(\omega_{1};{\Psi}_{j}) \frac{h(\Psi_{j}|\omega_{1};\mu_{\Psi},\Sigma_{\Psi})}{g(\Psi_{j})}
\end{eqnarray*}
where the $S$ draws of $\Psi$ are taken from $g(\Psi)$. Thus, the importance sampling serves to correct the sampling frequencies so that it is as though the sampling was done from $h(\Psi|\omega_{1};\mu,\Sigma)$.

The crucial insight of \cite{Ackerberg2009} is that this importance sampling procedure separates the problem of solving the model from the optimization of the econometric objective function.\footnote{Importance sampling approaches have also been used to estimate models of discrete games \citep{bajari2010identification} and demand with large consideration sets \citep{Molinari2021AER}.} That is, we solve the model for a fixed number of $S$ draws of $\Psi$ from $g(\Psi)$, and then $\boldsymbol{\rho}(\omega_{1};{\Psi}_{j})$ is calculated once for each draw. After these calculations, optimization of the objective function to determine $(\hat{\mu}_{\Psi}(\omega_{1}), \hat{\Sigma}_{\Psi}(\omega_{1}))$ simply requires repeatedly calculating the ratio of two densities, $\frac{h(\Psi_{j}|\omega_{1};\mu_{\Psi},\Sigma_{\Psi})}{g(\Psi)}$. To simplify the importance sampling process, we fix the support of $g(\cdot)$ and $h(\cdot)$ to be the same and let $g(\cdot)$ be a multivariate uniform distribution with the support [$\underline{\Psi}$,$\overline{\Psi}$] chosen after substantial experimentation to ensure it encompasses those patterns observed in our data.

To solve Equation \ref{eq:objfun}, we use a combination of global search algorithms and multiple starting values. We repeat this optimization for each $\omega_{1}$ which provides an estimate of the parameters of the distribution of market heterogeneity, $(\hat{\mu}_{\Psi}(\omega_{1}), \hat{\Sigma}_{\Psi}(\omega_{1}))$. To calculate the distribution of demand parameters across all flights, we then appropriately weight each estimate by the probability mass associated with that value of $\omega_{1}$ (Figure \ref{fig:initialcapacity}). We calculate standard errors for the estimates and the counterfactuals by re-sampling the individual passenger observations in the SIAT data. This procedure accounts for errors in survey responses and variation in our moments across flights.

\subsection{Identification}
In this section, we describe the moments we use to estimate the model and how they help inform the parameter estimates. Our model is nonlinear, we use several moments, and they help inform more than one parameter, but below, we discuss the primary relationships between the moments and the parameters. 

We estimate the parameters for twenty initial capacities, $\omega_{1}$, and the results we display in the next section weigh the parameters by the likelihood of each capacity in our sample. 
For a given initial capacity $\omega_1$, we use the following moments: 
\begin{enumerate}[(i)]
 \item the distribution of economy- and first-class fares for each period;
 \item the distribution of the change in fares for each period;
 \item the distribution of the maximum and minimum differences in first-class and economy fares, i.e., $ \max_{t=1,\ldots,T} \{ p^{f}_{t} - p^{e}_{t}\}$ and $ \min_{t=1,\ldots,T}\{p^{f}_{t} - p^{e}_{t}\}$, respectively;
 \item the distribution of the nonstop passengers sold each period as a fraction of total nonstop passengers;
 \item the distribution of load factors across flights, as shown in Figure \ref{fig:initialcapacity}(b);
 \item the distribution of connecting passengers sold each period as a fraction of total connecting passengers;
 \item the proportion of business travelers in each period, as shown in Figure \ref{fig:mktbus}.
\end{enumerate}

We assume that the airline chooses prices optimally, given the remaining capacity and knowledge of the demand process. Therefore, observed prices and changes in prices captured by moments (i) and (ii) directly inform the mean willingness-to-pay and the standard deviation of willingness-to-pay. In particular, since we assume that the willingness-to-pay distribution is constant across periods, first-period prices (before business passengers arrive and before the time-dependent shocks to the arrival process are realized) are particularly important for identifying mean willingness-to-pay. If prices change substantially from period to period, the airline is less certain about the state process, so the willingness-to-pay distribution must have a higher variance.\footnote{Our identification arguments are similar to others who identify demand in dynamic settings using restrictions from the supply side, e.g., \cite{HendelNevo2006, NevoHendel2013}, and \cite{HortacsuMcAdams2010}.} 

The maximum differences between economy and first-class prices captured by moments (iii) identify passengers' premium on first class $(\mu^\xi)$. To see this, consider a passenger with $(v,\xi)$. She buys a first-class ticket if $\xi\geq (p_{t}^{f}-p_{t}^{e})/v$, and economy if $\xi\leq(p_{t}^{f}-p_{t}^{e})/v$. Because $\xi$ is time-invariant and independent of $v$, this equality should hold across all passengers and at all times, which gives
$\frac{\max_{t}(p_{t}^{f}-p_{t}^{e})}{\min\{v: \texttt{bought first-class}\}}\leq (\xi-1)\leq \frac{\min_{t}(p_{t}^{f}-p_{t}^{e})}{\max\{v: \texttt{bought economy}\}}$,
where $\min\{v: \texttt{bought first-class}\}$ and $\max\{v: \texttt{bought economy}\}$ are the minimum and maximum value among those who buy first-class and economy, respectively.

The fractions of sales across periods (moments (iv)) and the variability of prices inform us about the arrival process, $\lambda_t^n$. From the assumption $\lambda_{t}^n= \lambda^n+\Delta^{\lambda}\times (t-1)$, it follows that in the initial period, there are on average $\lambda^n$-many passenger arrivals at the marketplace.  
Consider two markets with the same capacity but different arrival processes. 
If the average number of passenger arrivals in one market is higher than in the other, the prices in the former market start higher, and vice versa. If, however, the two markets have the same rate of initial arrivals ($\lambda^n$) but different slopes ($\Delta^\lambda$), then the prices in the market with a smaller slope start higher and have a flatter slope than the other market. In other words, differences in ($\lambda^n, \Delta^\lambda$) are reflected in the initial price differences and the price variance over time.
Thus, the positive relationship between the size of the demand, i.e., the average number of passenger arrivals, and the rate at which prices change (i.e., increase) over time for a given flight suggests that we can use variation in the difference in prices at $t=0$ and price at $t>0$ to identify $(\lambda^n, \Delta^{\lambda})$.\footnote{Thus, we rely on the assumption that shocks to the arrivals process across time are uncorrelated. So far as this assumption is violated, we may misinterpret variation in price levels across markets as reflecting differences in the willingness-to-pay distribution instead of correlated shocks within markets.} Likewise, the distribution of the load factor (moment (v)) is also informative about the size of the market. We do not observe the full purchase counts from the SIAT data, so the load factor from the T-100 data is crucial information to help discipline the market size. We also observe connecting passengers (moments (vi)), which identify $\lambda^c$.

We observe the fraction of business passengers (moments (vii) and Figure \ref{fig:timetrend}(b)), so that informs us about the share of business passengers in the model, $\theta_t$. 
To see the intuition, after suppressing the time index, we note that  
$
\theta:=\Pr(\texttt{business})=\Pr(\texttt{business|buy})\times \Pr(\texttt{buy})+\Pr(\texttt{business|not-buy)}\times \Pr(\texttt{not-buy}),
$
where $\Pr(\texttt{business|buy})$ is estimable from the data, and given $(\lambda_1^n, \Delta^\lambda)$ we can also determine $\Pr(\texttt{buy})$ and $\Pr(\texttt{not-buy})$. So we can identify $\theta$ as long as we can identify $\Pr(\texttt{business|not-buy})$.
{Business travelers have a higher mean willingness-to-pay than leisure travelers. So, conditional on those who do not buy, the share of business travelers is smaller than the share of leisure travelers, i.e., $\Pr(\texttt{business|not-buy})\leq \Pr(\texttt{leisure|not-buy})$. However, when there are no business passengers in the early periods, those who do not buy a ticket are all leisure passengers. Moreover, the functional form assumption, $\theta_t=\min\{\Delta^{\theta}\times (t-1),1\}$, implies that $\Pr(\texttt{business|not-buy})$ is proportional to a constant (time-invariant) parameter $\Delta^\theta$. Conditional on us identifying the valuation distributions, we identify $\Pr(\texttt{business|not-buy})$, and from that we can identify $\theta$. In particular, at $t=1$ we get $\theta_1=0$, and from the number of sales at $t=1$ we can determine $\Delta^\theta$ and hence $\theta_t=\Pr(\texttt{business in period $t$})$. }

Finally, to identify the distribution of the random coefficients, $h(\Psi|\omega_{1};\mu_{\Psi},\Sigma_{\Psi})$, we rely on the variation in all of the above moments across flights. For example, if prices vary substantially across flights, then we would estimate a large variance for the means of the willingness-to-pay.\footnote{Our approach is similar to that of \cite{NevoTurnerWilliams2016}, who use variation across households in their usage of telecommunications services when facing nonlinear prices with stochastic future demand to identify the weights on each latent household types.}


\section{Results\label{results}}
In this section, we present our estimation results. First, we discuss how our estimates capture sources of across-market heterogeneity as implied by the estimated conditional densities $h(\cdot|\omega_1; \widehat{\mu_{\Psi}}, \widehat{\Sigma_{\Psi}})$, given the initial capacity $\omega_1$.
Second, we calculate the distribution of opportunity costs for a seat and show how they vary across cabins and the time until departure. We discuss model fit in Appendix \ref{sec:fit}. 

\subsection{Market Heterogeneity}

We estimate parameters that determine the means and covariances of the demand across all markets and capacities in our sample. In Table \ref{demandhet0}, column (1), we present the mean values of these nine demand parameters averaged across all of our markets. The mean willingness-to-pay for a (one-way) economy-class seat by leisure passengers across our entire sample is \$392.43 (s.e. \$52.39), and for the small and large modal markets, the mean values are \$391.24 (s.e. \$61.42) and \$488 (s.e. \$54.21), respectively.\footnote{We use the bootstrap procedure to calculate the standard errors. In particular, we sample the data 100 times with replacement and compute the empirical moments. Then, using the model-implied moments for all 10,000 markets, we estimate $h(\cdot|\omega_1^*; \mu_{\Psi}, \Sigma_{\Psi})$ for each capacity, using the estimates from $h(\cdot|\omega_1^*; \widehat{\mu_{\Psi}}, \widehat{\Sigma_{\Psi}})$ as the starting values. We repeat these steps 100 times and determine the standard errors.} 

The demand estimates vary with the initial capacities. In Table \ref{demandhet0}, columns (2) and (3), we present the mean values of the demand implied by the estimates for the market with small and large modal capacities, respectively; see Figure \ref{fig:initialcapacity}(a). Additionally, in Tables \ref{table:appendix2} and \ref{table:appendix1}, we show the mean and mode of the nine demand parameters associated with the conditional estimated density $h(\cdot|\cdot;\widehat{\mu_{\Psi}}, \widehat{\Sigma_{\Psi}})$, different initial capacities $\omega_1$.

The mean value of the coefficient of variation of willingness-to-pay is 0.19, and, on average, a business traveler values an economy-class ticket 37\% more than a leisure passenger as $\mu^b=\mu^l\times(1+\delta^b)$. On average, passengers' willingness-to-pay for a first-class seat is 58\% more than an economy seat. The estimates of the arrival processes suggest that there are, on average, 32 nonstop and 50 connecting passengers in the first period. These numbers decrease over time because $\Delta^\lambda=-0.03$. However, the fraction of business passengers increases from 0 in the first period to 35\% in the last period. We see a similar pattern for the two modal markets. In Appendix \ref{appendix1}, Tables \ref{table:appendix1} and \ref{table:appendix2} we display the mean of the demand parameters across all twenty initial capacities for which we estimated the model.

\begin{table}[ht!]
\begin{center}
\caption{Estimates Implied Demand Parameters} \label{demandhet0}
\begin{tabular}{lccc}
\toprule
 Demand & Mean & Mean & Mean\\
 Parameters & All Markets & Small Modal Capacity & Large Modal Capacity\\
 & (1) & (2) & (3)\\
 \hline
$\mu^l$  &392.43 & 391.24 & 488.92 \\ 
  &(52.39) & (61.42) & (54.21) \\ 
$cv^l$  &0.19 & 0.08 & 0.08 \\ 
  &(0.06) & (0.09) & (0.03) \\ 
$\delta^b$  &0.37 & 0.12 & 0.28 \\ 
  &(0.06) & (0.04) & (0.09) \\ 
$cv^b$  &0.73 & 0.33 & 0.26 \\ 
  &(0.08) & (0.07) & (0.07) \\ 
$\mu^\xi$  &0.58 & 0.38 & 0.52 \\ 
  &(0.08) & (0.06) & (0.11) \\ 
$\lambda^n$  &32.85 & 25.05 & 28.26 \\ 
  &(2.97) & (0.73) & (4.00) \\ 
$\Delta^{\lambda}$ &-0.03 & -0.02 & -0.00 \\ 
  &(0.01) & (0.01) & (0.01) \\ 
$\Delta^{\theta}$ &0.05 & 0.04 & 0.06 \\ 
  &(0.01) & (0.01) & (0.00) \\ 
$\lambda^c$  &50.28 & 30.55 & 36.59 \\ 
  &(3.86) & (5.51) & (3.44) \\ 
 \bottomrule
\end{tabular}
\end{center}
\begin{figurenotes}
In this table we present the mean of the demand parameters from the marginal density $h(\Psi|\omega_1^*;\widehat{\mu}_{\Psi},\widehat{\Sigma}_{\Psi})$ (in column (1)) weighted across all capacities, and conditional densities $h(\Psi|\omega_{1}^{*};\widehat{\mu_{\Psi}},\widehat{\Sigma_{\Psi}})$, where $\omega_{1}^{*}\in\{(160,18), (265, 47)\}$ are the small and large modal capacities (in columns (2) and (3)), respectively. Bootstrapped standard errors are in the parentheses.
\end{figurenotes}
\end{table}

\begin{figure}[t!]
 \centering
 \caption{\em Market Heterogeneity: Marginal Densities of Demand Parameters}
 \includegraphics[width=0.9\textwidth]{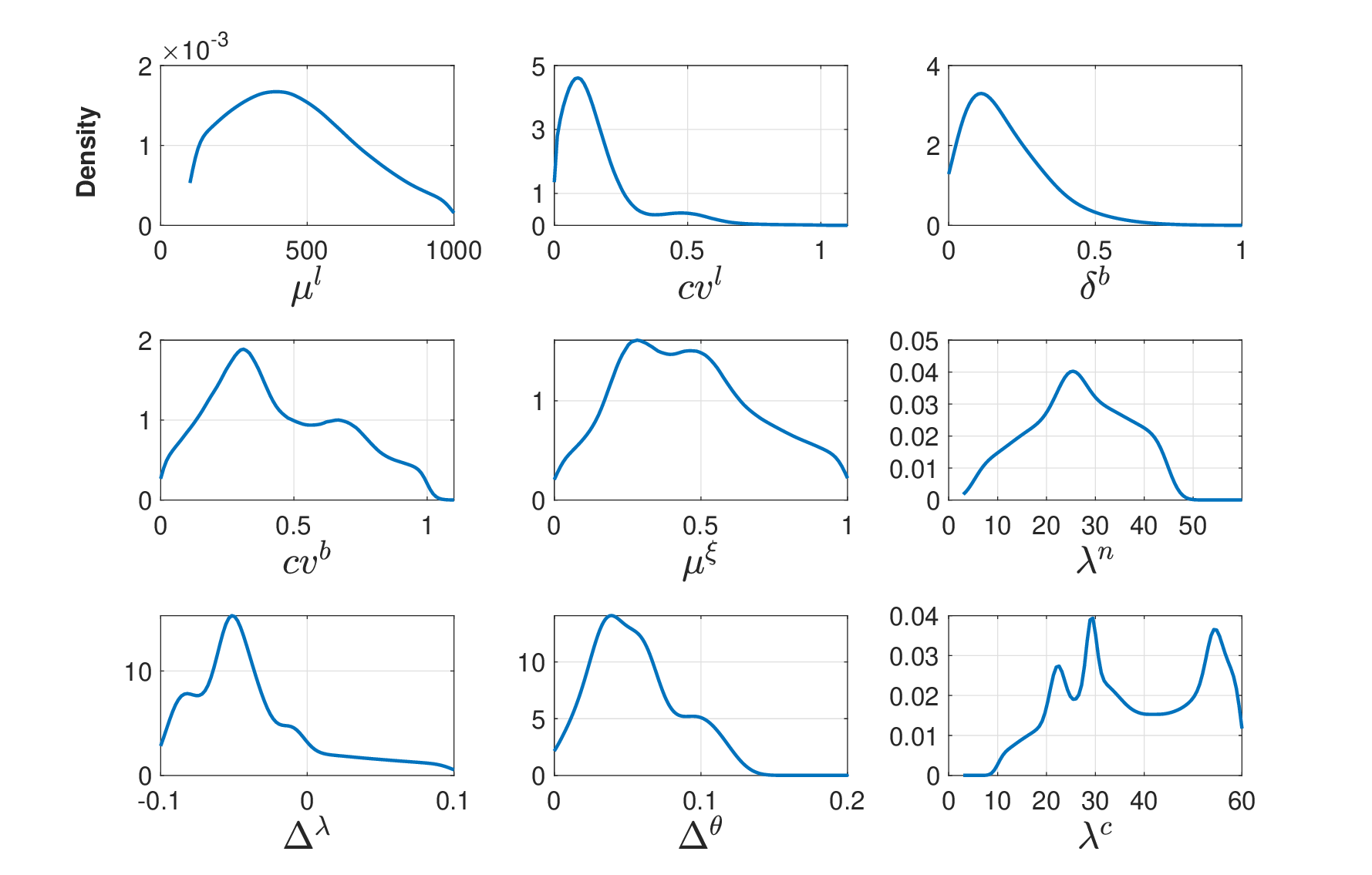} 
 \begin{figurenotes}
This figure displays the weighted average of marginal densities for the nine parameters parameters associated with $h(\cdot|\omega_{1}^{*};\widehat{\mu_{\Psi}},\widehat{\Sigma_{\Psi}})$ for the 20 most common initial capacities, where the weights are densities of these 20 capacities given in Figure \ref{fig:initialcapacity}(a). For example, the density of $\mu_{\ell}$ (the mean willingness to pay for leisure passengers) is displayed in the top left cell.\label{fig:marginal}
\end{figurenotes} 
\end{figure}

However, we have many different markets in our sample. To get a sense of the heterogeneity across these markets, in Figures \ref{fig:marginal} and \ref{fig:joint}, together, we show the marginal and joint densities of the demand parameters ($\Psi$). In fig \ref{fig:marginal} we display the marginal densities of each parameter, and in Fig \ref{fig:joint} we display the joint between the parameters that we allow to have covariance, weighted by the joint density of initial capacities as shown in Figure \ref{fig:initialcapacity}(a). The main takeaway is that there is substantial heterogeneity in all the nine demand parameters across markets and that the arrival rate parameters are correlated: the rate of arrival of nonstop passengers is correlated with that of connecting passengers.

To get a sense of what the estimates imply about the demand curves, we conduct a simple simulation exercise to measure the implied price elasticities of demand. For simplicity, we consider the case when capacity constraints do not bind, and there are enough passengers and focus only on the economy class seats. We simulate willingness-to-pay from the (overall) mean market, Table \ref{demandhet0}, column (1), and compute elasticities of demand at the model implied average price for leisure and business passengers separately. We find that price elasticity for leisure passengers is $-3.9$, and for business passengers, it is $-0.51$. These elasticities are similar to \cite{BerryJia2010}, who find that the own-price elasticities for leisure and business passengers of $-5.01$ and $-0.44$, respectively, for a sample of U.S. domestic flights from 1999.
 
\begin{figure}[t!]
 \centering
 \caption{\em Market Heterogeneity: Joint Densities of Demand Parameters}
 \includegraphics[width=0.98\textwidth]{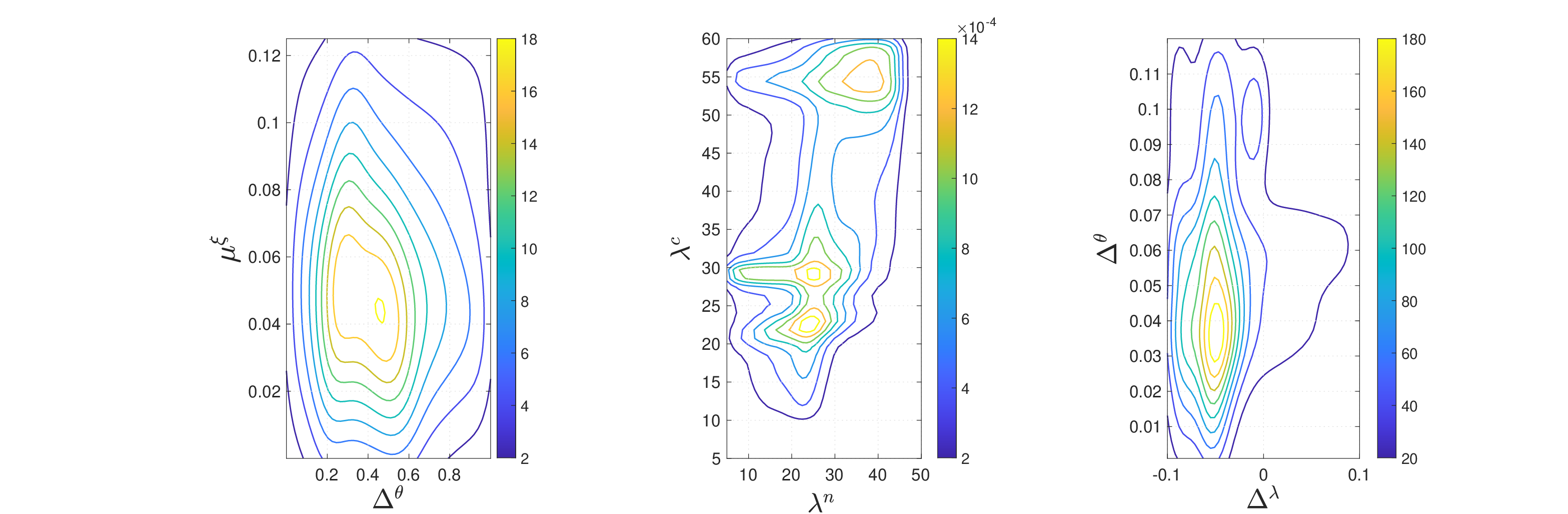} 
 \begin{figurenotes}
This figure displays (clockwise) the contours of weighted average of joint densities of $(\Delta^{\theta},\mu^\xi), (\lambda^n,\lambda^c)$ and $(\Delta^{\lambda}, \Delta^\theta)$ implied by $h(\cdot|\omega_{1}^{*};\widehat{\mu_{\Psi}},\widehat{\Sigma_{\Psi}})$ for the 20 most common initial capacities, where the weights are densities of these 20 capacities given in Figure \ref{fig:initialcapacity}(a).\label{fig:joint} 
 \end{figurenotes}

\end{figure}

\subsection{Estimates of Willingness-to-Pay and Shadow Costs}
Using our estimates, we can determine the implied densities of the willingness-to-pay for an economy seat and a first-class seat and how these densities change over time.
Recall that the density of the willingness-to-pay for an economy seat in period $t$ is the mixture $\theta_t \times f_v^b(\cdot) + (1-\theta_t)\times f_v^l(\cdot)$, and the density of the willingness-to-pay for an first-class seat in period $t$ is similar to the economy seat augmented by $f_{\xi}(\cdot)$. These densities, for periods $t=1,3,5,8$, are displayed in Figure \ref{fig:modelmean_wtp}(a) and (b).
Densities of the willingness-to-pay for a first-class seat are to the right of those for an economy seat by $\xi$. 
Using the estimates for $\lambda_t$ and $\theta_t$ from column 3 of Table \ref{demandhet0}, in Figure \ref{fig:modelmean_arrival}, we display the average number of direct passenger arrivals and the share of business passengers by time.  

\begin{figure}[ht!]
 \centering
 \caption{\em Willingness-to-Pay for a Seat by Cabin, for Large Modal Capacity }
 \begin{subfigure}[b]{0.45\textwidth}
 \includegraphics[width=\textwidth]{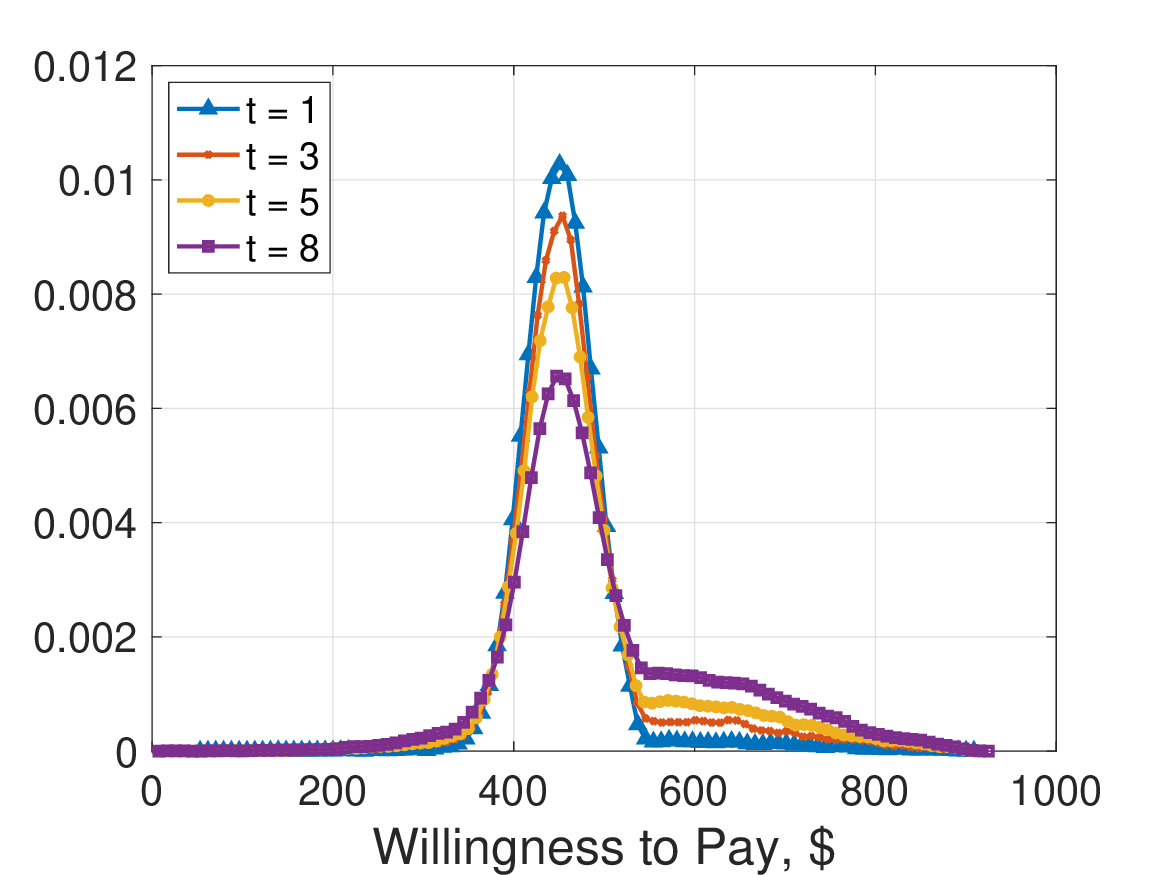}
 \caption{PDF of WTP for an Economy Seat}
 \end{subfigure}
 \begin{subfigure}[b]{0.45\textwidth}
 \includegraphics[width=\textwidth]{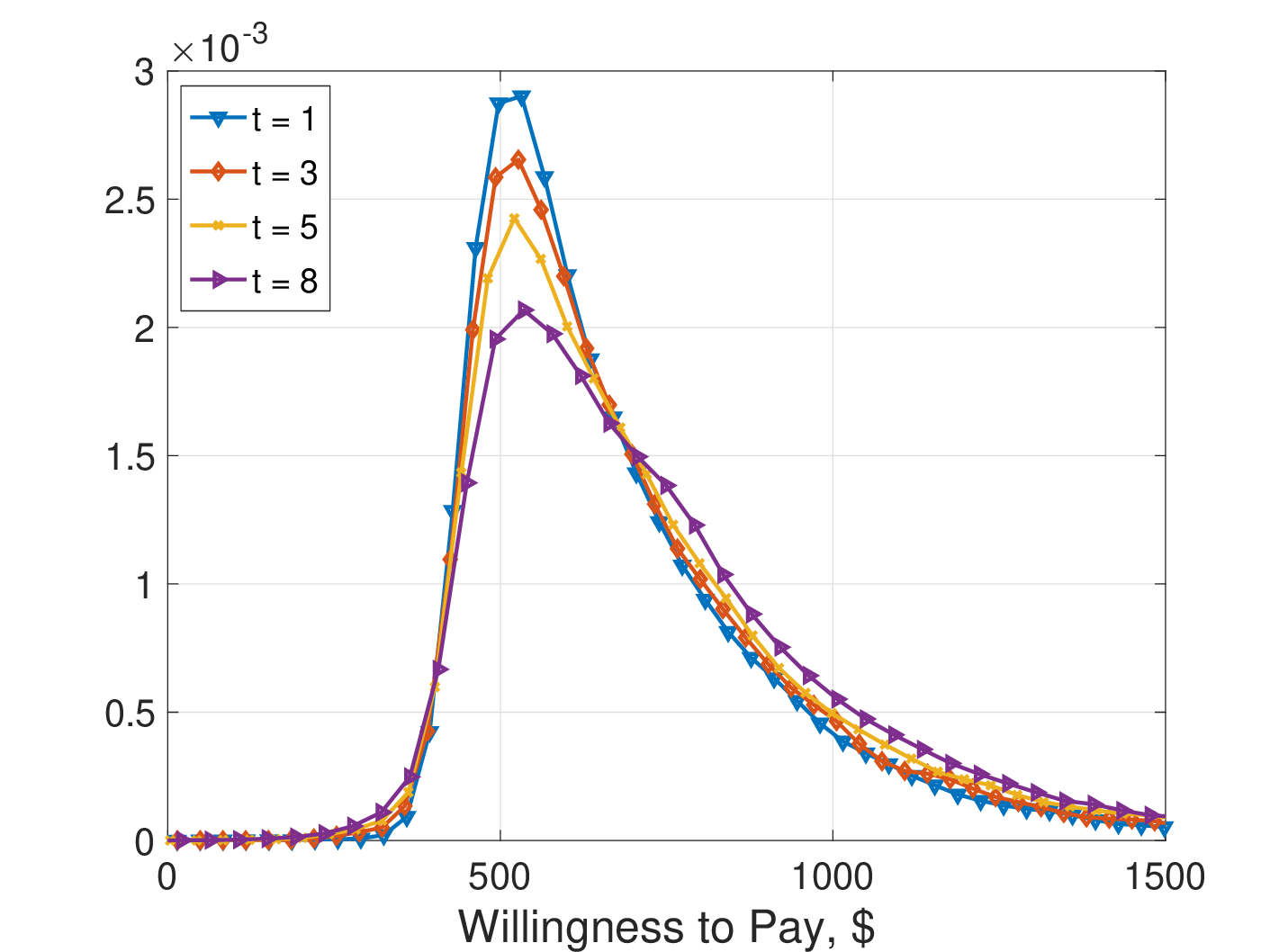}
 \caption{PDF of WTP for a First-Class Seat}
 \end{subfigure}
 \begin{figurenotes}
This figure displays the densities of WTP for an economy seat and a first-class seat for the large modal capacity $\omega_1^*=(265, 47)$ and the mean market. \label{fig:modelmean_wtp}
\end{figurenotes}
\end{figure}

\begin{figure}[ht!]
 \centering
 \caption{\em Number of Passenger Arrivals for Large Modal Capacity }
 \includegraphics[scale=0.34]{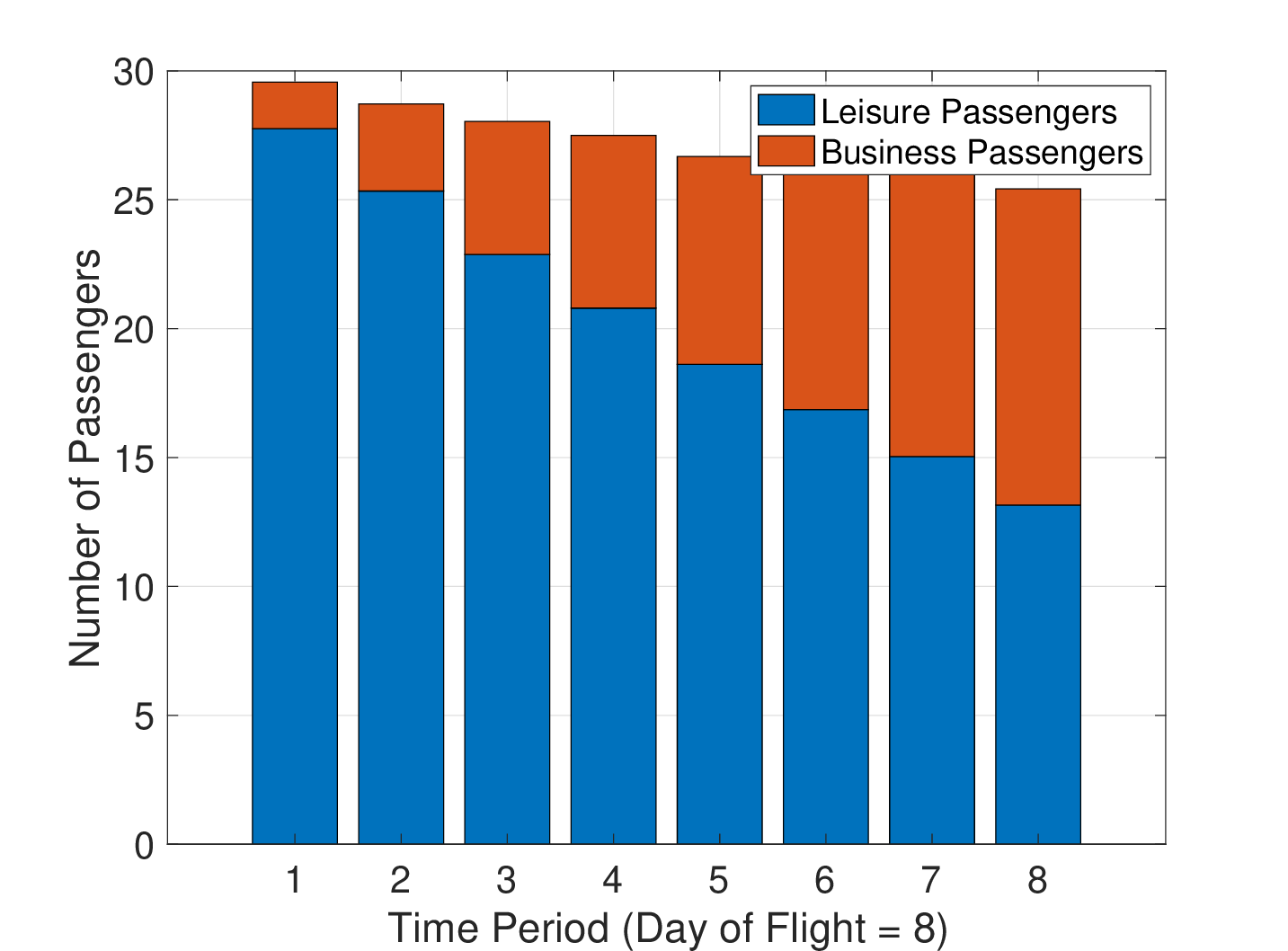}
 \begin{figurenotes}
This figure displays the number of passenger arrivals by their purpose of travel for the large modal capacity $\omega_1^*=(265, 47)$ and the mean market. \label{fig:modelmean_arrival}
\end{figurenotes}
\end{figure}

It is also illustrative to consider what these parameters imply about the time-varying shadow costs of a seat. The total marginal cost of a seat comprises its ``peanut cost," which is constant, and the opportunity cost varies over time depending on the state's evolution, i.e., the number of unsold economy seats and first-class seats. The shadow costs are the right-hand side of Equation \ref{eq:foc-f}, the change in expected value for a change in today's price. In other words, the shadow cost is the cost of future revenues to the airline of selling an additional seat today. 

In Figure \ref{fig:statetransition}, we present the state's evolution in terms of the contours corresponding to the state's joint density, as implied by our model estimates. 
Consider $\omega_{1}$, which is the initial capacity for this modal capacity market. So, when we move to the subsequent few periods, we see that the uncertainty increases. However, as we get closer to the departure time, the contours move towards the origin, which means fewer seats might remain unsold with time.
The contour of the state at the time of departure ($\omega_{\texttt{dept}}$) denotes the distribution of the state at the time of departure. {Even at the time of departure, the contours are not degenerate but they suggest slightly higher load factor than the data; see Figure \ref{fig:initialcapacity}(b).} That contours get dispersed as the flight date approaches is a way to visualize demand uncertainty. 

\begin{figure}[th!]
\begin{center}
 \caption{\em Evolution of Seats Remaining} \label{fig:statetransition}
 \includegraphics[scale=.4]{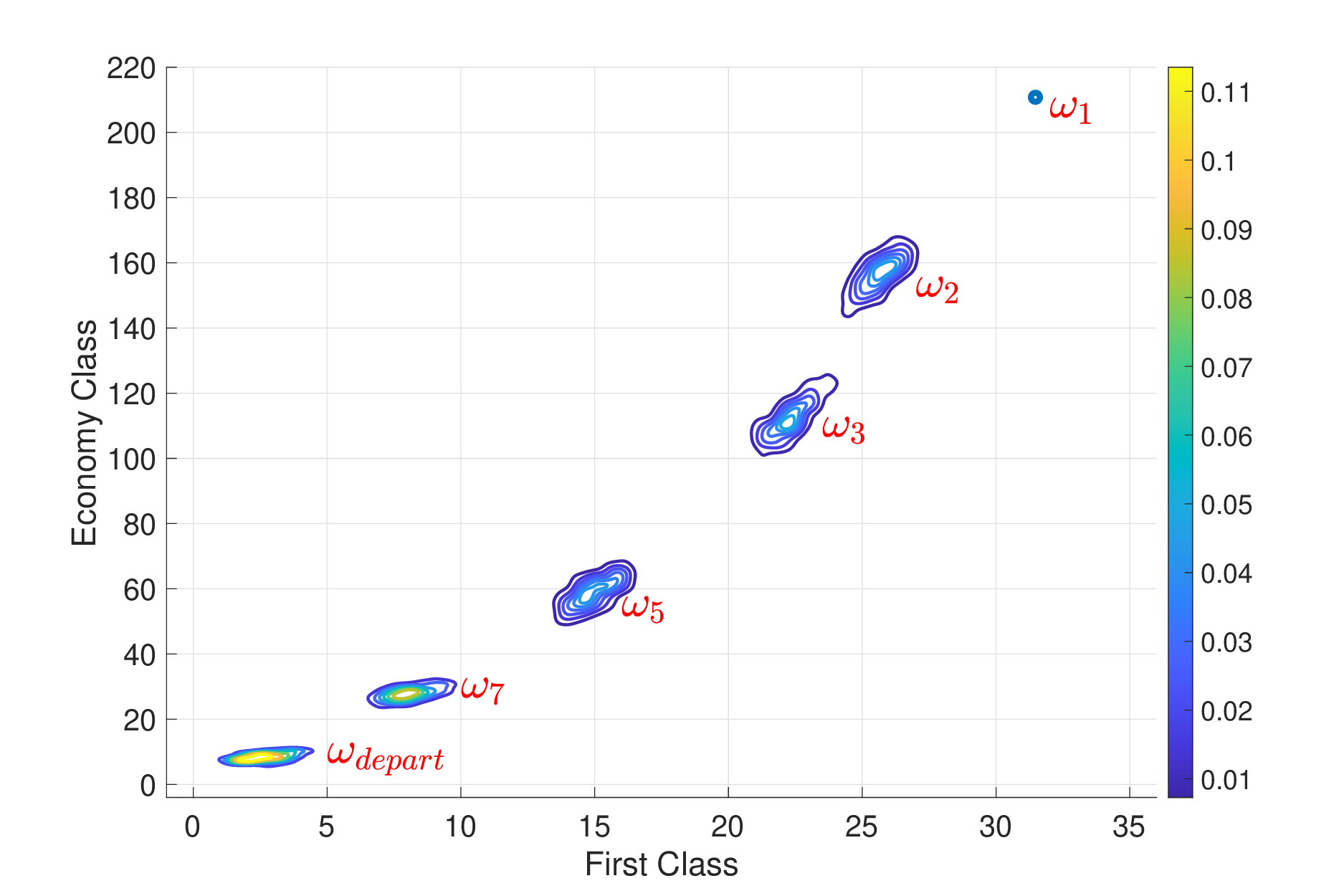}
\end{center}
\begin{figurenotes}
The figure displays the contours corresponding to the joint density of unsold seats for every period across all capacities and markets.
\end{figurenotes}
\end{figure}

An implication of this demand volatility is the implied volatility in the value of a seat to the airline, i.e., the seat's opportunity cost.
In Figure \ref{fig:costCDF}, we present the distributions of the marginal cost for an economy and first-class seat realized in equilibrium, averaged across all markets and all capacities. This feature graphically relates the state transitions to the shadow costs.
In particular, in Figure \ref{fig:costCDF}, we take the distribution of states realized in a given period (Figure \ref{fig:statetransition}) and sample the total marginal costs (sum of the derivative of value functions with respect to price and the peanut cost) based on those frequencies and then plot the distributions. In panel (a), we present the distributions for an economy seat, and in panel (b), we present the distributions of a first-class seat.
As can be seen, there is a significant variation in the costs. These variations are crucial for our identification, as they are the underlying reason for dispersion in the observed fares.

In the first period, $t=1$, there is no uncertainty about the state, which means the marginal cost distribution degenerates at \$127.06 for an economy seat and \$381.3 for a first-class seat. With $t>1$, the distributions become more dispersed but with little change in the mean. For instance, the means of the marginal costs for an economy seat in periods $t=3, 4, 5$ and $7$ are \$139.95, \$145.2, \$151.77, and \$159.03, respectively. We observe similar pattern for a first-class seat; the mean marginal costs are, approximately, \$456.54, \$463.22, \$449.48 and \$430.65, in periods $t\in\{3,4,5,7\}$, respectively. 
Finally, in the last period, $t=8$, the opportunity cost of a seat is zero, as the marginal cost is only the peanut cost.

\begin{figure}[t!]
 \centering
 \caption{\em Distributions of Marginal (Opportunity) Cost of a Seat}
 \hspace{-0.65cm}
 \begin{subfigure}[b]{0.50\textwidth}
 \includegraphics[width=\textwidth]{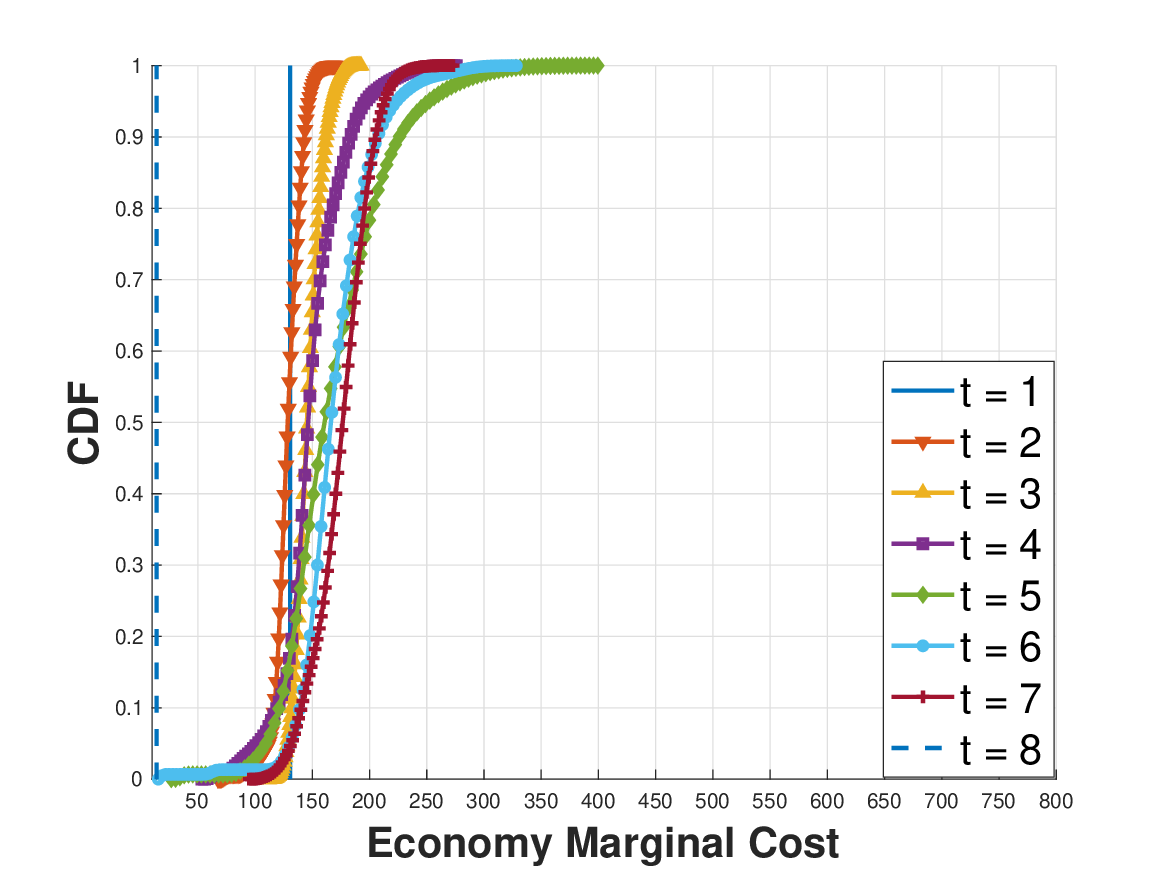}
 \caption{ Economy}
 \end{subfigure}
 \hspace{-0.65cm}
 \begin{subfigure}[b]{0.50\textwidth}
 \includegraphics[width=\textwidth]{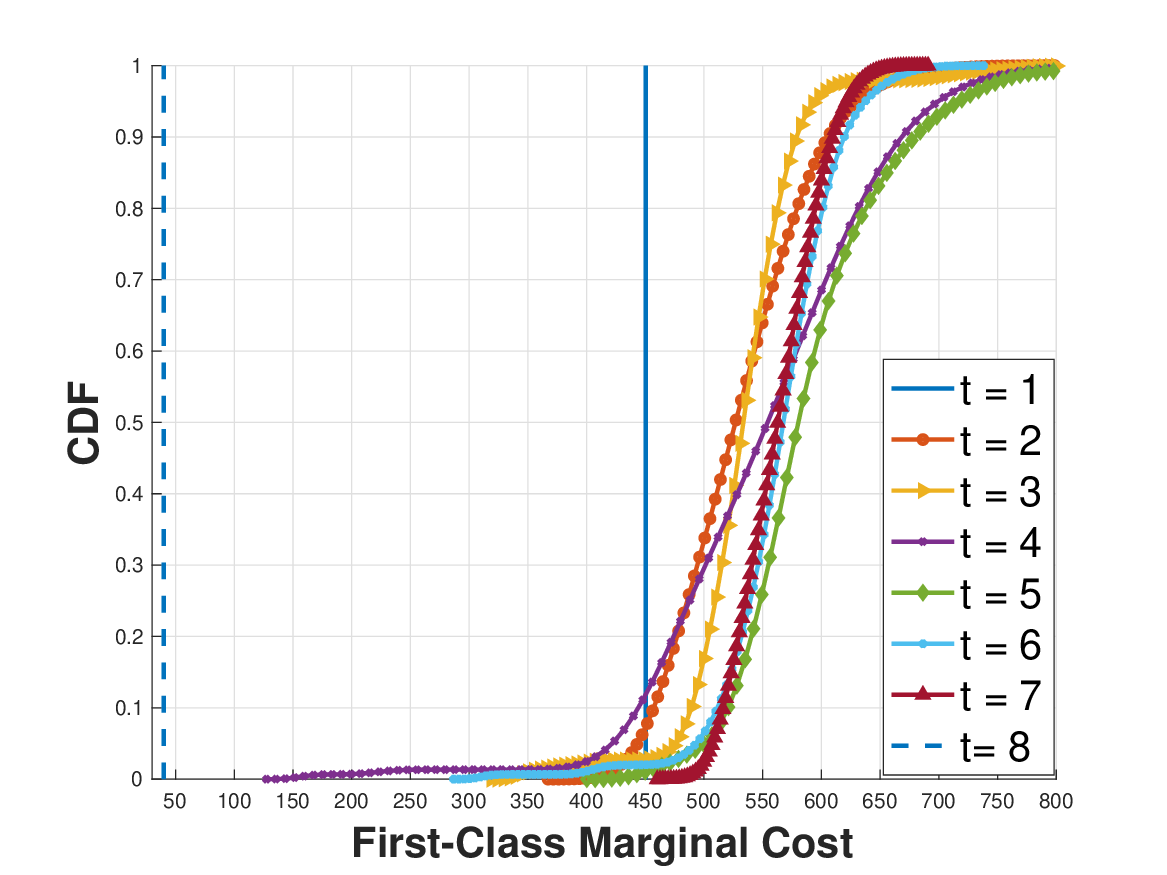}
 \caption{First-Class}
 \end{subfigure}
\begin{figurenotes}
Figures display the opportunity-cost distributions at each time for an economy and a first-class seat, respectively. \label{fig:costCDF}
\end{figurenotes}
\end{figure}

\section{Inefficiency and Welfare\label{counterfactuals}}
Two sources of informational friction in this market contribute to inefficiency: asymmetric information about passengers' valuations and uncertainty about future demand. A passenger's valuation may be due to idiosyncratic preferences and may also be associated with their reason for travel. Airlines' inability to price based on a passenger's reason for travel or even the idiosyncratic valuation can distort the seats' final allocation.\footnote{This source of inefficiency is present in a static monopoly setting when the monopolist chooses one price. However, in our setting, the monopolist must consider the unknown valuations of future arriving passengers.} The second source leads to inefficient allocations of limited capacity because the airline chooses its prices and {seat-release} policies before the demand is realized. Intra-temporal and inter-temporal misallocations introduced by these frictions represent opportunities for welfare-improving trade. Using counterfactual pricing and allocation mechanisms, we quantify the inefficiencies attributable to these sources. We first show how to visualize these sources of inefficiencies using a schematic representation of a welfare triangle. We then present and discuss our results.

\subsection{Welfare Triangle}

Consider the first-best allocation: seats are allocated to the highest-valuation passengers $(v,\xi)$ regardless of the timing of their arrival at the market. Under this allocation, the division of the surplus would depend on the prices. Figure \ref{fig:welfare} shows the line between A and B associated with this efficient benchmark and forms the welfare triangle (CAB). Point A represents the complete extraction of consumer surplus (i.e., price equals valuation), and point B represents maximum consumer surplus (i.e., prices equal the peanut costs).

 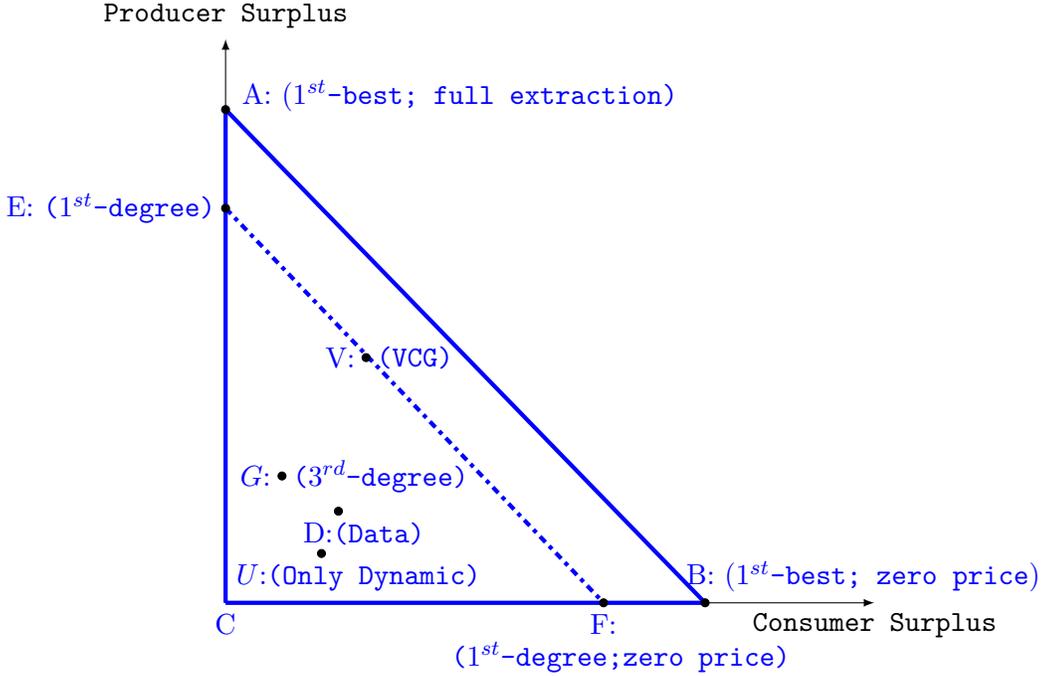
\begin{figure}[ht!]
\begin{center}
\caption{\em Welfare Triangle \label{fig:welfare}}
 \begin{tikzpicture}[>=latex, y=2.5cm, font=\small, scale =0.75]
 \draw [<->] (0,4) node [above] {\texttt{Producer Surplus}} |- (11.5,0) node [below] {\texttt{Consumer Surplus}}; 
 \draw [blue, ultra thick] (0,0) -| (0,3.5); 
 \filldraw (0.1,3.6)node[right] {\blue A: (\texttt{$1^{st}$-best; full extraction)}}; 
 \draw [blue, ultra thick] (0,0) -| (8.5,0); 
 \draw [blue, ultra thick] (0, 3.5) -- (8.5,0); 
 \draw (11.3,0.001) node [above]{\blue B: (\texttt{$1^{st}$-best; zero price})}; 
 \filldraw (0,3.5) circle (2pt); 
 \filldraw (8.5,0) circle (2pt); 
 \filldraw (2,.65) circle (2pt); 
 \filldraw (2.45,.65) node [below]{\blue D:\texttt{(Data)}}; 
 \draw [blue, ultra thick,dash dot] (0,2.8) -- (6.7,0); 
 \filldraw (6.7,0) circle (2pt) node[below] {\blue F:}; 
 \filldraw (7,-0.2) node[below] {\blue \texttt{($1^{st}$-degree;zero price)}}; 
 \filldraw (0,2.8) circle (2pt) node [left]{\blue E: \texttt{($1^{st}$-degree)}}; 
 \filldraw (1,0.9) circle (2pt)node [left]{\blue $G$:}node [right]{\blue \texttt{($3^{rd}$-degree)}}; 
 \filldraw (2.49,1.74)circle (2pt) node[left]{\blue V:}node[right]{\blue \texttt{(VCG)}};
 \filldraw (0,0) node[below]{\blue C};
 \filldraw (1.7,.35) circle (2pt); 
 \filldraw (2.35,.35) node [below]{\blue $U$:\texttt{(Only Dynamic)}}; 
 \end{tikzpicture}
 \end{center}
 \begin{figurenotes}
A schematic representation of the welfare triangle under a different pricing regime.\end{figurenotes}
\end{figure}

Point D (``Data") in Figure \ref{fig:welfare}, which is in the interior of the triangle CAB, denotes the division of surplus resulting from current pricing practices by airlines that we observe in our data and that we use to estimate our model. 
This outcome is preferable to no trade by the airline and consumers, but it is strictly inside the welfare frontier due to the two sources of inefficiency discussed above. The distance of D from the line $\overline{\text {AB}}$ illustrates the magnitude of welfare-improving opportunities relative to current practice.

Suppose the airline did not employ second-degree price discrimination each period before the flight and only changed prices across time. Or, in other words, the airline did not exploit the difference in quality between an economy seat and a first-class seat when choosing its prices, and each period chose one price for both cabins.\footnote{Formally, airline solves $\max_{\{p_t, q_t\}_{t=1}^T}\sum_{t=1}^T\E_t \left\{ \pi (\chi_{t},\omega_{t};\Psi_{t})\right\}$, where the flow profit is $\pi(\chi_{t},\omega_{t};\Psi_{t}) = (p_{t}-c^{f})\times\E(q_{t}^{f}|(p_t, q_t))+ (p_{t}-c^{e})\times\E(q_{t}^{e}|(p_t, q_t))$, and $\E(q_{t}^{f}|(p_t, q_t))$ and $\E(q_{t}^{e}|(p_t, q_t))$ are expected demands.}
The airline would still adjust the price each period depending on the opportunity cost of selling a seat in that period. This counterfactual of choosing one price across cabins 
corresponds to point {U} in Figure \ref{fig:welfare} (``Only Dynamic''), i.e., uniform pricing within each period.
While the producer surplus under {U} will be lower than the producer surplus under D, the effect on consumer surplus is theoretically ambiguous.
Choosing one price across two cabins should improve welfare for those who buy first-class under the current prices, but it should lower welfare because economy-class seats become expensive and total sales will adjust.

Airlines do not observe passengers' reasons for travel, limiting their ability to price based on the difference between business and leisure passengers' willingness-to-pay, possibly resulting in the exclusion of leisure passengers based on expectations of greater demand from business passengers. Permitting the airline to price based on the reason for travel, i.e., third-degree price discrimination, can increase profits for the airline, but the implication for passengers is ambiguous. Leisure passengers may benefit, but it may come at the cost of business passengers. Since leisure and business travelers arrive at different times, and the airline faces capacity concerns, the change in consumer surplus depends on the entire demand process. Furthermore, the total number of seats sold may increase or decrease, affecting the total welfare. Point {G} (``$3^{rd}$-degree"), i.e., group pricing, in Figure \ref{fig:welfare} represents a division of welfare when the airline can charge different prices based on a passenger's reason for travel and one seat-release policy for each cabin. In particular, the airline sets four prices every period depending on the reason to travel and seat class. However, we impose that airlines can choose only one seat-release policy for each ticket class.

Even with the airline's ability to price based on the reason for travel, asymmetric information about idiosyncratic valuations can create inefficiencies. For example, some leisure passengers may have unusually high valuations, and some business travelers may have low valuations. To ascertain the importance of this information asymmetry, we consider a setting where the airline practices first-degree price discrimination. The airline observes valuations each period and decides which passengers to accommodate, charging each passenger their valuation. However, the airline is still uncertain about future demand realizations. This outcome corresponds to point E in Figure \ref{fig:welfare}. Likewise, point F in Figure \ref{fig:welfare} corresponds to the first-degree allocation of seats but with the price equal to the peanut cost.

The (dotted) line that joins E and F is informative about the extent of dynamic inefficiency in the market. In particular, line $\overline{\text {EF}}$ represents the frontier in the welfare triangle (CEF) when the airline knows $(v,\xi)$ for passengers in a given period but cannot foresee future realizations of the demand process. One way to divide the surplus along the $\overline{\text {EF}}$ frontier is by implementing Vickery-Clarke-Groves auctions every period. The division of the per-period surplus depends on the composition of passengers. All else equal, the larger the number of passengers with a high willingness-to-pay, the higher the ``VCG price" and closer to the y-axis, and vice versa. Such a division of surplus averaged across all periods is denoted by point V (``VCG") in Figure \ref{fig:welfare}. Thus, the potential outcomes in CAB but not CEF represent surplus lost due to inter-temporal demand uncertainty. One could envision a secondary-market run by the airline to resolve these dynamic inefficiencies, and our estimates provide the value that such an exchange could create.

\subsection{Counterfactual Results}
Table \ref{tab:counterfactuals} presents the welfare estimates under these alternative pricing strategies, averaged across all markets and capacities observed in our data. Columns of Table \ref{tab:counterfactuals} are indexed by a letter (e.g., D, E, U) that corresponds to the point in the welfare triangle of Figure \ref{fig:welfare}.

\paragraph{Current Pricing.}
Recall that point D in Figure \ref{fig:welfare} denotes the division of surplus resulting from the airlines' current pricing practices.
The surplus associated with this pricing strategy is in the first column of Table \ref{tab:counterfactuals}.
The average total surplus is \$60,739 per flight (with a standard error of \$4,713), with 76\% of the surplus going to the airline.\footnote{To compute the standard errors, we follow the bootstrap procedure described in Section \ref{results}. To compute the bootstrapped samples of counterfactual welfare, we solve the counterfactual outcomes for only those 1,889 market types, out of the 10,000 types, with an implied weight greater than $10^{-8}$. Although the set of types with the implied weight of greater than $10^{-8}$ may be different for each bootstrap sample, for computational tractability, while computing the standard errors of the surpluses, we hold the 1,889 types fixed and re-weight the types implied in each bootstrap sample.} We find that, on average, the consumer and producer surplus per seat sold is approximately \$144.883 and \$465.091, respectively.
This outcome is preferable to no trade for both the airline and consumers but is strictly inside the welfare frontier due to the inefficiencies, as mentioned earlier.
Comparing the total surpluses under columns D and $\overline{\text {AB}}$ in Table \ref{tab:counterfactuals}, we find that the surplus associated with current pricing represents 77\% of the average market's potential attainable surplus.

\newcolumntype{g}{>{\columncolor{Gray}}c}
\begin{table}[ht!]
\centering
\caption{Price Discrimination Counterfactuals}
\label{tab:counterfactuals}
\begin{adjustbox}{scale=0.9}
\begin{threeparttable}
\begin{tabular}{lcgcccccccccc}
\toprule
 & Only Dynamic& Data & $3^{rd}$-Degree & $1^{st}$-Degree & VCG & $1^{st}$-Best \\
 & {U} & D & {G} & E & V & $\overline{\text{A}\text{B}}$ \\
 \midrule
Producer Surplus & 43,228 & 46,594 & 49,460 & 76,591 & 44,479 & - \\
 & (3,780) & (3,904) & (3,903) & (5,853) & (4,958) & \\
 -- Economy Class & 34,595 & 30,927 & 33,075 & - & - & - \\
 & (2,851) & (2,919) & (2,875) & & & \\
Consumer Surplus & 17,191 & 14,145 & 13,729 & - & 32,111 & 78,885 \\
 & (1,570) & (1,192) & (1,200) & & (2,467) & (6,123) \\
 -- Business & 8,259 & 5,789 & 4,476 & -& 11,413 & 27,620 \\
 & (1,393) & (1,008) & (741) & & (1,979) & (3,372) \\
 -- Leisure & 8,931 & 8,355 & 9,253 & -& 20,698 & 51,265 \\
 & (1,338) & (1,032) & (1,012) && (2,076) & (5,981) \\
Total Surplus & 60,419 & 60,739 & 63,190 & 76,591 & 76,591 & 78,885 \\
 & (4,793) & (4,713) & (4,812) & (5,853) & (5,853) & (6,123) \\
\bottomrule
\end{tabular}
\begin{figurenotes}
In this table, we present measures of average welfare (expressed in U.S. dollars) for six different outcomes, corresponding to points in Figure \ref{fig:welfare}, where Column D is the data. These calculations are performed for all market types receiving positive weight for a given capacity, then averaged across types and capacities. Bootstrapped standard errors are in parentheses.
\end{figurenotes}
\end{threeparttable}
\end{adjustbox}
\label{table:cf}
\end{table}

We can also determine when the efficiency losses occur under the current pricing. 
Inefficiency increases with the number of incorrectly-allocated seats sold. For example, much of the loss will be when seats are sold early, and there are positive shocks to demand at later dates. Then, passengers who arrive late feel these losses. Similarly, if too few seats are released early and later demand is lower than expected, the early passengers are relatively worse off. 

\begin{table}[ht!]
\begin{center}
\caption{Efficiency Losses, by Time}\label{table:dynamic_ineff}
\begin{adjustbox}{scale=0.90}
\begin{tabular}{lllllllll}
\toprule
Days before Flight & 101+ & 61-100& 45-60& 30-44 & 15-29 & 8-14 & 4-7 & 0-3 \\
\midrule
Inefficiency(\%) & -12.69 & -4.2 & 9.57 & 13.56 & 18.16 & 23.7 & 25.5 & 26.4 \\
\bottomrule
\end{tabular}
\end{adjustbox}
\begin{figurenotes}
This table presents the average efficiency losses for each period (rounded up to a near decimal point). The inefficiency per period is defined for each market and each capacity. It is expressed as a percentage of total inefficiency and then averaged across markets and initial capacities. The exact formula to calculate this measure is given in the text.
\end{figurenotes}
\end{center}
\end{table}

In Table \ref{table:dynamic_ineff}, we present the time-specific inefficiencies for each market and then average them across markets and capacities. Efficiency losses occur early and decrease as the flight approaches, suggesting that the airline sells too many early seats than under the $1^{st}$-best allocation. For example, passengers who arrive in the first period (i.e., 101+ days before the flight) have 12.69\% more surplus under the current pricing than under the first-best allocation. In contrast, those at the end get 26.4\% less surplus under the current pricing than under the first-best.

\paragraph{Only-Dynamic Pricing.} Restricting the airline to only one price each period for both cabins, as expected, slightly lowers the total welfare. Note that the airline can still change prices over time. The producer surplus is 92\% of the baseline (see columns D and U), and the consumer surplus is 121\% of the baseline. 
Although the airline's ability to use second-degree price discrimination to screen passengers between cabins (from {U} to C) has a negligible effect on total surplus, airlines capture all the additional surplus. The surplus for business travelers under uniform pricing is larger than under current pricing because, on average, prices are lower under uniform pricing.

\paragraph{Group Pricing: Business versus Leisure.} Column {\red G} of Table \ref{tab:counterfactuals} provides surplus estimates when the airlines are permitted to price based on the reason for travel, i.e., third-degree price discrimination. Relative to current pricing practice (i.e., Column D), airline surplus increases by six percent, but consumer surplus falls by about three percent and the total surplus increases by four percent. Thus, group pricing based on business or leisure increases the revenue and total surplus slightly but lowers consumer surplus.

\paragraph{Static versus Dynamic Inefficiencies.} We begin with the surplus under the first-degree price discrimination in column E of Table \ref{tab:counterfactuals}, where the airline can price equal to the passengers' willingness-to-pay but still faces uncertain future demand. By construction, the airline can capture the entire surplus. The total surplus increases a lot relative to the second-and third-degree price discrimination (D and {G} of Table \ref{tab:counterfactuals}, respectively).
We find that a VCG auction would result in relatively low prices or that consumers would capture 41\% of the surplus in the presence of a period-by-period VCG auction.
Thus, these results suggest that the business-leisure distinction increases the total surplus.

Comparing total surpluses under D and $\overline{\text{AB}}$, we find that stochastic demand and asymmetric information lead to approximately 23\% loss of welfare. Comparing D to E and E to $\overline{\text{AB}}$, we find that almost all of this inefficiency is due to private information about passengers' willingness-to-pay and not to stochastic demand.

\paragraph{Inefficiencies and Market Variance.}
The welfare measures reported thus far correspond to the average flight in our sample, but welfare and inefficiency may vary with the demand uncertainty an airline faces on a route. 
In other words, airlines would be less able to use market-level information to set prices and quantities for markets if the variance of demand primitives is high. So, we would expect the dynamic inefficiency to increase with variance. 
To shed light on this, we calculate average welfare for high- and low-variance markets separately. We classify a market as a low (respectively, high) variance market if $\{cv^{\ell}, cv^b, \Delta^\theta\}$ are \emph{all} less than (larger than) the averages of $\{0.19, 0.73, 0.05\}$ as shown in Table \ref{demandhet0}. 

As expected, the average total $1^{st}$-best welfare for markets with high variance at 113,130 is larger than that for the markets with low variance at 80,212. 
Comparing C with the AB line, we find the inefficiencies for the low- and high-variance markets are 77\% and 74\%, respectively. Furthermore, the private information versus dynamic inefficiency breakdown for these two types of markets is 87.4\% vs. 12.6\% and 86.4\% vs. 13.6\%, respectively. These estimates suggest that the dynamic inefficiency increases with market variance, albeit only modestly, consistent with our finding that overall inefficiency due to dynamic demand uncertainty is relatively smaller than the inefficiency due to private information.
       
\section{Conclusion\label{conclusion}}

We develop a model of intra-temporal and inter-temporal price discrimination by airlines that sell a fixed number of seats of different quality to heterogeneous consumers arriving before a flight. We specify demand as non-stationary and stochastic, accommodating the salient features of airline pricing. Using unique data from international airline markets, we flexibly estimate the distribution of preferences for flights. The estimation exploits the relationship between a passenger's seat chosen, the timing of purchases, reasons for travel, and the fare paid to identify how effectively airlines discriminate using sources of passenger heterogeneity. We find that the flexibility of the model and estimation algorithm successfully captures key features of our data, such as the price paths for both economy and first-class tickets, the number of purchases, and the differences in fares of the two cabins across time.

Next, through several counterfactual exercises, we use the estimates to explore the role stochastic demand and asymmetric information have on efficiency and the distribution of surplus. Because of private information about demand, there are substantial inefficiencies relative to the first-best outcome under the current pricing practices. In particular, total welfare is only 77\% of the welfare without demand uncertainty and asymmetric information. To isolate the role of different sources of asymmetric information in determining welfare, we solve for optimal seat release and prices when the airline can discriminate based on passengers' reason to travel and when the airline can observe their preferences. The first case (i.e., third-degree price discrimination) achieves 80\% of the first-best welfare, representing a 3\% increase from current practices. Business passengers' surpluses decrease, but leisure passengers' surpluses increase due to the business passengers' loss of informational rent. The second case (i.e., first-degree price discrimination), where the only remaining source of inefficiency is inter-temporal demand uncertainty, significantly affects welfare compared to the first case. Thus, private information about demand accounts for 87\% of the total welfare loss, while dynamic uncertainty about demand accounts for the rest.

There are many avenues for future research on related topics. One interesting topic for future research is to consider competition. Like other studies of dynamic pricing, we model a monopolistic market structure that accurately reflects our data. This modeling choice limits our ability to examine the impact of competition on discriminatory-pricing practices. Another interesting path for future research is considering the possibility that consumers are forward-looking and strategic in their purchasing decisions. Furthermore, purchases are increasingly made online, allowing firms to track search behavior and adapt pricing accordingly. Given the growing theoretical literature on this topic (e.g., \cite{BoardSkrzypacz2016} and \cite{DilmeLi2012}) that yield testable implications from strategic behavior by consumers, empirical studies like ours and \cite{Sweeting2010} represent an opportunity to offer insight to future modeling efforts.

As firms gather more information about consumers' preferences and purchasing habits, exploiting this information becomes a paramount concern. For more on the role of privacy and efficiency, see \cite{Hirshleifer1971} and \cite{Posner1981}. Although few papers study the role of privacy \citep{Taylor2004, CalzolariPavan2006} in price discrimination, more empirical research on this topic is needed. Future research should study the trade-off between efficiency and privacy, especially in industries with greater access to such information.

\section*{\bf Data Availability Statement.} 
U.S. Department of Commerce and Department of Transportation provided the SIAT and the T100 data underlying this article. Additional data on capacities were purchased from OAG. We do not have permission to share the data. However, SIAT and OAG data can be purchased and T100 data accessed by following the steps provided in the replication folder available in Zenodo at \url{https://zenodo.org/record/7392123}.

\bibliographystyle{jpe}
\bibliography{dispersion}

\setcounter{section}{0}
\setcounter{equation}{0}
\setcounter{figure}{0}
\setcounter{table}{0}
\begin{center}
\noindent\textbf{\huge{Appendix}}
\end{center}
\renewcommand{\thesection}{A.\arabic{section}}
\renewcommand{\theequation}{A.\arabic{equation}}
\renewcommand\thefigure{\thesection.\arabic{figure}}
\renewcommand\thetable{\thesection.\arabic{table}}



\section{Survey of International Air Travel\label{siat}} 

We present additional details about the SIAT data. 
As we mentioned in Section \ref{Data}, the data are collected by the Department of Commerce. The DOC contracts with a private survey firm, CIC Research Inc. We use data from the surveys conducted in 2009, 2010, and 2011. There are two data collection methods: (1) Direct participation of the airlines, which arrange for their flight crews to distribute and collect surveys on-board; (2) Use of sub-contractors to distribute and collect the questionnaires in the airport departure gate area. According to the SIAT, in 2009, these two methods accounted for 60\% and 40\% of all collections, respectively. The dataset can be purchased at \url{https://rb.gy/fop8cc}. 
A copy of the survey questionnaire is available at \url{http://charliemurry.github.io/files/SIAT_Data_Doc_2009.pdf}.\footnote{DOC processes the surveys and implements a quality control process to ensure data integrity.}

There are 413,309 survey responses in the data from the Department of Commerce.
We impose many restrictions to arrive at our final sample. In Table \ref{table:summarystat}, we display summary statistics at four stages of the sample selection process: (1) the original data, (2) after we drop responses that do not report a price, (3) after we make additional selection criteria, like dropping flights with less than ten responses, responses with other partial information, non-revenue and other exotic tickets, and connecting tickets and (4) our final sample after we select monopoly markets and merge with auxiliary data on capacities. Approximately 38\% of the initial survey responses do not have information about fares, and we drop those. Out of the remaining 62\% who report fares, approximately 53\% report traveling with at least one companion. If multiple people are traveling together, e.g., a family, the survey is intended to be administered to one group member. When a respondent reports flying with other passengers, we duplicate the ticket data for each passenger they report flying with them. We exclude respondents who report buying their tickets as a part of a tour package, using airline miles, or through any other discounted fare. We also restrict our sample to responses that report traveling with at most ten people (which is 98.23\% of the original sample) to minimize the chances that the tickets were bought as part of some tour package.

\begin{table}[htbp]
\hspace{-1in}
\caption{Summary Statistics\label{table:summarystat}}
\scalebox{0.75}{
\begin{threeparttable}
\begin{tabular}{l ccccc |c ccccc}
\toprule
Variable & Mean & Min & Median & Max & Non-missing && Mean & Min & Median & Max & Non-missing\\
\midrule
 & \multicolumn{5}{c}{(1) Raw Data} & \multicolumn{5}{c}{(2) Non-missing Fares}  \\
\midrule
Business Passenger         &      0.17&      --  &      --  &      --  & 413,309&&     0.19&      --  &      --  &      --  & 256,843\\
-- business                 &      0.04&      --  &      --  &      --  & 413,309&&     0.04&      --  &      --  &      --  & 256,843\\
-- convention              &      0.02&      --  &      --  &      --  & 413,309&&     0.02&      --  &      --  &      --  & 256,843\\
-- government              &      0.00&      --  &      --  &      --  & 413,309&&     0.00&      --  &      --  &      --  & 256,843\\
-- teaching                &      0.02&      --  &      --  &      --  & 413,309&&     0.02&      --  &      --  &      --  & 256,843\\
-- vacation                &      0.17&      --  &      --  &      --  & 413,309&&     0.20&      --  &      --  &      --  & 256,843\\
-- visit friends/family    &      0.12&      --  &      --  &      --  & 413,309&&     0.15&      --  &      --  &      --  & 256,843\\
-- religion                &      0.01&      --  &      --  &      --  & 413,309&&     0.01&      --  &      --  &      --  & 256,843\\
-- health                  &      0.01&      --  &      --  &      --  & 413,309&&     0.01&      --  &      --  &      --  & 256,843\\
Nonstop                    &      0.49&      --  &      --  &      --  & 413,309&&     0.46&      --  &      --  &      --  & 256,843\\
Economy Class              &      0.90&      --  &      --  &      --  & 413,309&&     0.90&      --  &      --  &      --  & 256,843\\
Advanced Purchase          &     72.37&      0.00&     45.00&    365.00& 413,309&&    67.39&      0.00&     45.00&    365.00& 256,843\\
Fare (on-way adjusted)     &    567.97&      0.00&    400.00& 24,258.00& 256,843&&   567.97&      0.00&    400.00&  24,258.00& 256,843\\
Total travelers            &      2.98&      1.00&      2.00&     10.00& 413,309&&     2.80&      1.00&      2.00&     10.00& 256,843\\
Trav. with Family*\dotfill  &      0.36&      --  &      --  &      --  & 413,309&&      0.34&      --  &      --  &      --  & 256,843\\
US Resident*         &      0.48&      --  &      --  &      --  & 413,309&&      0.51&      --  &      --  &      --  & 256,843\\
Male*               &      1.54&      1.00&      2.00&      2.00& 363,586&&      1.56&      1.00&      2.00&      2.00& 236,658\\
Age*                 &     43.01&     18.00&     42.00&     99.00& 234,546&&     43.10&     18.00&     42.00&     92.00& 153,092\\
Income Bin*          &      5.68&      1.00&      5.00&     11.00& 305,647&&      5.77&      1.00&      5.00&     11.00& 213,284\\
\midrule
 & \multicolumn{5}{c}{(3) Selection On Additional Variables} & \multicolumn{5}{c}{(4) Final Sample}  \\
\midrule
Business Passenger         &      0.16&      --  &      --  &      --  &  25,945&&       0.14&      --  &      --  &      --  &  14,930\\
-- business                 &      0.04&      --  &      --  &      --  &  25,945&&       0.03&      --  &      --  &      --  &  14,930\\
-- convention              &      0.02&      --  &      --  &      --  &  25,945&&       0.02&      --  &      --  &      --  &  14,930\\
-- government              &      0.00&      --  &      --  &      --  &  25,945&&       0.00&      --  &      --  &      --  &  14,930\\
-- teaching                &      0.02&      --  &      --  &      --  &  25,945&&       0.01&      --  &      --  &      --  &  14,930\\
-- vacation                &      0.19&      --  &      --  &      --  &  25,945&&       0.17&      --  &      --  &      --  &  14,930\\
-- visit friends/family    &      0.14&      --  &      --  &      --  &  25,945&&       0.13&      --  &      --  &      --  &  14,930\\
-- religion                &      0.01&      --  &      --  &      --  &  25,945&&       0.00&      --  &      --  &      --  &  14,930\\
-- health                  &      0.01&      --  &      --  &      --  &  25,945&&       0.01&      --  &      --  &      --  &  14,930\\
Nonstop                    &      0.58&      --  &      --  &      --  &  25,945&&       1.00&      --  &      --  &      --  &  14,930\\
Economy Class              &      0.93&      --  &      --  &      --  &  25,945&&       0.92&      --  &      --  &      --  &  14,930\\
Advanced Purchase          &     73.65&      0.00&     60.00&    365.00&  25,945&&      75.25&      0.00&     60.00&    365.00&  14,930\\
Fare (on-way adjusted)     &    528.43&     70.00&    400.00&   5000.00&  25,945&&     480.37&     70.00&    375.00&   5000.00&  14,930\\
Total travelers            &      2.90&      1.00&      2.00&     10.00&  25,945&&       3.06&      1.00&      2.00&     10.00&  14,930\\
Travel with Family*  &      0.38&      --  &      --  &      --  &  25,945&&      0.42&      --  &      --  &      --  &  14,930\\
US Resident*         &      0.51&      --  &      --  &      --  &  25,945&&      0.47&      --  &      --  &      --  &  14,930\\
Male*                &      1.54&      1.00&      2.00&      2.00&  23,698&&      1.54&      1.00&      2.00&      2.00&  13,604\\
Age*                &     43.67&     18.00&     43.00&     88.00&  16,209&&     43.34&     18.00&     42.00&     88.00&   9,434\\
Income Bin*         &      5.69&      1.00&      5.00&     11.00&  21,247&&      5.70&      1.00&      5.00&     11.00&  12,011\\
\bottomrule\end{tabular}
\begin{tablenotes}[flushleft]
\item \textit{Note:} In this table, we present the summary statistics of variables that we observe from the SIAT survey. Variables with an asterisk (*) denote the variables not used in our empirical analysis. The table is divided into four sub-tables, each of which displays the summary statistics after each round of sample selection. Sub-table (1) denotes the original data; Sub-table (2) denotes the sample after we have dropped any observation that does not report fares; Sub-table (3) denotes the sample after we drop flights with less than ten responses, responses with partial information for our purposes, and non-revenue tickets, and connecting tickets; and sub-table (4) denotes the sample after we restrict to monopoly markets.
\end{tablenotes}
\end{threeparttable}
}
\end{table}

\setcounter{equation}{0}
\setcounter{figure}{0}
\setcounter{table}{0}
\renewcommand{\thesection}{A.\arabic{section}}
\renewcommand{\theequation}{A.\arabic{equation}}
\renewcommand\thefigure{\thesection.\arabic{figure}}
\renewcommand\thetable{\thesection.\arabic{table}}


\section{Additional Estimation Results\label{appendix1}}

In this section, we provide additional results referenced in the main text. In the main text, we present the results integrated over the distribution of aircraft capacities shown in Table 5. To construct that grand average, we estimate parameters at 20 distinct initial capacities and then take a weighted average over the parameter estimates using the capacity weights. In Tables \ref{table:appendix2} and \ref{table:appendix1} we present the parameter point estimates capacity by capacity, for completeness. Note that each capacity level represents different flights in the data, so we should expect variation across capacities. However, there are no clear patterns across capacities for many of the parameters, which makes sense if the capacity is not a function of those aspects of the data-generating process. We estimate that the number of passenger arrivals (for both nonstop and connecting) tends to be higher at higher capacity levels.

\begin{landscape}
\begin{table}
  \centering
  \caption{Estimated Mean Values of $\Psi$, by Capacities\label{table:appendix2}}
\begin{adjustbox}{scale=0.65}
\begin{tabular}{l ccccccccccccccccccccc}
\toprule
Parameters${\backslash}$Capacity Index   & (1) & (2) & (3) & (4) & (5) & (6) & (7) & (8) & (9) & (10) & (11) & (12) & (13) & (14) & (15) & (16) & (17) & (18) & (19) & (20) \\
\midrule
     $\mu^l$                    &338.87  &  249.04  &  453.35  &  458.80  &  310.15  &  255.63  &  427.86  &  445.68  &  248.54  &  591.49  &  561.36  &  484.82  &  507.55  &  469.35  &  388.38  &  433.15  &  342.28  &  454.47  &  399.94  &  263.81  &  \\ 
    $\frac{\sigma^l}{\mu^l}$    &0.13  &  0.49  &  0.14  &  0.09  &  0.07  &  0.05  &  0.32  &  0.16  &  0.22  &  0.12  &  0.07  &  0.17  &  0.10  &  0.06  &  0.10  &  0.09  &  0.11  &  0.31  &  0.07  &  0.19  &  \\ 
    $\delta^b$                  &0.26  &  0.09  &  0.30  &  0.22  &  0.11  &  0.37  &  0.22  &  0.08  &  0.13  &  0.08  &  0.07  &  0.19  &  0.08  &  0.19  &  0.14  &  0.08  &  0.17  &  0.05  &  0.24  &  0.37  &  \\ 
    $\frac{\sigma^b}{\mu^b}$    &0.52  &  0.20  &  0.14  &  0.48  &  0.33  &  0.20  &  0.27  &  0.64  &  0.30  &  0.98  &  0.48  &  0.48  &  0.89  &  0.26  &  0.67  &  0.67  &  0.33  &  0.84  &  0.24  &  0.73  &  \\ 
    $\mu^{\xi}$                 &0.40  &  0.62  &  0.29  &  0.36  &  0.34  &  0.25  &  0.31  &  0.46  &  0.27  &  0.64  &  0.62  &  0.91  &  0.53  &  0.42  &  1.00  &  0.29  &  0.32  &  0.88  &  0.51  &  0.55  &  \\ 
    $\lambda$                   &10.18  &  29.68  &  19.11  &  18.48  &  24.82  &  30.04  &  21.26  &  19.09  &  23.68  &  32.33  &  22.39  &  44.81  &  42.46  &  32.37  &  43.64  &  41.24  &  41.13  &  24.61  &  37.22  &  36.58  &  \\ 
    $\Delta^{\bf a}$            &-0.06  &  -0.04  &  -0.05  &  -0.06  &  -0.05  &  -0.04  &  -0.08  &  -0.04  &  -0.05  &  -0.08  &  0.05  &  -0.06  &  -0.09  &  -0.10  &  -0.05  &  -0.10  &  -0.02  &  -0.01  &  -0.04  &  -0.04  &  \\ 
    $\Delta^{\theta}$           &0.06  &  0.10  &  0.04  &  0.05  &  0.04  &  0.05  &  0.04  &  0.04  &  0.05  &  0.03  &  0.11  &  0.04  &  0.09  &  0.06  &  0.03  &  0.04  &  0.06  &  0.10  &  0.03  &  0.05  &  \\ 
    $\gamma$                    &29.24  &  27.41  &  32.50  &  18.16  &  28.40  &  38.19  &  22.19  &  30.77  &  36.84  &  48.55  &  54.28  &  57.32  &  57.19  &  38.31  &  57.93  &  55.91  &  56.73  &  57.70  &  58.51  &  59.53  &  \\ 
\bottomrule
\end{tabular}
\end{adjustbox}
\begin{figurenotes}In this table we present the mean demand parameters from the conditional distribution $h(\Psi|\omega_1;\widehat{\mu}_{\Psi},\widehat{\Sigma}_{\Psi})$, given the initial-capacity $\omega_1$. We estimated the model for 20 different capacities, and each capacity index corresponds to one such capacity.\end{figurenotes}
\end{table}

\begin{table}
  \centering
  \caption{Estimated Mode Values of $\Psi$, by Capacities \label{table:appendix1}}
\begin{adjustbox}{scale=0.65}
\begin{tabular}{l ccccccccccccccccccccc}
\toprule
Parameters${\backslash}$ Capacity Index   & (1) & (2) & (3) & (4) & (5) & (6) & (7) & (8) & (9) & (10) & (11) & (12) & (13) & (14) & (15) & (16) & (17) & (18) & (19) & (20) \\
\midrule
     $\mu^l$                    &415.02  &  352.42  &  470.62  &  405.89  &  397.85  &  429.13  &  397.18  &  427.69  &  433.23  &  504.16  &  518.59  &  527.39  &  510.85  &  450.46  &  510.52  &  499.05  &  453.19  &  472.65  &  424.85  &  448.82  &  \\ 
    $\frac{\sigma^l}{\mu^l}$    &0.13  &  0.48  &  0.11  &  0.09  &  0.10  &  0.09  &  0.19  &  0.17  &  0.20  &  0.15  &  0.10  &  0.18  &  0.12  &  0.08  &  0.13  &  0.11  &  0.17  &  0.34  &  0.07  &  0.19  &  \\ 
    $\delta^b$                  &0.36  &  0.15  &  0.29  &  0.29  &  0.13  &  0.38  &  0.27  &  0.13  &  0.20  &  0.11  &  0.08  &  0.20  &  0.10  &  0.26  &  0.18  &  0.11  &  0.15  &  0.10  &  0.27  &  0.37  &  \\ 
    $\frac{\sigma^b}{\mu^b}$    &0.69  &  0.18  &  0.25  &  0.53  &  0.36  &  0.28  &  0.34  &  0.65  &  0.38  &  0.53  &  0.45  &  0.51  &  0.78  &  0.25  &  0.52  &  0.68  &  0.38  &  0.60  &  0.27  &  0.68  &  \\ 
    $\mu^{\xi}$                 &0.54  &  0.37  &  0.41  &  0.28  &  0.35  &  0.27  &  0.31  &  0.49  &  0.37  &  0.56  &  0.37  &  0.50  &  0.49  &  0.50  &  0.52  &  0.36  &  0.44  &  0.50  &  0.51  &  0.59  &  \\ 
    $\lambda$                   &18.36  &  29.20  &  26.17  &  19.62  &  26.51  &  25.43  &  20.90  &  22.29  &  24.40  &  30.83  &  31.13  &  19.46  &  32.45  &  29.37  &  24.56  &  35.98  &  32.87  &  26.65  &  33.85  &  32.33  &  \\ 
    $\Delta^{\bf a}$            &-0.05  &  -0.05  &  -0.05  &  -0.06  &  -0.04  &  -0.04  &  -0.08  &  -0.04  &  -0.05  &  -0.08  &  0.03  &  -0.08  &  -0.09  &  -0.02  &  -0.03  &  -0.04  &  -0.00  &  -0.01  &  -0.04  &  -0.02  &  \\ 
    $\Delta^{\theta}$           &0.04  &  0.09  &  0.04  &  0.04  &  0.04  &  0.06  &  0.05  &  0.04  &  0.06  &  0.02  &  0.11  &  0.04  &  0.04  &  0.06  &  0.03  &  0.04  &  0.06  &  0.09  &  0.03  &  0.05  &  \\ 
    $\gamma$                    &29.16  &  27.58  &  28.28  &  19.49  &  37.48  &  31.10  &  22.20  &  30.20  &  36.74  &  49.55  &  53.89  &  50.76  &  40.22  &  35.74  &  55.26  &  53.08  &  54.54  &  49.15  &  39.12  &  46.53  &  \\ 
\bottomrule
\end{tabular}
\end{adjustbox}
\begin{figurenotes}In this table we present the mode of the demand parameters from the conditional distribution $h(\Psi|\omega_1;\widehat{\mu}_{\Psi},\widehat{\Sigma}_{\Psi})$, given the initial-capacity $\omega_1$. We estimated the model for 20 different capacities, and each capacity index corresponds to one such capacity.
\end{figurenotes}
\end{table}

\end{landscape}

\section{Model Fit\label{sec:fit}}

In Figures \ref{fig:fit1}, \ref{fig:fit2}, we display the empirical moments and the model-implied moments evaluated at the estimated parameters, of prices, by time, for economy seats and first-class seats, respectively. The moments take the form of deciles of the cumulative density functions of the data and the model predictions. 

\begin{figure}[h!]
\begin{center}
  \caption{Model Fit: Economy Fares} \label{fig:fit1}
  \includegraphics[scale=0.4]{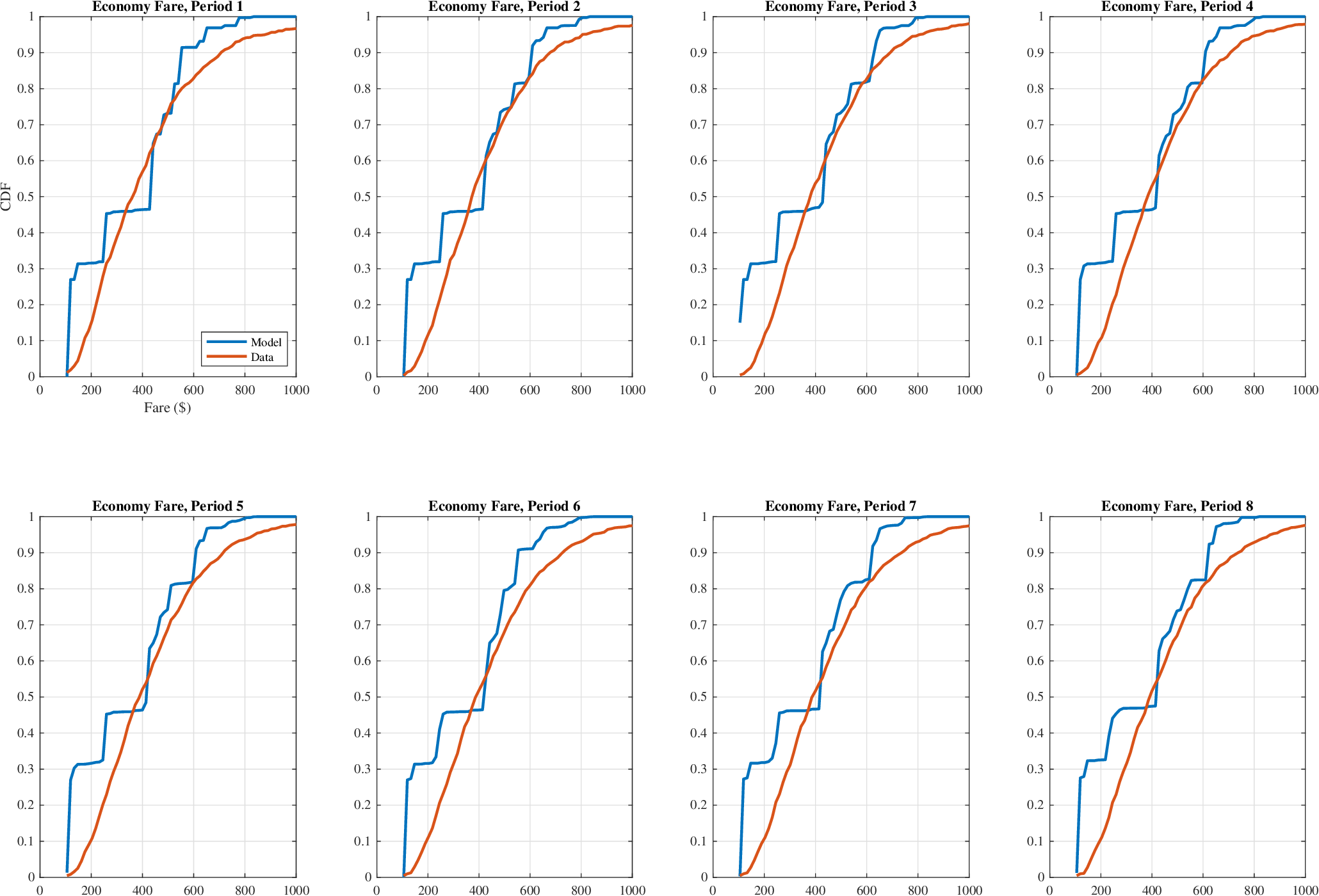}
  \end{center}
  \begin{figurenotes}
  
  This figure displays the distributions of economy fares according to the model (blue) for the smaller modal capacity of 160 economy seats and 18 first-class seats and the distributions of economy fares in the data (red) by time. 
  \end{figurenotes}
\end{figure}

The deciles from our data are shown in red, and the deciles predicted by the model are in blue. In our estimation step, for each initial capacity, we seek to match 620 moments to determine weights, i.e., the conditional density $h(\Psi|\omega_1)$ in Equation \ref{eq:mixture} for each period.\footnote{We use these 20 conditional densities to generate average welfare for our counterfactuals in Table \ref{tab:counterfactuals}. } Here, we display the fit for the ``smaller" modal capacity--160 economy seats and 18 first-class seats. Furthermore, in Figure \ref{fig:fit3}, we display the empirical CDF of the share of business passengers and the model-implied share of business passengers evaluated at the estimated parameters. 

\begin{figure}[h!]
\begin{center}
  \caption{Model Fit: First-Class Fares}\label{fig:fit2}
  \includegraphics[scale=0.4]{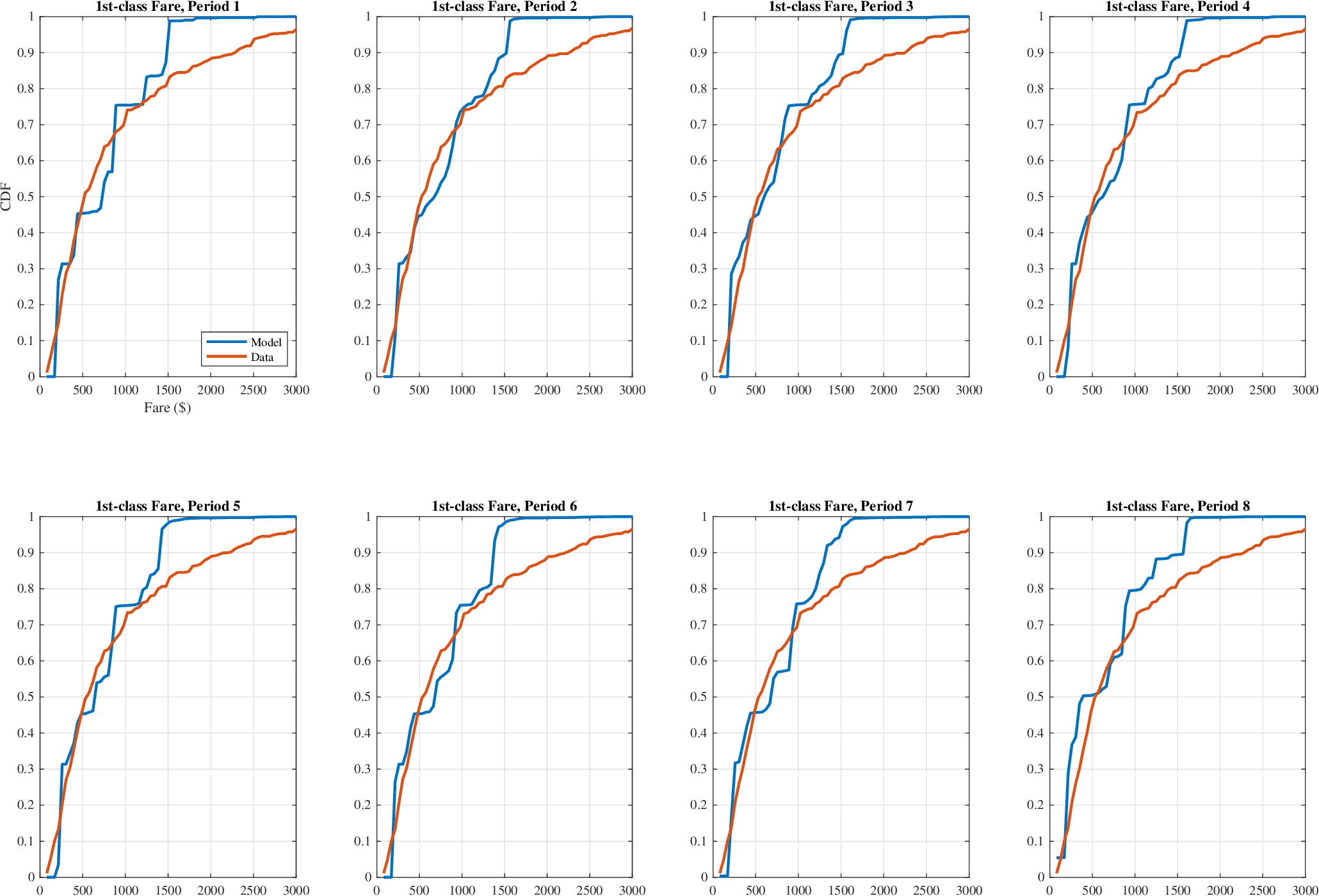}
  \end{center}
  \begin{figurenotes} This figure displays the distributions of first-class fares according to the model (blue) for the smaller modal capacity of 160 economy seats and 18 first-class seats and the distributions of first-class fares in the data (red) by time. \end{figurenotes}
\end{figure}

\begin{figure}[h!]
  \caption{Model Fit: Share of Business Passengers}
  \begin{center}
  \includegraphics[scale=0.3]{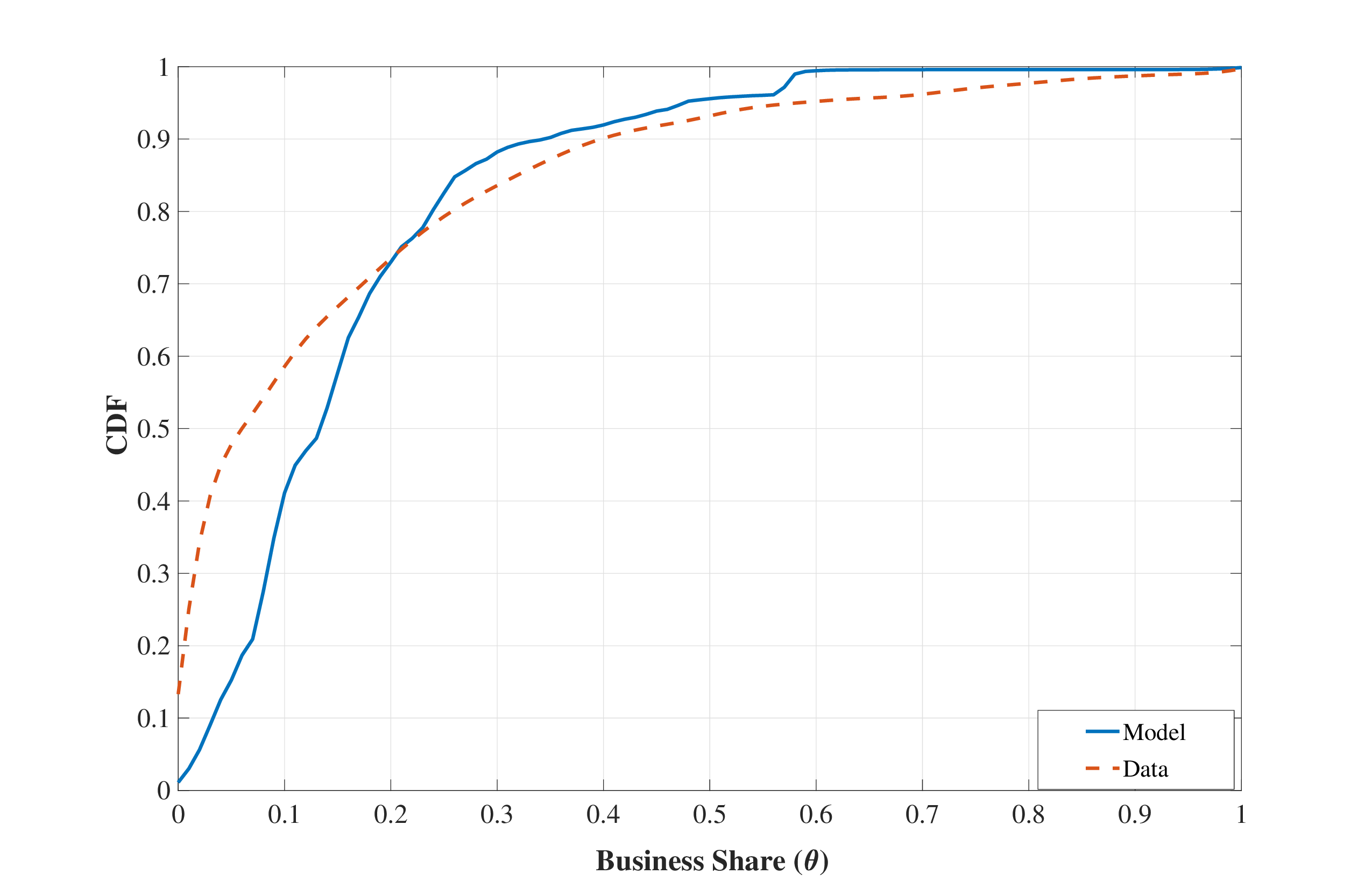}
  \label{fig:fit3}
  \end{center}
  \begin{figurenotes}This figure displays the distributions of the share of business passengers across all periods, according to the model (blue) for the smaller modal capacity of 160 economy seats and 18 first-class seats and the distributions of first-class fares in the data (red). \end{figurenotes}
\end{figure}

\section{Dynamic Programming}
{\subsection{Simulation Results \label{section:dp}}
In this section, we briefly present simulated optimal seat release policies when the Poisson arrival rates can be either high or low. We hold the initial capacity fixed for this exercise and choose parameters with a high willingness-to-pay. 
In particular, for illustration we consider two cases (out of 10,000 we consider) where in the first one, the Poisson arrival rates are $\lambda_t^n = 42.4 -0.086\times (t-1)$ and $\lambda_t^c = 45.33 -0.086\times (t-1)$ and in the other they are $\lambda_t^n = 9.084 +0.075\times (t-1)$ and $\lambda_t^c = 21.67 +0.075\times (t-1)$.

Given that in our model, an airline can update its prices once every period, the airline wants to ``protect" itself by limiting the number of seats it releases if, for some reason, it chooses a price that is too low and experiences an unexpectedly high number of bookings. 
By this logic, we expect that in the ``high $\lambda$ market," it is optimal for the airlines to restrict the number of seats released in the early periods and then release a greater fraction as demand is realized over time. 

In Figure \ref{fig:dp} (top left), we display the boxplot of the gap between the unsold economy seats and the number of economy seats released for eight periods for 150 simulations. As expected, the gap decreases as we get closer to the flight date. In the last period, the airline released all unsold seats. In Figure \ref{fig:dp} (top right), we display the boxplot of the gap between the number of economy seats released and the realized demand. In other words, this subfigure shows whether the seat release policy binds depending on the demand parameters and varies across simulations. 

In contrast, when the airline expects the number of arrivals to be small, it releases most of its unsold seats early on, but the exact number varies across simulation, as shown in Figure \ref{fig:dp} (bottom left). In this case, however, the seat release policy tends to bind earlier than in the previous case. These two cases highlight how the airline optimally responds to the expected demand by appropriately choosing its seat release policies. 

Thus the model implied prices can distinguish between (i) a low number of arrivals in each period but almost everyone buys a ticket, and (ii) a high number of arrivals in each period, but only a few buy tickets, by relying on the observed prices. All else equal, the price will start ``low" and remain low in the first case. In contrast, in the second case, the price will start at a higher level and change depending on the share of purchasers. In other words, higher demand (relative to the capacity) manifests itself in a more significant increase in prices over time because it implies a more significant variation in the opportunity cost of a seat due to more substantial ``surprises'' in the number of sales at given prices.  }

\begin{figure}[ht!]
  \caption{Simulated Seat Release Policies for Economy Seats}
 \hspace{-0.45in} \includegraphics[scale=0.45]{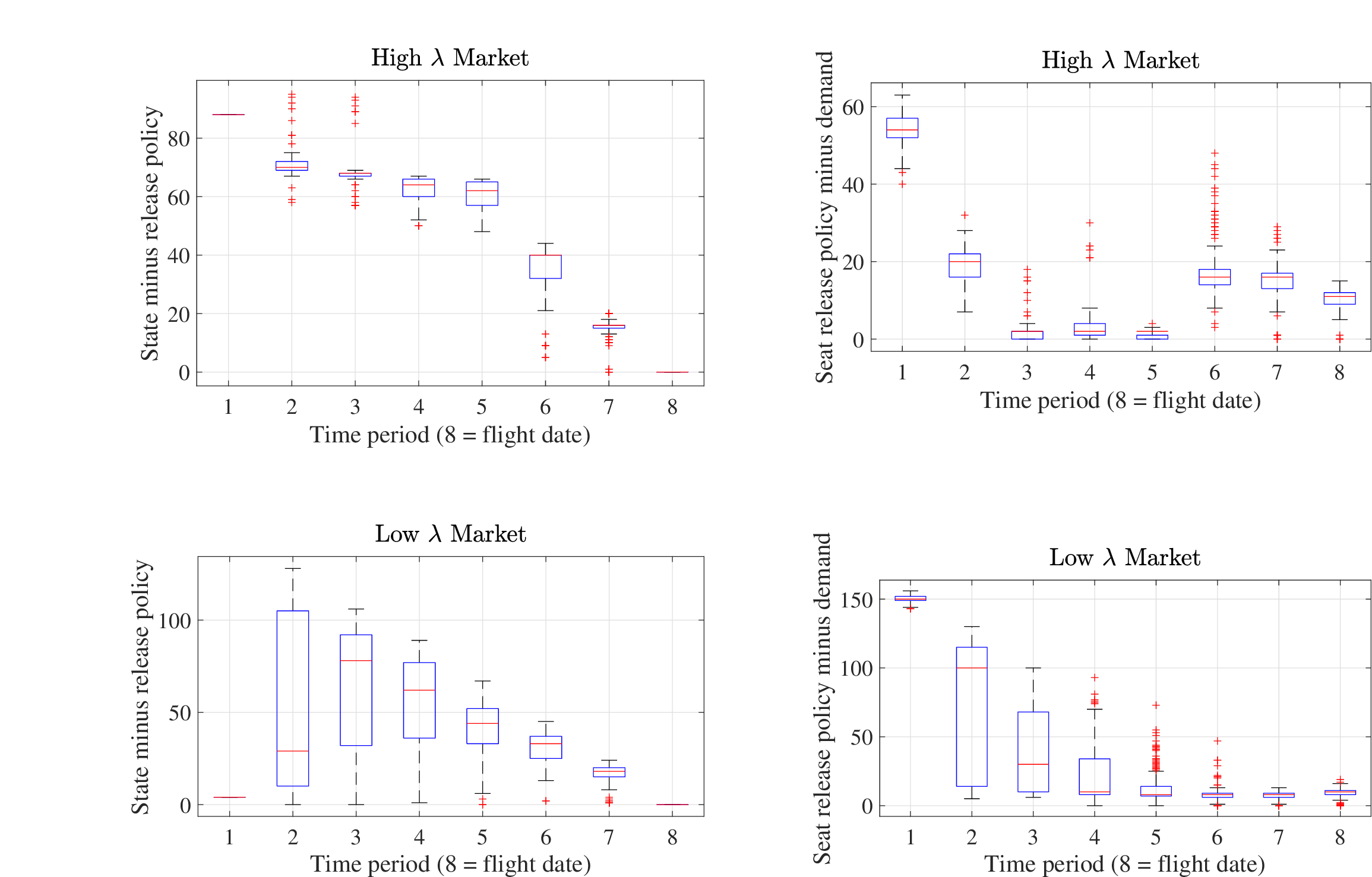}
  \label{fig:dp}
  \begin{figurenotes}The first column of the figure displays the boxplots (across 150 simulations) of the gap between unsold economy seats and the number of economy seats released at each of the eight periods. The top figure corresponds to ``high $\lambda$ market" with ${\Psi}=(922.54,	0.54,	0.22,	0.83,	0.566,	42.4,	-0.086,	0.079,	45.33)$ and the bottom figure corresponds to ``low $\lambda$ market" with ${\Psi}=(952.84,	0.063,	0.032,	0.702,	0.63,	9.084,	0.075,	0.037,	21.67)$. The second column of the figure displays the boxplots (across 150 simulations) of the gap between the number of economy seats released and the realized demand, at each of the eight periods.\end{figurenotes}
\end{figure}

\subsection{Uniqueness of the Optimal Policy\label{section:prooflemma}}
In this section, we show that the optimal policy is unique under some regularity conditions.
These regularity conditions are widely used in the literature and ensure that demand decreases in its own price and that the demand for each seating class is concave. We begin by presenting these conditions below, but for notational ease and without loss of generality, suppress the time index and ignore the connecting passengers.
\begin{assumption}
\begin{enumerate}
\item ({\bf Downward Demand}): $\left(\frac{\partial \E q^{e}(p^{e},p^{f})}{\partial p^{e}}\right) \leq0$, and $\left(\frac{\partial \E q^{f}(p^{e},p^{f})}{\partial p^{f}}\right) \leq0$.
\item ({\bf Concave Demand}): $\frac{\partial }{\partial p^{e}}\left(\frac{\partial \E q^{e}(p^{e},p^{f})}{\partial p^{e}}\right)<\frac{\partial }{\partial p^{f}}\left(\frac{\partial \E q^{e}(p^{e},p^{f})}{\partial p^{e}}\right)\leq0$, and\\ $\frac{\partial }{\partial p^{f}}\left(\frac{\partial \E q^{f}(p^{e},p^{f})}{\partial p^{f}}\right)<\frac{\partial }{\partial p^{e}}\left(\frac{\partial \E q^{f}(p^{e},p^{f})}{\partial p^{f}}\right)\leq0$.
\item ({\bf Cross Price Curvature}): $\frac{\partial }{\partial p^{f}}\left(\frac{\partial \E q^{e}(p^{e},p^{f})}{\partial p^{f}}\right)<0$, and $\frac{\partial }{\partial p^{e}}\left(\frac{\partial \E q^{f}(p^{e},p^{f})}{\partial p^{e}}\right)<0$.
\end{enumerate}
 \label{assumption:1}
\end{assumption}
The first assumption says that the demand for either of the cabins must be weakly decreasing in its own price.
The second assumption says that the demand is concave in its own price, which ensures the revenue is well defined. It also says that the demand for economy seats decreases more with respect to economy fare than business fare. The third assumption says that the change in demand for economy seats with respect to first-class price decreases with the first-class price and vice versa.
Although these assumptions are not on the primitives of the model, we present them in these forms because they are more intuitive, self-explanatory, and thus easier to understand than the equivalent assumptions on the primitives.
\begin{lemma}Under Assumption \ref{assumption:1} there is a unique policy function $\{\sigma_{t}:t=,1\ldots, T\}$.\label{lemma}
\end{lemma}

\begin{proof}
To prove this result, we use induction on $T$:
\begin{enumerate}
\item Suppose $T=1$ and $K^{e}$ and $K^{f}$ denote the cabin specific capacities. There are $S:=(K^{e}+1)\times(K^{f}+1)$ possible seats combinations (state-variables) that could be realized. We show that for each $s\in\{1, \ldots, S\}$ there is a unique optimal pair $\{p^{e}(s), p^{f}(s)\}$.

\item Suppose uniqueness is true for $T=\tilde{t}$, then we show that the uniqueness holds even when $T=\tilde{t}+1$.\footnote{ With time we also keep track of the remaining seats. We use the short-hand notation of $s_{t}$ for $s_{t}\in S$.
}
\end{enumerate}

{\bf Step 1.} Here, $T=1$ and
for notational ease, suppress the time index.
The airline solves:
\begin{eqnarray*}
V(\sigma^*)=\max_{p^{e}, p^{f}}\left\{\sum_{k=e,f}(p_{t}^{k}-c^{k})\int q^{k}(p^{e}, p^{f})g^{k}(q^{k}(p^{e}, p^{f}); p^{e}, p^{f})dq^{k}\right\}=\max_{p^{e}, p^{f}}\sum_{k=e,f}(p_{t}^{k}-c^{k})\E q^{k}(p^{e}, p^{f}).
\end{eqnarray*}
Then the equilibrium prices $(p^{e}, p^{f})$ solve the following system of equations:
\begin{eqnarray}
\left[\begin{array}{c}\E q^{e}(p^{e}, p^{f})+(p^{e}-c^{e})\frac{\partial {\E} q^{e}(p^{e},p^{f})}{\partial p^{e}}+(p^{f}-c^{f})\frac{\partial {\E} q^{f}(p^{e},p^{f})}{\partial p^{e}}\\{\E} q^{f}(p^{e}, p^{f})+(p^{f}-c^{f})\frac{\partial \E q^{f}(p^{e},p^{f})}{\partial p^{f}}+(p^{e}-c^{e})\frac{\partial {\E} q^{e}(p^{e},p^{f})}{\partial p^{f}}\end{array}\right]=\left[\begin{array}{c}0\\0\end{array}\right].
\end{eqnarray}
The above system has a unique solution $(p^{e}, p^{f})$ if the negative of the Jacobian corresponding to the above system is a $P$-matrix \citep{GaleNikaido1965}.
In other words, all principal minors of the Jacobian matrix are non-positive, which follows from Assumption \ref{assumption:1}.

{\bf Step 2.} Suppose we have a unique solution when $T=\tilde{t}$ and all finite pair $\{K^{e}, K^{f}\}$. Next, we want to show that the solution is still unique if we have one additional period, i.e., $T=\tilde{t}+1$.
Consider the value function.

\begin{eqnarray*}
V(\sigma^{*\tilde{t}}):=\max_{\{\sigma_{t}\}_{t=0}^{\tilde{t}}}\mathbb{E}_0\Bigg[\sum_{t=0}^{\tilde{t}}\sum_{k=e,f}(p_{t}^{k}-c^{k}) q_{t}^{k}(\sigma_{t})\big{|}K_0^f, K_0^e\Bigg],
\end{eqnarray*}
where $\sigma^{*\tilde{t}}:=(\sigma_1^{*},\ldots, \sigma_{\tilde{t}}^{*})$ is the unique optimal policy.
Now, suppose we have $\tilde{t}+1$ periods to consider. So the maximization problem faced by the airline becomes
\begin{eqnarray*}
&&\max_{\{\sigma_{t}\}_{t=0}^{\tilde{t}+1}}\mathbb{E}_0\Bigg[\sum_{t=0}^{\tilde{t}+1}\sum_{k=e,f}(p_{t}^{k}-c^{k}) q_{t}^{k}(\sigma_{t})\big{|}K_0^f, K_0^e\Bigg]=\max_{\{\sigma_{t}\}_{t=0}^{\tilde{t}}}\mathbb{E}_0\Bigg[\sum_{k=e,f}(p_{t}^{k}-c^{k}) q_{t}^{k}(\sigma_{t})\big{|}K_0^f, K_0^e\Bigg]\\&&+\max_{\sigma_{\tilde{t}+1}}\sum_{\omega_{\tilde{t}+1}\Omega_{\tilde{t}+1}}\mathbb{E}_{\tilde{t}+1}\Bigg[\sum_{k=e,f}(p_{t}^{k}-c^{k}) q_{t}^{k}(\sigma_{t})\big{|}\omega_{\tilde{t}+1}\Bigg]\Pr(\omega_{\tilde{t}+1}|\sigma^{\tilde{t}}).
\end{eqnarray*}
Consider the last period.
We have shown that for any realized state space $\omega_{\tilde{t}+1}$, there is a unique optimal policy that solves the second term.
The question is if the uniqueness is preserved when we take an expectation with respect to the state variable $\omega_{\tilde t+1}$.

For the solution to be unique, it is sufficient that the transition probability $\Pr(\omega_{\tilde t+1}|\sigma^{t})$ is log-concave,  which guarantees the expected profit is quasi-concave, hence the solution is unique.\footnote{ A positive and discrete random variable $V$ is log-concave if its p.m.f. $\Pr(V=i)$ forms a log-concave sequence. A non-negative sequence $\{r_{i}: i\geq0\}$ is log-concave if for all $i\geq1: (r_{i})^{2}\geq (r_{i-1})\times(r_{i+1})$.}
Then, the fact that the uniqueness extends from $\tilde{t}$ to $\tilde{t}+1$ follows from the usual backward induction argument of the finite-periods maximization problem.
Therefore it is enough to show that the transition probability is a (generalized) Poisson distribution, which is log-concave \cite[see][]{Johnson2007}.

For simplicity and to provide some intuition as to why transition probability is a (generalized) Poisson, we present the derivation of the transition probability when there is only one cabin and without censoring. Extending the argument to two cabins and incorporating rationing is straightforward, albeit tedious, once we recognize that the Poisson structure is preserved under truncation.
Suppose there is only one cabin and no seat-release  policy and hence no censoring.
And let $\tilde{K}_{t}=m$ is the number of seats remaining at time $t$.
Then, the probability of reaching $\tilde{K}_{t+1}=m'$ in $t+1$ from $m$ in period $t$ is
\begin{eqnarray*}
&&\Pr(m'|m, \text{Price}=p)=\Pr(Sales_{t} = \underbrace{m-m'}_{=d}|p)=\sum_{s=0}^{\infty}\Pr(Sales_{t}=d, S_{t}=s|p)\notag\\
&=&\sum_{s=0}^{\infty}\Pr(Sales_{t}=d|S_{t}=s, p)\times \Pr(S_{t}=s|p)=\sum_{s=d}^{\infty}{s \choose d} (1-\tilde{F}_t(p))^{d}(\tilde{F}_t(p))^{s-d}\frac{e^{-\lambda_{t}^n}(\lambda_{t}^{n})^s}{s!}\notag\\
&=&e^{-\lambda_{t}^n}(1-\tilde{F}_t(p))^{d}(\lambda_{t}^n)^{d}\sum_{s=d}^{\infty}{s \choose d} (\tilde{F}_t(p))^{s-d}\frac{(\lambda_{t}^n)^{(s-d)}}{s!}=e^{-(\lambda_{t}^n)^s}\frac{[\lambda_{t}^n(1-\tilde{F}_t(p))]^{d}}{d!}\sum_{s=d}^{\infty}\frac{[\lambda_{t}^n\tilde{F}_t(p)]^{s-d}}{(s-d)!} 
\end{eqnarray*}
\begin{eqnarray*}
&&\Pr(m'|m, \text{Price}=p)=e^{-\lambda_{t}^n}\frac{[\lambda_{t}^n(1-\tilde{F}_t(p))]^{d}}{d!}e^{\lambda_{t}\tilde{F}_t(p)}=e^{-\lambda_{t}^n(1-\tilde{F}_t(p))}\frac{[\lambda_{t}^n(1-\tilde{F}_t(p))]^{d}}{d!}\notag\\ &=&e^{-\lambda_{t}^n(1-\tilde{F}_t(p))}\frac{[\lambda_{t}^n(1-\tilde{F}_t(p))]^{(m-m')}}{(m-m')!},\label{eq:transition1}
\end{eqnarray*}
where $\lambda_t^n$ is the rate of nonstop passengers,   
$
\tilde{F}(t)=\left\{\begin{array}{cc}
\int_0^tf_t(v)dv,& \texttt{ economy cabin}\\
\int_0^t\int_0^{\infty} f_t(v)f_{\xi}(w/v)\frac{1}{v}dvdw, &\texttt{first-class cabin}
\end{array}\right.\label{eq:Ftilde}
$
and in the fourth equality the sum starts from $s=d$ as the probability of $d$ sales when $s<d$ is zero. Therefore, the transition probability is a Poisson with parameter $\lambda_{t}(1-\tilde{F}_t(p))$.

\end{proof}

\end{document}